\documentclass[11pt]{llncs}

\usepackage{latexsym}

\usepackage{fullpage}

\usepackage[backend=bibtex, style=alphabetic]{biblatex}
\bibliography{timed}
\usepackage{amssymb,amsmath,amsfonts,bm}
\usepackage{subfigure} 
\usepackage{epic} 
\usepackage{eepic} 
\usepackage{enumerate}

\usepackage{amsthm}
\usepackage{vinayak}
\usepackage{graphicx}
\usepackage{verbatim}
\usepackage[mathscr]{eucal}
\usepackage{xspace}
\usepackage[boxed, linesnumbered,nokwfunc]{algorithm2e}
 \SetAlgoInsideSkip{smallskip}
\SetAlCapSkip{9pt}
 \SetAlFnt{\normalsize\rm}
\usepackage{enumerate}
\usepackage{paralist}
\usepackage[usenames,dvipsnames]{pstricks}
 \usepackage{epsfig}
 \usepackage{pst-grad} 
 \usepackage{pst-plot} 

\usepackage{aurical}
\usepackage[T1]{fontenc}

\title{Quantitative  
 Timed Simulation Functions and Refinement Metrics for Timed Systems
\thanks{%
 This work has been financially supported in part by the
European Commission FP7-ICT Cognitive Systems, Interaction, and Robotics under the contract  \# 270180 (NOPTILUS);  
by Funda\c c\~ ao para Ci\^ encia e Tecnologia under project 
PTDC/EEA-CRO/104901/2008 (Modeling and control of Networked vehicle systems in persistent autonomous
operations);
by Austrian Science Fund (FWF) Grant No P 23499-N23 on Modern
Graph Algorithmic Techniques in Formal Verification; FWF NFN Grant No
S11407-N23 (RiSE); ERC Start grant (279307: Graph Games); and the Microsoft faculty
fellows award
}
 }
\author{Krishnendu Chatterjee$^1$    
\and Vinayak S. Prabhu$^2$ }
\institute{
$^1$ IST Austria (Institute of Science and Technology Austria) $\qquad$
$^2$Faculty of Engineering, University of Porto \\
{\tt Krishnendu.Chatterjee@ist.ac.at, vinayak@fe.up.pt}
}
\date{}

\begin{document}
\maketitle

\thispagestyle{plain}
\pagestyle{plain}

\begin{abstract}

We introduce quantatitive timed refinement  and timed simulation (directed) metrics,  incorporating  zenoness check
s,
 for timed systems.
These metrics assign  positive real numbers between zero and infinity which quantify the
\emph{timing mismatches}  between two timed systems, amongst non-zeno runs.
We quantify timing mismatches in three ways:
(1)~the maximal timing mismatch that can arise,
(2)~the ``steady-state'' maximal timing mismatches, where initial transient timing mismatches are ignored;
and
(3)~the (long-run) average timing mismatches amongst two systems.
These three kinds of mismatches constitute three important types of timing differences.
Our event times are the \emph{global times}, measured from the start of the system
execution, not just the  time durations of individual steps.
We present algorithms over timed automata for computing the
 three quantitative simulation distances to
within any desired degree of accuracy.
In order to compute the values of the quantitative simulation distances, we
use a game theoretic formulation.
We introduce  two new kinds of objectives for two player games on finite-state game graphs:
(1)~eventual debit-sum level objectives, and
(2)~average debit-sum level objectives.
We present algorithms for computing the optimal values for these objectives in graph games,
and then use these algorithms to compute the values of the  timed simulation
distances over timed automata.

\end{abstract}

\section{Introduction}

Theories of system approximation for continuous systems
 are used for analyzing systems that differ to a 
small extent, as opposed to the traditional boolean yes/no view of system
refinement for discrete systems.
These theories are necessary as formal models are only approximations
of the real world, and are subject to estimation and modelling errors.
Approximation  theories have been traditionally developed for continuous control 
systems~\cite{Antoulas01asurvey}
and more recently for  linear and non-linear systems
 \cite{Pola2010365, GirardSwitch, GLPappas08},
  timed systems~\cite{HMP05}, 
labeled Markov Processes~\cite{DesharnaisGJP04},
 probabilistic automata~\cite{BreugelMOW03},
quantitative transition systems~\cite{AlfaroFS09}, games~\cite{CARV10},
and software systems~\cite{ChaudhuriGL12}.

Timed and hybrid systems model  the evolution of system outputs as well
as the timing aspects related to the system evolution.
In this work we develop a theory of system approximation for timed systems by quantifying
the \emph{timing differences} between corresponding system events.
We first generalize timed refinement relations to metrics on timed systems that quantitatively
estimate the closeness of two systems.
Given a timed model $T_s$ denoting the abstract specification model, and a model $T_r$ 
denoting the concrete refined  implementation of $T_s$, we assign a positive real number between
zero and infinity to the pair $(T_r,T_s)$ which denotes the quantitative refinement distance
between $T_r$ and $T_s$.
Given a trace $\trace_r$ of $T_r$, and  a trace  $\trace_s$ of $T_s$, we define various 
distances between the two traces, e.g., the distance being $\infty$ if the untimed
trace sequences differ, and being the supremum of the differences of the matching
timepoints for matching events otherwise.
Our event times are the \emph{global times}, measured from the start of the system
execution, not just the time durations of individual steps.
The distance between the systems  $T_r$ and  $T_s$ is taken to be the supremum of
closest matching trace differences from the initial states.

Timed trace inclusion is undecidable on timed automata~\cite{AlurD94}, thus timed refinement is
conservatively estimated using  \emph{timed simulation relations}~\cite{Cerans92}.
Simulation relations take a branching time 
view, unlike the linear view of refinement relations, and
can be defined using two player \emph{games}.
We generalize timed simulation relations to quantitative timed simulation functions, and
define the values of quantitative timed simulation functions as the real-valued
outcome of games played on the corresponding timed graphs.

\emph{Zeno} runs where time converges is an artifact present in models of timed
systems due to model imperfections; such runs are obviously absent in the physical
systems which our timed models are meant to represent.
We thus exclude Zeno runs in our computation of quantitative timed refinement and
quantitative timed simulation relations.

We define three illustrative quantitative timed simulation directed distances which measure
three important system differences.
The \emph{maximal time difference} quantitative simulation distance denotes the maximal
time discrepancy that can arise amongst matching transitions.
The \emph{eventual maximal time difference} quantitative simulation distance denotes
 the eventual maximal time discrepancy that arises (ignoring finite time trace 
prefix discrepancies)  amongst matching transitions.
This corresponds to the ``steady-state'' difference between systems, ignoring transient
behavior.
The \emph{(long-run) average time difference} quantitative simulation distance denotes
the average time discrepancy amongst matching transitions.
This distance measures the long-run 
average time discrepancies, per transition, amongst two
timed systems.
Ideally, we want all three simulation distances to be as small as possible between the
specification and the implementation systems, but
minimizing one may lead to increase in values for others.
Thus, all three simulation distances give important information about systems.
We illustrate  the various quantitative timed simulation distances  via examples.

\begin{example}[Maximal Time Difference]
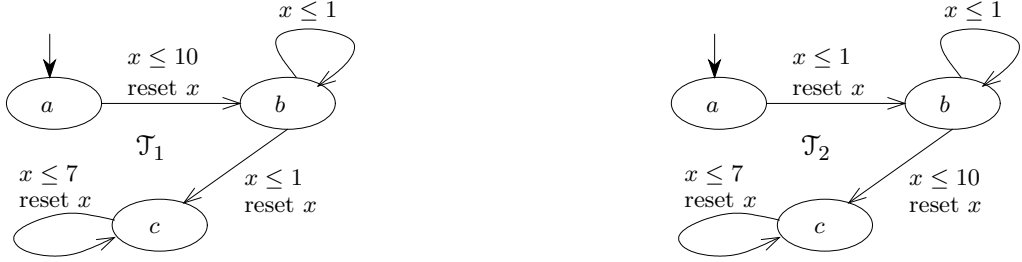
\begin{figure}[t]
\vspace*{-1em}
 \begin{minipage}[t]{0.6\linewidth}
   \hspace*{-10mm}
    \strut\centerline{\setlength{\unitlength}{0.00049869in}
\begingroup\makeatletter\ifx\SetFigFont\undefined%
\gdef\SetFigFont#1#2#3#4#5{%
  \reset@font\fontsize{#1}{#2pt}%
  \fontfamily{#3}\fontseries{#4}\fontshape{#5}%
  \selectfont}%
\fi\endgroup%
{\renewcommand{\dashlinestretch}{30}
\begin{picture}(3640,2724)(0,-10)
\put(503,1602){\ellipse{990}{540}}
\put(2933,1602){\ellipse{990}{540}}
\put(1605,335){\ellipse{990}{540}}
\path(998,1602)(2438,1602)
\path(2258.000,1542.000)(2438.000,1602.000)(2258.000,1662.000)
\path(2933,1332)(1853,567)
\path(1965.203,720.005)(1853.000,567.000)(2034.565,622.082)
\path(458,2322)(458,1872)
\blacken\path(398.000,2052.000)(458.000,1872.000)(518.000,2052.000)(458.000,1998.000)(398.000,2052.000)
\path(1133,387)(1130,388)(1124,389)
	(1113,392)(1096,396)(1073,401)
	(1046,407)(1015,414)(982,421)
	(947,428)(912,434)(878,440)
	(845,446)(814,450)(785,454)
	(757,456)(731,458)(705,459)
	(681,459)(656,459)(632,457)
	(608,454)(583,451)(558,447)
	(533,443)(507,437)(480,431)
	(453,424)(426,416)(398,408)
	(371,399)(344,389)(317,379)
	(291,368)(267,357)(243,346)
	(221,335)(200,324)(181,313)
	(164,302)(149,291)(135,280)
	(123,270)(113,259)(105,250)
	(99,241)(93,231)(89,222)
	(86,212)(85,202)(85,193)
	(86,183)(88,173)(92,164)
	(97,154)(103,145)(109,135)
	(117,127)(126,118)(136,110)
	(146,102)(157,94)(169,87)
	(181,80)(193,74)(206,68)
	(219,62)(233,57)(250,51)
	(267,45)(286,40)(305,35)
	(326,31)(347,27)(370,23)
	(393,20)(416,17)(441,15)
	(465,13)(490,13)(515,12)
	(539,13)(563,14)(586,15)
	(609,17)(632,20)(654,23)
	(676,27)(695,31)(715,36)
	(735,41)(756,47)(777,54)
	(800,62)(824,71)(849,81)
	(877,92)(905,104)(936,117)
	(966,130)(997,144)(1027,158)
	(1055,170)(1079,182)(1099,191)(1133,207)
\path(995.680,76.067)(1133.000,207.000)(944.585,184.646)
\path(3023,1872)(3020,1875)(3015,1882)
	(3006,1893)(2994,1909)(2979,1928)
	(2963,1949)(2947,1970)(2932,1990)
	(2918,2009)(2905,2027)(2894,2043)
	(2883,2057)(2874,2071)(2866,2084)
	(2858,2097)(2851,2108)(2845,2120)
	(2838,2132)(2833,2144)(2827,2156)
	(2822,2169)(2818,2181)(2814,2194)
	(2812,2206)(2810,2219)(2809,2231)
	(2809,2242)(2811,2253)(2813,2264)
	(2817,2273)(2822,2283)(2828,2291)
	(2836,2300)(2844,2307)(2854,2315)
	(2866,2322)(2880,2329)(2895,2336)
	(2912,2342)(2930,2348)(2950,2354)
	(2971,2359)(2993,2363)(3017,2367)
	(3040,2370)(3065,2372)(3089,2374)
	(3114,2375)(3138,2376)(3163,2375)
	(3188,2375)(3211,2373)(3234,2371)
	(3258,2369)(3283,2366)(3308,2362)
	(3334,2358)(3360,2353)(3385,2348)
	(3411,2342)(3436,2335)(3460,2328)
	(3483,2320)(3505,2312)(3525,2304)
	(3544,2295)(3560,2286)(3575,2276)
	(3588,2267)(3599,2257)(3608,2247)
	(3615,2236)(3621,2225)(3625,2214)
	(3628,2202)(3628,2189)(3628,2176)
	(3626,2162)(3622,2148)(3617,2133)
	(3611,2118)(3603,2104)(3595,2089)
	(3586,2074)(3576,2060)(3566,2046)
	(3555,2033)(3544,2020)(3533,2008)
	(3522,1996)(3511,1984)(3498,1972)
	(3485,1960)(3471,1948)(3457,1935)
	(3441,1923)(3424,1909)(3406,1895)
	(3386,1880)(3364,1864)(3342,1848)
	(3320,1833)(3299,1818)(3281,1805)(3248,1782)
\path(3361.364,1934.147)(3248.000,1782.000)(3429.979,1835.699)
\put(368,1512){\makebox(0,0)[lb]{\smash{{\SetFigFont{9}{10.8}{\familydefault}{\mddefault}{\updefault}$a$}}}}
\put(2798,1512){\makebox(0,0)[lb]{\smash{{\SetFigFont{9}{10.8}{\familydefault}{\mddefault}{\updefault}$b$}}}}
\put(1493,252){\makebox(0,0)[lb]{\smash{{\SetFigFont{9}{10.8}{\familydefault}{\mddefault}{\updefault}$c$}}}}
\put(1358,1062){\makebox(0,0)[lb]{\smash{{\SetFigFont{11}{13.2}{\familydefault}{\mddefault}{\updefault}$\A_1$}}}}
\put(1268,1692){\makebox(0,0)[lb]{\smash{{\SetFigFont{9}{10.8}{\familydefault}{\mddefault}{\updefault}$\text{reset } x$}}}}
\put(2483,747){\makebox(0,0)[lb]{\smash{{\SetFigFont{9}{10.8}{\familydefault}{\mddefault}{\updefault}$x\leq 1$}}}}
\put(1268,2007){\makebox(0,0)[lb]{\smash{{\SetFigFont{9}{10.8}{\familydefault}{\mddefault}{\updefault}$x\leq 10$}}}}
\put(2483,477){\makebox(0,0)[lb]{\smash{{\SetFigFont{9}{10.8}{\familydefault}{\mddefault}{\updefault}$\text{reset } x$}}}}
\put(143,792){\makebox(0,0)[lb]{\smash{{\SetFigFont{9}{10.8}{\familydefault}{\mddefault}{\updefault}$x\leq 7$}}}}
\put(143,522){\makebox(0,0)[lb]{\smash{{\SetFigFont{9}{10.8}{\familydefault}{\mddefault}{\updefault}$\text{reset } x$}}}}
\put(2843,2502){\makebox(0,0)[lb]{\smash{{\SetFigFont{9}{10.8}{\familydefault}{\mddefault}{\updefault}$x\leq 1$}}}}
\end{picture}
}}
    \end{minipage}
  \begin{minipage}[t]{0.35\linewidth}
    \strut\centerline{\setlength{\unitlength}{0.00049869in}
\begingroup\makeatletter\ifx\SetFigFont\undefined%
\gdef\SetFigFont#1#2#3#4#5{%
  \reset@font\fontsize{#1}{#2pt}%
  \fontfamily{#3}\fontseries{#4}\fontshape{#5}%
  \selectfont}%
\fi\endgroup%
{\renewcommand{\dashlinestretch}{30}
\begin{picture}(3640,2724)(0,-10)
\put(503,1602){\ellipse{990}{540}}
\put(2933,1602){\ellipse{990}{540}}
\put(1605,335){\ellipse{990}{540}}
\path(998,1602)(2438,1602)
\path(2258.000,1542.000)(2438.000,1602.000)(2258.000,1662.000)
\path(2933,1332)(1853,567)
\path(1965.203,720.005)(1853.000,567.000)(2034.565,622.082)
\path(458,2322)(458,1872)
\blacken\path(398.000,2052.000)(458.000,1872.000)(518.000,2052.000)(458.000,1998.000)(398.000,2052.000)
\path(1133,387)(1130,388)(1124,389)
	(1113,392)(1096,396)(1073,401)
	(1046,407)(1015,414)(982,421)
	(947,428)(912,434)(878,440)
	(845,446)(814,450)(785,454)
	(757,456)(731,458)(705,459)
	(681,459)(656,459)(632,457)
	(608,454)(583,451)(558,447)
	(533,443)(507,437)(480,431)
	(453,424)(426,416)(398,408)
	(371,399)(344,389)(317,379)
	(291,368)(267,357)(243,346)
	(221,335)(200,324)(181,313)
	(164,302)(149,291)(135,280)
	(123,270)(113,259)(105,250)
	(99,241)(93,231)(89,222)
	(86,212)(85,202)(85,193)
	(86,183)(88,173)(92,164)
	(97,154)(103,145)(109,135)
	(117,127)(126,118)(136,110)
	(146,102)(157,94)(169,87)
	(181,80)(193,74)(206,68)
	(219,62)(233,57)(250,51)
	(267,45)(286,40)(305,35)
	(326,31)(347,27)(370,23)
	(393,20)(416,17)(441,15)
	(465,13)(490,13)(515,12)
	(539,13)(563,14)(586,15)
	(609,17)(632,20)(654,23)
	(676,27)(695,31)(715,36)
	(735,41)(756,47)(777,54)
	(800,62)(824,71)(849,81)
	(877,92)(905,104)(936,117)
	(966,130)(997,144)(1027,158)
	(1055,170)(1079,182)(1099,191)(1133,207)
\path(995.680,76.067)(1133.000,207.000)(944.585,184.646)
\path(3023,1872)(3020,1875)(3015,1882)
	(3006,1893)(2994,1909)(2979,1928)
	(2963,1949)(2947,1970)(2932,1990)
	(2918,2009)(2905,2027)(2894,2043)
	(2883,2057)(2874,2071)(2866,2084)
	(2858,2097)(2851,2108)(2845,2120)
	(2838,2132)(2833,2144)(2827,2156)
	(2822,2169)(2818,2181)(2814,2194)
	(2812,2206)(2810,2219)(2809,2231)
	(2809,2242)(2811,2253)(2813,2264)
	(2817,2273)(2822,2283)(2828,2291)
	(2836,2300)(2844,2307)(2854,2315)
	(2866,2322)(2880,2329)(2895,2336)
	(2912,2342)(2930,2348)(2950,2354)
	(2971,2359)(2993,2363)(3017,2367)
	(3040,2370)(3065,2372)(3089,2374)
	(3114,2375)(3138,2376)(3163,2375)
	(3188,2375)(3211,2373)(3234,2371)
	(3258,2369)(3283,2366)(3308,2362)
	(3334,2358)(3360,2353)(3385,2348)
	(3411,2342)(3436,2335)(3460,2328)
	(3483,2320)(3505,2312)(3525,2304)
	(3544,2295)(3560,2286)(3575,2276)
	(3588,2267)(3599,2257)(3608,2247)
	(3615,2236)(3621,2225)(3625,2214)
	(3628,2202)(3628,2189)(3628,2176)
	(3626,2162)(3622,2148)(3617,2133)
	(3611,2118)(3603,2104)(3595,2089)
	(3586,2074)(3576,2060)(3566,2046)
	(3555,2033)(3544,2020)(3533,2008)
	(3522,1996)(3511,1984)(3498,1972)
	(3485,1960)(3471,1948)(3457,1935)
	(3441,1923)(3424,1909)(3406,1895)
	(3386,1880)(3364,1864)(3342,1848)
	(3320,1833)(3299,1818)(3281,1805)(3248,1782)
\path(3361.364,1934.147)(3248.000,1782.000)(3429.979,1835.699)
\put(368,1512){\makebox(0,0)[lb]{\smash{{\SetFigFont{9}{10.8}{\familydefault}{\mddefault}{\updefault}$a$}}}}
\put(2798,1512){\makebox(0,0)[lb]{\smash{{\SetFigFont{9}{10.8}{\familydefault}{\mddefault}{\updefault}$b$}}}}
\put(1493,252){\makebox(0,0)[lb]{\smash{{\SetFigFont{9}{10.8}{\familydefault}{\mddefault}{\updefault}$c$}}}}
\put(1268,1692){\makebox(0,0)[lb]{\smash{{\SetFigFont{9}{10.8}{\familydefault}{\mddefault}{\updefault}$\text{reset } x$}}}}
\put(2483,477){\makebox(0,0)[lb]{\smash{{\SetFigFont{9}{10.8}{\familydefault}{\mddefault}{\updefault}$\text{reset } x$}}}}
\put(143,792){\makebox(0,0)[lb]{\smash{{\SetFigFont{9}{10.8}{\familydefault}{\mddefault}{\updefault}$x\leq 7$}}}}
\put(143,522){\makebox(0,0)[lb]{\smash{{\SetFigFont{9}{10.8}{\familydefault}{\mddefault}{\updefault}$\text{reset } x$}}}}
\put(2843,2502){\makebox(0,0)[lb]{\smash{{\SetFigFont{9}{10.8}{\familydefault}{\mddefault}{\updefault}$x\leq 1$}}}}
\put(1268,2007){\makebox(0,0)[lb]{\smash{{\SetFigFont{9}{10.8}{\familydefault}{\mddefault}{\updefault}$x\leq 1$}}}}
\put(2483,747){\makebox(0,0)[lb]{\smash{{\SetFigFont{9}{10.8}{\familydefault}{\mddefault}{\updefault}$x\leq 10$}}}}
\put(1358,1062){\makebox(0,0)[lb]{\smash{{\SetFigFont{11}{13.2}{\familydefault}{\mddefault}{\updefault}$\A_2$}}}}
\end{picture}
}}
  
  \end{minipage}
   \caption{Two timed automata $\A_1, \A_2$.
  }
  \label{figure:ExampleOne} 
\vspace{-1em}
\end{figure}
Consider the two timed automata $\A_1$ and $\A_2$ in Figure~\ref{figure:ExampleOne}.
The locations are labelled with the observations.
The starting location of each automaton is the one labelled with the observation $a$, 
and the starting value of the clock $x$ is 0. 
Let us first look at the value of the \emph{maximal time difference} 
 quantitative timed simulation distance 
$\simfunc_{\maxdiff}$ for the
state pair $\left(\tuple{a,x=0}^{\A_1}, \tuple{a,x=0}^{\A_2}\right)$.
The value is 
(1)~infinity
if the state of $\A_1$  does not \emph{time-abstract simulate} (combined with time-divergence encoded as a fairness constraint~\cite{Cerans92, FairSimulation}, to allow only time-divergent runs in $\A_1$) the state of $\A_2$; 
(2)~the maximal time difference between matching transitions of $\A_1$ and
 $\A_2$ otherwise, amongst time-divergent runs.
For the two timed automata in  Figure~\ref{figure:ExampleOne}, it can be checked
that $\tuple{a,x=0}^{\A_1}$ time-abstract simulates $ \tuple{a,x=0}^{\A_2}$, and that
the maximal time difference between matching transitions is
$9$ time units,
(\emph{e.g.} between the paths 
$\tuple{a,x=0}^{\A_1} \stackrel{10}{\longrightarrow} \tuple{b,x=0}^{\A_1} 
\stackrel{0}{\longrightarrow} \tuple{c,x=0}^{\A_1} \stackrel{5}{\longrightarrow} 
\tuple{c,x=0}^{\A_1}\stackrel{5}
 {\longrightarrow} \cdots$ and
$\tuple{a,x=0}^{\A_2} \stackrel{1}{\longrightarrow} \tuple{b,x=0}^{\A_2} 
\stackrel{9}{\longrightarrow} \tuple{c,x=0}^{\A_2} \stackrel{5}{\longrightarrow} 
\tuple{c,x=0}^{\A_2}\stackrel{5}
{ \longrightarrow} \cdots$).
\qed
\end{example}

\begin{example}[Global Event Times]
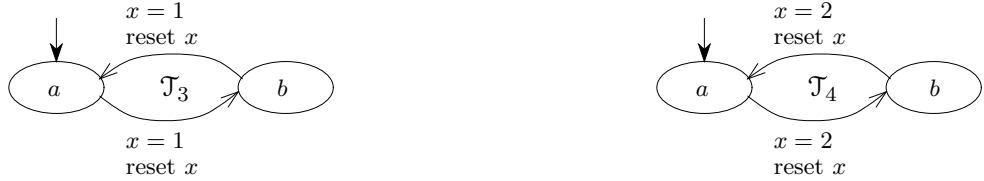
\begin{figure}[h]
\vspace*{-1em}
 \hspace*{-9mm}
 \begin{minipage}[t]{0.6\linewidth}
 \hspace*{-2mm}
      \strut\centerline{\setlength{\unitlength}{0.00049869in}
\begingroup\makeatletter\ifx\SetFigFont\undefined%
\gdef\SetFigFont#1#2#3#4#5{%
  \reset@font\fontsize{#1}{#2pt}%
  \fontfamily{#3}\fontseries{#4}\fontshape{#5}%
  \selectfont}%
\fi\endgroup%
{\renewcommand{\dashlinestretch}{30}
\begin{picture}(3391,1893)(0,-10)
\put(503,987){\ellipse{990}{540}}
\put(2888,987){\ellipse{990}{540}}
\path(503,1707)(503,1257)
\blacken\path(443.000,1437.000)(503.000,1257.000)(563.000,1437.000)(503.000,1383.000)(443.000,1437.000)
\path(2438,1077)(2436,1079)(2431,1083)
	(2422,1091)(2410,1102)(2393,1116)
	(2373,1132)(2351,1150)(2327,1169)
	(2302,1187)(2278,1204)(2253,1221)
	(2230,1235)(2207,1249)(2184,1260)
	(2161,1271)(2138,1280)(2114,1288)
	(2089,1295)(2063,1302)(2042,1307)
	(2020,1311)(1997,1315)(1973,1318)
	(1948,1322)(1922,1325)(1895,1327)
	(1867,1329)(1838,1331)(1808,1333)
	(1778,1334)(1748,1334)(1717,1335)
	(1686,1334)(1656,1334)(1625,1333)
	(1596,1331)(1566,1329)(1538,1327)
	(1510,1325)(1484,1322)(1458,1318)
	(1434,1315)(1410,1311)(1387,1307)
	(1366,1302)(1338,1295)(1312,1288)
	(1287,1280)(1262,1271)(1238,1260)
	(1213,1249)(1188,1235)(1161,1221)
	(1135,1204)(1107,1187)(1079,1169)
	(1052,1150)(1027,1132)(1004,1116)
	(985,1102)(953,1077)
\path(1057.906,1235.097)(953.000,1077.000)(1131.783,1140.534)
\path(953,897)(955,895)(961,891)
	(971,883)(985,872)(1003,858)
	(1026,842)(1050,824)(1077,805)
	(1103,787)(1130,770)(1155,753)
	(1180,739)(1204,725)(1227,714)
	(1250,703)(1272,694)(1295,686)
	(1319,679)(1343,672)(1364,667)
	(1386,662)(1409,658)(1432,655)
	(1457,651)(1482,648)(1509,646)
	(1536,644)(1564,642)(1593,641)
	(1622,641)(1652,641)(1681,641)
	(1710,642)(1739,644)(1767,646)
	(1795,648)(1822,651)(1848,655)
	(1873,659)(1898,664)(1921,668)
	(1944,674)(1966,679)(1989,686)
	(2011,694)(2033,702)(2055,712)
	(2077,722)(2100,734)(2123,747)
	(2147,762)(2172,778)(2198,795)
	(2224,814)(2252,833)(2278,853)
	(2304,872)(2328,890)(2348,906)
	(2365,920)(2393,942)
\path(2288.532,783.613)(2393.000,942.000)(2214.393,877.971)
\put(1223,357){\makebox(0,0)[lb]{\smash{{\SetFigFont{9}{10.8}{\familydefault}{\mddefault}{\updefault}$x=1$}}}}
\put(1223,87){\makebox(0,0)[lb]{\smash{{\SetFigFont{9}{10.8}{\familydefault}{\mddefault}{\updefault}$\text{reset } x$}}}}
\put(1583,897){\makebox(0,0)[lb]{\smash{{\SetFigFont{11}{13.2}{\familydefault}{\mddefault}{\updefault}$\A_3$}}}}
\put(2798,897){\makebox(0,0)[lb]{\smash{{\SetFigFont{9}{10.8}{\familydefault}{\mddefault}{\updefault}$b$}}}}
\put(413,897){\makebox(0,0)[lb]{\smash{{\SetFigFont{9}{10.8}{\familydefault}{\mddefault}{\updefault}$a$}}}}
\put(1223,1437){\makebox(0,0)[lb]{\smash{{\SetFigFont{9}{10.8}{\familydefault}{\mddefault}{\updefault}$\text{reset } x$}}}}
\put(1223,1707){\makebox(0,0)[lb]{\smash{{\SetFigFont{9}{10.8}{\familydefault}{\mddefault}{\updefault}$x=1$}}}}
\end{picture}
}}
    \end{minipage}
  \begin{minipage}[t]{0.35\linewidth}
 \hspace*{5mm}
    \strut\centerline{\setlength{\unitlength}{0.00049869in}
\begingroup\makeatletter\ifx\SetFigFont\undefined%
\gdef\SetFigFont#1#2#3#4#5{%
  \reset@font\fontsize{#1}{#2pt}%
  \fontfamily{#3}\fontseries{#4}\fontshape{#5}%
  \selectfont}%
\fi\endgroup%
{\renewcommand{\dashlinestretch}{30}
\begin{picture}(3391,1893)(0,-10)
\put(503,987){\ellipse{990}{540}}
\put(2888,987){\ellipse{990}{540}}
\path(503,1707)(503,1257)
\blacken\path(443.000,1437.000)(503.000,1257.000)(563.000,1437.000)(503.000,1383.000)(443.000,1437.000)
\path(2438,1077)(2436,1079)(2431,1083)
	(2422,1091)(2410,1102)(2393,1116)
	(2373,1132)(2351,1150)(2327,1169)
	(2302,1187)(2278,1204)(2253,1221)
	(2230,1235)(2207,1249)(2184,1260)
	(2161,1271)(2138,1280)(2114,1288)
	(2089,1295)(2063,1302)(2042,1307)
	(2020,1311)(1997,1315)(1973,1318)
	(1948,1322)(1922,1325)(1895,1327)
	(1867,1329)(1838,1331)(1808,1333)
	(1778,1334)(1748,1334)(1717,1335)
	(1686,1334)(1656,1334)(1625,1333)
	(1596,1331)(1566,1329)(1538,1327)
	(1510,1325)(1484,1322)(1458,1318)
	(1434,1315)(1410,1311)(1387,1307)
	(1366,1302)(1338,1295)(1312,1288)
	(1287,1280)(1262,1271)(1238,1260)
	(1213,1249)(1188,1235)(1161,1221)
	(1135,1204)(1107,1187)(1079,1169)
	(1052,1150)(1027,1132)(1004,1116)
	(985,1102)(953,1077)
\path(1057.906,1235.097)(953.000,1077.000)(1131.783,1140.534)
\path(953,897)(955,895)(961,891)
	(971,883)(985,872)(1003,858)
	(1026,842)(1050,824)(1077,805)
	(1103,787)(1130,770)(1155,753)
	(1180,739)(1204,725)(1227,714)
	(1250,703)(1272,694)(1295,686)
	(1319,679)(1343,672)(1364,667)
	(1386,662)(1409,658)(1432,655)
	(1457,651)(1482,648)(1509,646)
	(1536,644)(1564,642)(1593,641)
	(1622,641)(1652,641)(1681,641)
	(1710,642)(1739,644)(1767,646)
	(1795,648)(1822,651)(1848,655)
	(1873,659)(1898,664)(1921,668)
	(1944,674)(1966,679)(1989,686)
	(2011,694)(2033,702)(2055,712)
	(2077,722)(2100,734)(2123,747)
	(2147,762)(2172,778)(2198,795)
	(2224,814)(2252,833)(2278,853)
	(2304,872)(2328,890)(2348,906)
	(2365,920)(2393,942)
\path(2288.532,783.613)(2393.000,942.000)(2214.393,877.971)
\put(1223,87){\makebox(0,0)[lb]{\smash{{\SetFigFont{9}{10.8}{\familydefault}{\mddefault}{\updefault}$\text{reset } x$}}}}
\put(1223,357){\makebox(0,0)[lb]{\smash{{\SetFigFont{9}{10.8}{\familydefault}{\mddefault}{\updefault}$x=2$}}}}
\put(1583,897){\makebox(0,0)[lb]{\smash{{\SetFigFont{11}{13.2}{\familydefault}{\mddefault}{\updefault}$\A_4$}}}}
\put(2843,897){\makebox(0,0)[lb]{\smash{{\SetFigFont{9}{10.8}{\familydefault}{\mddefault}{\updefault}$b$}}}}
\put(413,897){\makebox(0,0)[lb]{\smash{{\SetFigFont{9}{10.8}{\familydefault}{\mddefault}{\updefault}$a$}}}}
\put(1223,1437){\makebox(0,0)[lb]{\smash{{\SetFigFont{9}{10.8}{\familydefault}{\mddefault}{\updefault}$\text{reset } x$}}}}
\put(1223,1707){\makebox(0,0)[lb]{\smash{{\SetFigFont{9}{10.8}{\familydefault}{\mddefault}{\updefault}$x=2$}}}}
\end{picture}
}}
    \end{minipage}
   \caption{Two timed automata $\A_3, \A_4$.}
  \label{figure:ExampleThree} 
\end{figure}
Consider the two timed automata in Figure~\ref{figure:ExampleThree}.
 The value of the maximal time difference
 quantitative timed simulation distance
$\simfunc_{\maxdiff}$ for the
state pai  $\left(\tuple{a,x=0}^{\A_3}, \tuple{a,x=0}^{\A_4}\right)$
is $\infty$, since timing mismatch corresponding to the $n$-th transition is
$n$ (the $n$-th transition in $\A_3$ occurs at global time $n$, the 
$n$-th transition in $\A_4$ occurs at global time $2\cdot n$).
We depict the timelines in Figure~\ref{figure:exampleTimeline}.
\begin{figure}[h]
\vspace{-1em}
\scalebox{1} 
{
\begin{pspicture}(0,-2.168125)(8.329687,2.168125)
\definecolor{color681b}{rgb}{0.19215686274509805,0.0,1.0}
\definecolor{color683b}{rgb}{0.07450980392156863,0.49019607843137253,0.19215686274509805}
\definecolor{color685b}{rgb}{0.6392156862745098,0.023529411764705882,0.09803921568627451}
\definecolor{color687b}{rgb}{0.5843137254901961,0.08627450980392157,0.6901960784313725}
\definecolor{color689b}{rgb}{0.47058823529411764,0.29411764705882354,0.03137254901960784}
\definecolor{color691}{rgb}{0.7450980392156863,0.34901960784313724,0.12156862745098039}
\psline[linewidth=0.04cm,tbarsize=0.07055555cm 5.0,arrowsize=0.05291667cm 2.0,arrowlength=1.4,arrowinset=0.4]{|->}(0.0,1.5496875)(6.82,1.5496875)
\psline[linewidth=0.04cm,tbarsize=0.07055555cm 5.0,arrowsize=0.05291667cm 2.0,arrowlength=1.4,arrowinset=0.4]{|->}(0.0,0.3496875)(6.8,0.3496875)
\usefont{T1}{ptm}{m}{n}
\rput(6.6853123,1.0946875){$\A_3$ event timeline}
\psellipse[linewidth=0.04,dimen=outer,fillstyle=solid,fillcolor=color681b](0.58,1.5596875)(0.16,0.13)
\psellipse[linewidth=0.04,dimen=outer,fillstyle=solid,fillcolor=color683b](1.78,1.5596875)(0.16,0.13)
\psellipse[linewidth=0.04,dimen=outer,fillstyle=solid,fillcolor=color685b](2.38,1.5596875)(0.16,0.13)
\psellipse[linewidth=0.04,dimen=outer,fillstyle=solid,fillcolor=color687b](3.0,1.5596875)(0.16,0.13)
\psellipse[linewidth=0.04,dimen=outer,fillstyle=solid,fillcolor=color689b](2.4,0.3596875)(0.16,0.13)
\psline[linewidth=0.02cm,linecolor=color691,linestyle=dashed,dash=0.16cm 0.16cm](0.54,1.5496875)(1.18,0.3896875)
\psline[linewidth=0.02cm,linecolor=color691,linestyle=dashed,dash=0.16cm 0.16cm](1.22,1.5296875)(2.42,0.4096875)
\psline[linewidth=0.02cm,linecolor=color691,linestyle=dashed,dash=0.16cm 0.16cm](1.76,1.5896875)(3.62,0.3496875)
\psline[linewidth=0.02cm,linecolor=color691,linestyle=dashed,dash=0.16cm 0.16cm](2.42,1.5496875)(4.82,0.3696875)
\usefont{T1}{ptm}{m}{n}
\rput(6.7053127,-0.0853125){$\A_4$ event timeline}
\psellipse[linewidth=0.04,dimen=outer,fillstyle=solid,fillcolor=color681b](1.18,0.3596875)(0.16,0.13)
\psellipse[linewidth=0.04,dimen=outer,fillstyle=solid,fillcolor=color689b](1.2,1.5796875)(0.16,0.13)
\psellipse[linewidth=0.04,dimen=outer,fillstyle=solid,fillcolor=color683b](3.6,0.3796875)(0.16,0.13)
\psellipse[linewidth=0.04,dimen=outer,fillstyle=solid,fillcolor=color685b](4.78,0.3596875)(0.16,0.13)
\psline[linewidth=0.04cm,linecolor=color681b](1.6,-0.2503125)(2.2,-0.2503125)
\psline[linewidth=0.04cm,linecolor=color689b](1.6,-0.6503125)(2.76,-0.6503125)
\psline[linewidth=0.04cm,linecolor=color683b](1.58,-1.0503125)(3.4,-1.0503125)
\psline[linewidth=0.04cm,linecolor=color685b](1.6,-1.4303125)(4.0,-1.4303125)
\usefont{T1}{ptm}{m}{n}
\rput(2.929375,-1.9453125){Timing mismatch magnitude}
\psellipse[linewidth=0.04,dimen=outer,fillstyle=solid,fillcolor=color681b](1.0,-0.2603125)(0.16,0.13)
\psellipse[linewidth=0.04,dimen=outer,fillstyle=solid,fillcolor=color689b](1.0,-0.6603125)(0.16,0.13)
\psellipse[linewidth=0.04,dimen=outer,fillstyle=solid,fillcolor=color683b](1.0,-1.0403125)(0.16,0.13)
\psellipse[linewidth=0.04,dimen=outer,fillstyle=solid,fillcolor=color685b](1.0,-1.4603125)(0.16,0.13)
\usefont{T1}{ptm}{m}{n}
\rput(0.57140625,1.9746875){$b$}
\usefont{T1}{ptm}{m}{n}
\rput(1.7914063,1.9746875){$b$}
\usefont{T1}{ptm}{m}{n}
\rput(1.1814063,1.9746875){$a$}
\usefont{T1}{ptm}{m}{n}
\rput(2.3614063,1.9746875){$a$}
\usefont{T1}{ptm}{m}{n}
\rput(2.5114062,-0.2453125){$1$}
\usefont{T1}{ptm}{m}{n}
\rput(3.0514061,-0.6453125){$2$}
\usefont{T1}{ptm}{m}{n}
\rput(3.7114062,-1.0253125){$3$}
\usefont{T1}{ptm}{m}{n}
\rput(4.311406,-1.4453125){$4$}
\end{pspicture} 
}
\vspace*{-4mm}
 \caption{Timeline of  $\A_3, \A_4$ events.}
\label{figure:exampleTimeline}
 \vspace{-1em}
\end{figure}
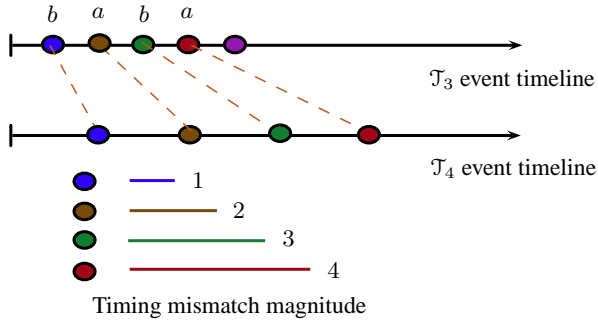
\qed
\end{example}

\begin{example}[Eventual Maximal Time Difference]
Consider the two timed automata $\A_1$ and $\A_2$ in Figure~\ref{figure:ExampleOne}.
Let us look at the value of the \emph{eventual maximal time difference} 
 quantitative timed simulation distance
$\simfunc_{\limmaxdiff}$ for the
state pair  $\left(\tuple{a,x=0}^{\A_1}, \tuple{a,x=0}^{\A_2}\right)$.
The value is 
(1)~infinity
if the state of $\A_1$ does not \emph{time-abstract simulate} (combined with time-divergence  encoded as a fairness constraint to allow only time-divergent runs in $\A_1$) the state of $\A_2$;
(2)~the \emph{eventual} 
maximal time difference between matching transitions of $\A_1$ and
 $\A_2$ otherwise (ignoring the time differences amongst finite timed trace prefixes), 
amongst time-divergent runs.
In the automata $\A_1,\A_2$, there is a time mismatch only at the transitions from
$a$, and this transition can only occur before time $10$.
Once the executions reach the location $c$, the automaton $\A_2$ is able to match
the transitions of $\A_1$ at the exact times, with zero time discrepancy.
Thus, $\simfunc_{\limmaxdiff}$ denotes the ``steady-state'' time discrepancy between
$\A_1$, $\A_2$, and this value is zero for the
state pair $\left(\tuple{a,x=0}^{\A_1}, \tuple{a,x=0}^{\A_2}\right)$,
in contrast to the value of $9$ for $\simfunc_{\maxdiff}$ for the state pair.
Note that we ignore time-discrepancies for finite \emph{time} (by discarding Zeno
runs), not just finite
trace prefixes.
If we ignore only finite trace prefixes, then we would have obtained a value of $9$,
as $\A_1$ can loop on the location $b$ by preventing time from progressing (note that
the clock $x$ is not reset on the $b$ loop transition).
\qed
\end{example}

\begin{example}[Eventual Maximal Time Difference]
%
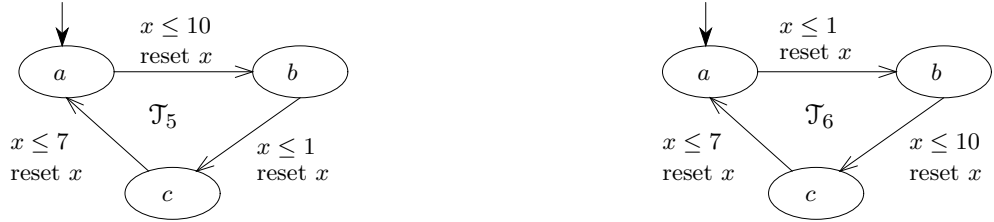
\begin{figure}[h]
\vspace{-1em}
 \begin{minipage}[t]{0.6\linewidth}
 \hspace*{-2mm}
      \strut\centerline{\setlength{\unitlength}{0.00049869in}
\begingroup\makeatletter\ifx\SetFigFont\undefined%
\gdef\SetFigFont#1#2#3#4#5{%
  \reset@font\fontsize{#1}{#2pt}%
  \fontfamily{#3}\fontseries{#4}\fontshape{#5}%
  \selectfont}%
\fi\endgroup%
{\renewcommand{\dashlinestretch}{30}
\begin{picture}(3533,2291)(0,-10)
\put(600,1544){\ellipse{990}{540}}
\put(3030,1544){\ellipse{990}{540}}
\put(1702,277){\ellipse{990}{540}}
\path(1095,1544)(2535,1544)
\path(2355.000,1484.000)(2535.000,1544.000)(2355.000,1604.000)
\path(1455,509)(600,1274)
\path(774.151,1198.691)(600.000,1274.000)(694.136,1109.262)
\path(3030,1274)(1950,509)
\path(2062.203,662.005)(1950.000,509.000)(2131.565,564.082)
\path(555,2264)(555,1814)
\blacken\path(495.000,1994.000)(555.000,1814.000)(615.000,1994.000)(555.000,1940.000)(495.000,1994.000)
\put(465,1454){\makebox(0,0)[lb]{\smash{{\SetFigFont{9}{10.8}{\familydefault}{\mddefault}{\updefault}$a$}}}}
\put(2895,1454){\makebox(0,0)[lb]{\smash{{\SetFigFont{9}{10.8}{\familydefault}{\mddefault}{\updefault}$b$}}}}
\put(1590,194){\makebox(0,0)[lb]{\smash{{\SetFigFont{9}{10.8}{\familydefault}{\mddefault}{\updefault}$c$}}}}
\put(1365,1634){\makebox(0,0)[lb]{\smash{{\SetFigFont{9}{10.8}{\familydefault}{\mddefault}{\updefault}$\text{reset } x$}}}}
\put(15,734){\makebox(0,0)[lb]{\smash{{\SetFigFont{9}{10.8}{\familydefault}{\mddefault}{\updefault}$x\leq 7$}}}}
\put(2580,689){\makebox(0,0)[lb]{\smash{{\SetFigFont{9}{10.8}{\familydefault}{\mddefault}{\updefault}$x\leq 1$}}}}
\put(1365,1949){\makebox(0,0)[lb]{\smash{{\SetFigFont{9}{10.8}{\familydefault}{\mddefault}{\updefault}$x\leq 10$}}}}
\put(2580,419){\makebox(0,0)[lb]{\smash{{\SetFigFont{9}{10.8}{\familydefault}{\mddefault}{\updefault}$\text{reset } x$}}}}
\put(15,419){\makebox(0,0)[lb]{\smash{{\SetFigFont{9}{10.8}{\familydefault}{\mddefault}{\updefault}$\text{reset } x$}}}}
\put(1455,1004){\makebox(0,0)[lb]{\smash{{\SetFigFont{11}{13.2}{\familydefault}{\mddefault}{\updefault}$\A_5$}}}}
\end{picture}
}}
    \end{minipage}
  \begin{minipage}[t]{0.35\linewidth}
 \hspace*{5mm}
    \strut\centerline{\setlength{\unitlength}{0.00049869in}
\begingroup\makeatletter\ifx\SetFigFont\undefined%
\gdef\SetFigFont#1#2#3#4#5{%
  \reset@font\fontsize{#1}{#2pt}%
  \fontfamily{#3}\fontseries{#4}\fontshape{#5}%
  \selectfont}%
\fi\endgroup%
{\renewcommand{\dashlinestretch}{30}
\begin{picture}(3443,2291)(0,-10)
\put(510,1544){\ellipse{990}{540}}
\put(2940,1544){\ellipse{990}{540}}
\put(1612,277){\ellipse{990}{540}}
\path(1005,1544)(2445,1544)
\path(2265.000,1484.000)(2445.000,1544.000)(2265.000,1604.000)
\path(1365,509)(510,1274)
\path(684.151,1198.691)(510.000,1274.000)(604.136,1109.262)
\path(2940,1274)(1860,509)
\path(1972.203,662.005)(1860.000,509.000)(2041.565,564.082)
\path(465,2264)(465,1814)
\blacken\path(405.000,1994.000)(465.000,1814.000)(525.000,1994.000)(465.000,1940.000)(405.000,1994.000)
\put(375,1454){\makebox(0,0)[lb]{\smash{{\SetFigFont{9}{10.8}{\familydefault}{\mddefault}{\updefault}$a$}}}}
\put(2805,1454){\makebox(0,0)[lb]{\smash{{\SetFigFont{9}{10.8}{\familydefault}{\mddefault}{\updefault}$b$}}}}
\put(1500,194){\makebox(0,0)[lb]{\smash{{\SetFigFont{9}{10.8}{\familydefault}{\mddefault}{\updefault}$c$}}}}
\put(1230,1949){\makebox(0,0)[lb]{\smash{{\SetFigFont{9}{10.8}{\familydefault}{\mddefault}{\updefault}$x\leq 1$}}}}
\put(1230,1679){\makebox(0,0)[lb]{\smash{{\SetFigFont{9}{10.8}{\familydefault}{\mddefault}{\updefault}$\text{reset } x$}}}}
\put(15,734){\makebox(0,0)[lb]{\smash{{\SetFigFont{9}{10.8}{\familydefault}{\mddefault}{\updefault}$x\leq 7$}}}}
\put(2580,734){\makebox(0,0)[lb]{\smash{{\SetFigFont{9}{10.8}{\familydefault}{\mddefault}{\updefault}$x\leq 10$}}}}
\put(2580,419){\makebox(0,0)[lb]{\smash{{\SetFigFont{9}{10.8}{\familydefault}{\mddefault}{\updefault}$\text{reset } x$}}}}
\put(15,419){\makebox(0,0)[lb]{\smash{{\SetFigFont{9}{10.8}{\familydefault}{\mddefault}{\updefault}$\text{reset } x$}}}}
\put(1500,1004){\makebox(0,0)[lb]{\smash{{\SetFigFont{11}{13.2}{\familydefault}{\mddefault}{\updefault}$\A_6$}}}}
\end{picture}
}}
  \end{minipage}
   \caption{Two timed automata $\A_5, \A_6$.}
  \label{figure:ExampleTwo} 
\end{figure}
Consider the two timed automata $\A_5$ and $\A_6$ in Figure~\ref{figure:ExampleTwo}.
Let us look at the value of the \emph{eventual maximal time difference} 
 quantitative timed simulation distance
$\simfunc_{\limmaxdiff}$ for the state pair
 $\left(\tuple{a,x=0}^{\A_5}, \tuple{a,x=0}^{\A_6}\right)$.
In this case, a time difference of $9$ occurs infinitely often in 
\emph{time-divergent} runs, 
(\emph{e.g.} between the paths 
$\tuple{a,x=0}^{\A_5} \stackrel{10}{\longrightarrow} \tuple{b,x=0}^{\A_5} 
\stackrel{0}{\longrightarrow} \tuple{c,x=0}^{\A_5} \stackrel{5}{\longrightarrow} \tuple{a,x=0}^{\A_5}\stackrel{10}
 {\longrightarrow} \cdots$ and
$\tuple{a,x=0}^{\A_6} \stackrel{1}{\longrightarrow} \tuple{b,x=0}^{\A_6} 
\stackrel{9}{\longrightarrow} \tuple{c,x=0}^{\A_6} \stackrel{5}{\longrightarrow} \tuple{a,x=0}^{\A_6}\stackrel{1}
{ \longrightarrow} \cdots$.
The maximal time difference of $9$ time units arises when taking the transitions from
the $a$-labelled states.
Thus, the value of $\simfunc_{\limmaxdiff}$ for the
state pair $\left(\tuple{a,x=0}^{\A_5}, \tuple{a,x=0}^{\A_6}\right)$ is $9$.
It can be checked that in this case, the value of  $\simfunc_{\maxdiff}$
for the state pair is also $9$.
\qed
\end{example}


\begin{example}[Average Time Difference]
Consider the two timed automata $\A_5$ and $\A_6$ in Figure~\ref{figure:ExampleTwo}.
Let us look at the value of the (long-run) \emph{average time difference} 
 quantitative timed simulation distance
$\simfunc_{\limavg}$ for the state pair
 $\left(\tuple{a,x=0}^{\A_5}, \tuple{a,x=0}^{\A_6}\right)$.
As usual, for the value to be finite, we require time-abstract simulation (with time-divergence).
If time-abstract simulation holds, we 
 take the average with respect to the number of transitions (over non-Zeno runs).
For the state pair, a time difference of $9$ occurs infinitely often, but this difference occurs
in only  one-third of the transitions (the transitions from $a$ locations).
For the transitions from $b$ and $c$, the time discrepancy is zero.
Thus, the value for $\simfunc_{\limavg}$ is $ \frac{9+0+0}{3} = 3$.
\qed
\end{example}

To compute the values of the three simulation functions, we use  the framework
of turn-based games on finite-state game graphs.
We introduce two new game theoretic objectives (these objectives are required for 
computing two of the quantitative simulation functions) on these game
graphs, namely, \emph{eventual debit-sum level}  
and \emph{average debit-sum level} objectives,
 and present novel solutions for both.
We need to consider the  sums of the weights encountered as in our
quantitative simulation functions, the global time is the sum of the time
durations of all the preceding transitions.

Eventual debit-sum level and average debit-sum level games are also 
interesting on their own. 
We next illustrate average debit-sum level games.
These games are played on two-player finite-state turn-based game graphs.
Each transition in the game graph incurs a cost (denoted by a negative weight), or
a reward (denoted by a positive weight).
These costs can be viewed as monetary losses, or monetary gains.
The debit-sum level at a stage in the game denotes the absolute value of the
monetary balance if the balance  is negative
(the balance is the sum of all the positive and negative costs and rewards).
The objective of player~1 is to have the lowest possible average debit-sum level.
As a financial application, consider
  the case when  banks have to
take overnight loans from the Central Bank  loan windows in case of need (these loans
need to be renewed each day the loan is not repaid). 
It is in the banks interests to minimize the average of the loan amount per day.

\begin{example}[Debit Sum-Level Turn Based Games]
\label{example:dummy}
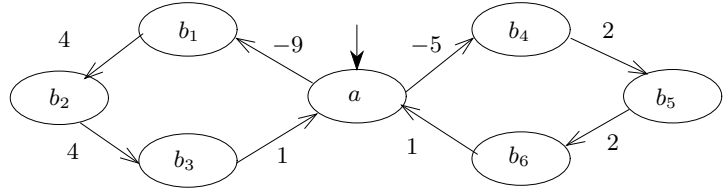
\begin{figure}[t]
\vspace{-1em}
  \strut\centerline{\setlength{\unitlength}{0.00052493in}
\begingroup\makeatletter\ifx\SetFigFont\undefined%
\gdef\SetFigFont#1#2#3#4#5{%
  \reset@font\fontsize{#1}{#2pt}%
  \fontfamily{#3}\fontseries{#4}\fontshape{#5}%
  \selectfont}%
\fi\endgroup%
{\renewcommand{\dashlinestretch}{30}
\begin{picture}(7197,1892)(0,-10)
\put(1804,1599){\ellipse{990}{540}}
\put(1801,277){\ellipse{990}{540}}
\put(5164,1600){\ellipse{990}{540}}
\put(5164,295){\ellipse{990}{540}}
\put(6694,925){\ellipse{990}{540}}
\put(3507,925){\ellipse{990}{540}}
\put(503,909){\ellipse{990}{540}}
\path(3507,1645)(3507,1195)
\blacken\path(3447.000,1375.000)(3507.000,1195.000)(3567.000,1375.000)(3507.000,1321.000)(3447.000,1375.000)
\path(3057,1060)(2292,1510)
\path(2477.569,1470.452)(2292.000,1510.000)(2416.727,1367.020)
\path(4002,970)(4677,1510)
\path(4573.925,1350.703)(4677.000,1510.000)(4498.962,1444.407)
\path(2292,250)(3102,745)
\path(2979.696,599.942)(3102.000,745.000)(2917.122,702.336)
\path(5667,1510)(6432,1150)
\path(6243.585,1172.354)(6432.000,1150.000)(6294.680,1280.933)
\path(6252,790)(5622,430)
\path(5748.515,571.400)(5622.000,430.000)(5808.052,467.210)
\path(4722,340)(3957,835)
\path(4140.718,787.589)(3957.000,835.000)(4075.528,686.841)
\path(1347,1555)(762,1105)
\path(868.090,1262.305)(762.000,1105.000)(941.255,1167.190)
\path(717,655)(1302,250)
\path(1119.853,303.126)(1302.000,250.000)(1188.158,401.789)
\put(1707,1510){\makebox(0,0)[lb]{\smash{{\SetFigFont{9}{10.8}{\familydefault}{\mddefault}{\updefault}$b_1$}}}}
\put(402,835){\makebox(0,0)[lb]{\smash{{\SetFigFont{9}{10.8}{\familydefault}{\mddefault}{\updefault}$b_2$}}}}
\put(5037,1510){\makebox(0,0)[lb]{\smash{{\SetFigFont{9}{10.8}{\familydefault}{\mddefault}{\updefault}$b_4$}}}}
\put(6522,835){\makebox(0,0)[lb]{\smash{{\SetFigFont{9}{10.8}{\familydefault}{\mddefault}{\updefault}$b_5$}}}}
\put(5037,250){\makebox(0,0)[lb]{\smash{{\SetFigFont{9}{10.8}{\familydefault}{\mddefault}{\updefault}$b_6$}}}}
\put(3417,880){\makebox(0,0)[lb]{\smash{{\SetFigFont{9}{10.8}{\familydefault}{\mddefault}{\updefault}$a$}}}}
\put(492,1420){\makebox(0,0)[lb]{\smash{{\SetFigFont{9}{10.8}{\familydefault}{\mddefault}{\updefault}$4$}}}}
\put(4047,1375){\makebox(0,0)[lb]{\smash{{\SetFigFont{9}{10.8}{\familydefault}{\mddefault}{\updefault}$-5$}}}}
\put(5982,1510){\makebox(0,0)[lb]{\smash{{\SetFigFont{9}{10.8}{\familydefault}{\mddefault}{\updefault}$2$}}}}
\put(6027,385){\makebox(0,0)[lb]{\smash{{\SetFigFont{9}{10.8}{\familydefault}{\mddefault}{\updefault}$2$}}}}
\put(2652,1375){\makebox(0,0)[lb]{\smash{{\SetFigFont{9}{10.8}{\familydefault}{\mddefault}{\updefault}$-9$}}}}
\put(582,295){\makebox(0,0)[lb]{\smash{{\SetFigFont{9}{10.8}{\familydefault}{\mddefault}{\updefault}$4$}}}}
\put(1662,205){\makebox(0,0)[lb]{\smash{{\SetFigFont{9}{10.8}{\familydefault}{\mddefault}{\updefault}$b_3$}}}}
\put(2697,250){\makebox(0,0)[lb]{\smash{{\SetFigFont{9}{10.8}{\familydefault}{\mddefault}{\updefault}$1$}}}}
\put(4002,340){\makebox(0,0)[lb]{\smash{{\SetFigFont{9}{10.8}{\familydefault}{\mddefault}{\updefault}$1$}}}}
\end{picture}
}}
  \caption{Debit sum-level game}
  \label{figure:exampleDebitSum}
\vspace{-1em}
\end{figure}
Consider the turn-based game depicted in Figure~\ref{figure:exampleDebitSum}.
All locations are player-1 locations.
The numbers on the edges denote the costs or rewards that player-1 gets when
that transition is taken. 
Positive weights denotes rewards, and negative weights denotes costs.
Viewing the weights as monetary transactions, and starting with a monetary 
balance of zero
at $a$, 
if player~1 loops around the left loop, then the trace, together with the monetary 
balances is:
$\left((a, 0)\  (b_1,-9)\ (b_2,-5)\  (b_3,-1)\right)^{\omega}$, where the numbers
denote the accumulated balances during the run of the play.
The average negative balance, \emph{i.e}, the average debit-sum level (per unit
location visit), is
$\frac{0 + 9 + 5 +1}{4} = \frac{15}{4}$.
If  player~1 loops around the right loop, then the trace, together with the balances is:
$\left((a, 0)\  (b_4,-5)\  (b_5,-3)\  (b_6,-1)\  \right)^{\omega}$.
The average negative balance is
$\frac{0 + 5 + 3 +1}{4} = \frac{9}{4}$.
Thus the optimum average debit sum-level value for player~1 is $9/4$, and the optimum
strategy is to loop around the right-hand side, where it needs to borrow less, on average.
\qed
\end{example}

\noindent\textbf{Our Contributions.}
Our main contributions in the present work are as follows.
\begin{compactenum}[$\star$]
\item We define  three quantitative refinement metrics incorporating Zenoness conditions
semantically, that is our refinement metrics ignore artificial Zeno runs present in
systems due to modelling artifacts.
We also show that these quantitative functions are actually (directed) metrics.

\item We define 
quantitative timed simulation functions corresponding to the refinement
metrics using a game theoretic formulation.
These quantitative simulation functions also incorporate Zenoness conditions
for obtaining physically meaningful system differences.
As far we know, this is is the first work which handles Zeno runs when
computing simulation functions.

\item We present \emph{algorithms} for computing all the defined quantitative timed simulation
functions to within any desired degree of accuracy for \emph{any} given timed automaton. 

\item We introduce new game theoretic objectives on
finite-state turn-based game graphs,  namely, eventual debit-sum level objectives and average debit-sum level objectives,
 and present novel solutions for both.
These new objectives are required in the computation of the  defined quantitative simulation functions.

\end{compactenum}
We have considered the (more challenging) framework of global event times
in our quantitative simulation functions.
Our solution framework is also applicable where the mismatches are only with
respect to transition \emph{durations} (simple algorithms are applicable in this case).
Our algorithms can easily be generalized to consider quantitative simulation functions
in which an observation $\sigma$ is allowed to match a different observation $\sigma'$, but with some matching penalty in case $\sigma\neq \sigma'$ (the
penalty being in addition to the timing mismatch of $\sigma$,  $\sigma'$).
Thus, our algorithms apply to the computation of quantitative simulation functions
which consider the \emph{Skorokhod metric}~\cite{Jacod} over mismatches.

\smallskip\noindent\textbf{Related Work.}
There has been a recent body of work on the
theory of \emph{approximate bisimulation} for continuous
and switched systems (\emph{e.g.}~\cite{GirardP11, GirardSwitch, GLPappas08, Girard201334, Girard2012947, tab08}).
The focus of the  approach in~\cite{GirardSwitch, GLPappas08}
is on systems with real-valued outputs
and the approximations are targeted towards 
 output values which change in a continuous fashion.
 The focus is not on the
 timing aspect.
The simulation relations are constrained to match output values at exactly the same
sample points, thus there is no mechanism to incorporate  the time discrepancies.
The work in~\cite{Girard201334} uses quantized system values in the bisimulation relation, 
and shows
how this can be used to synthesize controllers of lower complexity.
Approximate bisimilar models have also been used in symbolic frameworks to design controllers for various classes of systems and desired properties, see, \emph{e.g.}~\cite{Girard2012947, tab08}.
The work in~\cite{QueselFD11} presents a similarity relation for hybrid systems (which are more
general than the timed automata in our work)  where the
approximation is  with respect to the maximal timing mismatch over runs, 
as well as output values.
Computation of similarity relations is reduced to solving a class of derived hybrid games, 
however, these games are not decidable.
The work gives a sufficient condition which ensures decidability.
For timed systems, the work in ~\cite{HMP05} presented maximal time difference
quantitative timed simulation functions, however, Zeno issues were ignored.
Our solutions for the new objectives on finite-state game graphs builds on 
previous work on mean payoff parity games, multi-dimensional mean payoff, and
energy games~\cite{Bouyer_ATVA11, Chaloupka, CHJ05, CDHR10, CD10}.
The new game objectives presented in our work, that are
required for the quantitative timed
simulation functions, were previously unstudied, and require new ideas in their
solutions.

\section{Quantitative Timed Trace Difference and Refinement  Metrics}
\label{sec:refine}

In this section we define {\em quantitative} refinement functions on timed systems
which quantify timing mismatches. These functions  allow approximate
matching of timed traces and generalize timed and untimed refinement  relations.

\smallskip\noindent\textbf{Timed Transition System (TTS).}
A  \emph{timed transition system} (\tts) is a tuple 
$A = \tuple{S , \Sigma, \rightarrow , \mu , S_0}$ where
\begin{compactitem}
\item $S$ is the set of states.
\item $\Sigma$ is a set of atomic propositions (the observations).
\item $\rightarrow\, \subseteq\, S\times \reals^+ \times S$ is the transition relation such that
for all $s\in S$ there exists at least one $s'\in S$ such that for some $\Delta$, we have
that $(s,\Delta, s') $ belongs to $\rightarrow$.
\item $ \mu:\, S\, \mapsto\, 2^{\Sigma}$ is the observation map which assigns a truth value to atomic propositions in each state.
\item $S_0\,\subseteq\, S$ is the set of initial states.
\end{compactitem}
We write $s\stackrel{t}{\rightarrow}s'$ if $(s,t,s')$ belongs to $\rightarrow$.
A {\em state trajectory} is an infinite sequence 
$s_0 \stackrel{t_0}{\rightarrow} s_1 \stackrel{t_1}{\rightarrow} \cdots$,
where for each $j\geq 0$, we have $s_j \stackrel{t_j}{\rightarrow} s_{j+1}$. 
The state trajectory is {\em initialized} if $s_0\in S_0$ is an initial state.
A state trajectory $s_0\smash{\stackrel{t_0}{\rightarrow}} s_1\cdots$ induces 
a {\em trace} given by the observation sequence 
$\mu(s_0) \stackrel{t_0}{\rightarrow} \mu(s_1) \stackrel{t_1}{\rightarrow} \cdots$.
To emphasize the initial state, we say $s_0$-trace for a trace 
induced by a state trajectory
starting from $s_0$.
A trace is initialized if it is induced by an initialized state trajectory.
 Given a trace $\trace$ induced by a state trajectory
 $s_0 \stackrel{t_0}{\rightarrow} s_1 \stackrel{t_1}{\rightarrow} \cdots$,
let
$\mtime_{\tr}[i]$ denote $\sum_{j=0}^i t_j$, \emph{i.e.} the time of the $i$-th
transition.
The trace  $\tr$ is \emph{time-convergent} or \emph{Zeno} if 
$\lim_{i\rightarrow \infty}\mtime_{\tr}[i]$ is finite; otherwise it is \emph{time-divergent} or 
\emph{non-Zeno}.
We denote the sets  allof time-divergent initialized 
trajectories, and of time-divergent initialized traces
of a timed transition system $A$ by $\td(A)$ and $\mu(\td(A))$ respectively,
and the sets of intialized trajectories, and of all initialized traces of $A$ by $\trajecs(A)$
and $\mu(\trajecs(A))$ respectively.
A \tts is \emph{well-formed} if from every  $s_0\in S_0$, there \emph{exists} a $s_0$-trace
in   $\td(A)$.
We consider only well-formed \tts{s} in the sequel.
The  \tts $A_{\myr}$ \emph{refines} or \emph{implements} the \tts $A_{\mys}$ (the specification)
if every initialized trace of $A_{\myr}$ is also an initialized trace of $A_{\mys}$.
We first define various quantitative  notions of refinement
that quantify 
if the behavior of an implementation \tts is ``close enough'' to a specification \tts.
We begin by defining several  metrics on trace differences and refinements.

\medskip\noindent\textbf{Maximal Timing Mismatch Trace Difference Distance.}
Given two traces $\trace = \sigma_0 \stackrel{t_0}{\rightarrow}\sigma_1  \stackrel{t_1}{\rightarrow} \sigma_2 \dots$ and 
$\trace' = \sigma'_0 \stackrel{t'_0}{\rightarrow}\sigma'_1  \stackrel{t'_1}{\rightarrow} \sigma'_2 \dots$, 
the maximal timing mismatch trace difference distance $\dist_{\maxdiff}(\trace,\trace')$
is defined by
\[
 \dist_{\maxdiff}(\trace,\trace')  = \left\{
\begin{array}{ll}
\infty & \text{if }\sigma_n\neq \sigma'_n \text{ for some } n\\
 \sup_n\set{|\mtime_{\trace}[n] - \mtime_{\trace'}[n]|} &  \text{otherwise} 
\end{array}
\right.
\]
The distance $\dist_{\maxdiff}(\trace,\trace')$ indicates the maximal time
discrepancy between matching observations in the two traces $\trace$ and $\trace'$.
\begin{proposition}
\label{proposition:DistMaxDiffMetric}
The function  $\dist_{\maxdiff}()$ is a metric on timed traces.\qedhere\qed
\end{proposition}

\medskip\noindent\textbf{Refinement Distance Induced by $\dist_{\maxdiff}$.}
The trace difference metric $\dist_{\maxdiff}$ induces a {\em refinement distance} 
between two \tts{}s.
Given two \tts{}s $A_{\myr}$ (the refined system) and $A_{\mys}$ (the specification), 
with initial state sets
 $S_{\myr},S_{\mys}$ respectively, the 
\emph{refinement distance} of $A_{\myr}$ with respect to $A_{\mys}$ induced by $\dist_{\maxdiff}$
is given by
\[
\refine_{\maxdiff}(A_{\myr}, A_{\mys}) =
\sup_{\trace_{q_{\myr}}\in \mu(\td(A_{\myr}))} 
\inf_{\trace_{q_{\mys}} \in \mu(\trajecs(A_{\mys}))}
\dist_{\maxdiff}(\trace_{q_{\myr}},\trace_{q_{\mys}})
\]
where $\trace_{q_{\myr}}$ (respectively, $\trace_{q_{\mys}}$) is a $q_{\myr}$-trace (respectively, $q_{\mys}$-trace) for some 
$q_{\myr}\in S_{\myr}$
(respectively, $q_{\mys}\in S_{\mys}$).
We quantify over time-divergent traces of the refinement as Zeno traces are physically 
unrealizable\footnote{
Note that we do not need to put any time-divergence requirement on the traces from $A_{\mys}$; the ``$\inf$''
operator ensures that only time-divergent traces are considered from the well-fomed \tts 
$A_{\mys}$.
The distance $\dist_{\maxdiff}(\trace_{q_{\myr}},\trace_{q_{\mys}})$ is infinite if one trace is time-divergent and the other Zeno.}.
Notice that this refinement distance is asymmetric: it is a {\em directed distance}~\cite{AlfaroFS09}.
The refinement distance  $\refine_{\maxdiff}(A_{\myr}, A_{\mys}) $ indicates quantitatively 
how
well initialized time-divergent  traces of $A_{\myr}$ match corresponding initialized traces of 
$A_{\mys}$ with
respect to the $\dist_{\maxdiff}$ trace difference metric.

\begin{proposition}
The function  $\refine_{\maxdiff}()$ is a directed metric on timed transition systems.
\end{proposition}
\begin{proof}
The proof is similar to the proof of Proposition~\ref{prop:refinelimmaxdiffmetric}.
\qedhere
\end{proof}

We next define several other trace difference metrics, which in turn induce their
own  refinement distances on \tts{s}.

\medskip\noindent\textbf{Limit-Maximal Timing Mismatch Trace Difference Distance.}
Given two traces $\trace = \sigma_0 \stackrel{t_0}{\rightarrow}\sigma_1  \stackrel{t_1}{\rightarrow} \sigma_2 \dots$ and 
$\trace' = \sigma'_0 \stackrel{t'_0}{\rightarrow}\sigma'_1  \stackrel{t'_1}{\rightarrow} \sigma'_2 \dots$, 
the limit-maximal timing mismatch trace difference distance $\dist_{\limmaxdiff}(\trace,\trace')$
is defined by 
\[\dist_{\limmaxdiff}(\trace,\trace')  = 
  \left\{
\begin{array}{ll}
\infty & \text{if }\sigma_n\neq \sigma'_n 
  \text{ for some } n\\
 \lim_{M\rightarrow \infty}\sup_{n\geq M}\set{|\mtime_{\trace}[n] - \mtime_{\trace'}[n]|} &  \text{otherwise} 
\end{array}
\right.
\]
The distance $\dist_{\limmaxdiff}(\trace,\trace')$ indicates the limit-maximal time
discrepancy between matching observations in the two traces $\trace$ and $\trace'$.
That is, it indicates the eventual ``steady state'' maximal time
discrepancy, ignoring any initial spikes in the time discrepancy between the two traces (we still require all observations to be matched).

In the sequel, we view limits as having values on the extended real line (\emph{i.e.} in 
$\reals \cup \set{-\infty, \infty}$).
\begin{lemma}
\label{lemma:sumlimit}
Let $a_n$ and $b_n$ both be non-decreasing or both be  non-increasing
 sequences of real numbers for $n\geq 0$. 
Then $\lim_{n\rightarrow \infty} (a_n)$ and  $\lim_{n\rightarrow \infty} (b_n)$ both exist and 
$\lim_{n\rightarrow \infty} (a_n) \ +\ \lim_{n\rightarrow \infty} (b_n) \ =\ 
\lim_{n\rightarrow \infty}(a_n + b_n)$.\qed
\end{lemma}

\begin{lemma}
\label{lemma:supsum}
Let $a_n$ and $b_n$ be real numbers for $n\geq 0$ and let $M\geq 0$. 
Then 
$\sup_{n\geq M} \set{a_n} \ +\ \sup_{n\geq M} \set{b_n} \ \geq\ 
\sup_{n\geq M}\set{a_n + b_n}$.\qed
\end{lemma}

\begin{proposition}
\label{proposition:LimMaxDiffMetric}
The function  $\dist_{\limmaxdiff}()$ is a metric on timed traces.
\end{proposition}
\begin{proof}
We prove $\dist_{\limmaxdiff}(\trace_1,\trace_2) + 
\dist_{\limmaxdiff}(\trace_2,\trace_3) \ \geq \ 
\dist_{\limmaxdiff}(\trace_1,\trace_3)$.\\
If all  the observation sequences of $\trace_1, \trace_2, \trace_3$ are not the same,
or if  $\dist_{\limmaxdiff}(\trace_1,\trace_2) $ or $\dist_{\limmaxdiff}(\trace_2,\trace_3)$ is
infinite,
then the claim is straightforward.
So consider that the  observation sequences of the three traces are the same and that
 $\dist_{\limmaxdiff}(\trace_1,\trace_2) $ and $\dist_{\limmaxdiff}(\trace_2,\trace_3)$ are
both finite.
We have $\dist_{\limmaxdiff}(\trace_1,\trace_2) + \dist_{\limmaxdiff}(\trace_2,\trace_3) $
\begin{align*}
= & \lim_{M\rightarrow \infty}\sup_{n\geq M}\set{|\mtime_{\trace_1}[n] - \mtime_{\trace_2}[n]|}\ +\ 
  \lim_{M\rightarrow \infty}\sup_{n\geq M}\set{|\mtime_{\trace_2}[n] - \mtime_{\trace_3}[n]|}\\
= &   \lim_{M\rightarrow \infty}  \left(
  \begin{array}{l}
  \sup_{n\geq M}\set{|\mtime_{\trace_1}[n] - \mtime_{\trace_2}[n]|}\  + \\
  \sup_{n\geq M}\set{|\mtime_{\trace_2}[n] - \mtime_{\trace_3}[n]|}\\
  \end{array}
\right)
 \quad  \text{by Lemma~\ref{lemma:sumlimit}}.\\
\geq &   \lim_{M\rightarrow \infty}  \left( \sup_{n\geq M}\left\{
 \begin{array}{l}
  (|\mtime_{\trace_1}[n] - \mtime_{\trace_2}[n]|) + \\
  (|\mtime_{\trace_2}[n] - \mtime_{\trace_3}[n]|)
\end{array}
\right\}
\right) \text{ by Lemma~\ref{lemma:supsum}}.\\
\geq &  \lim_{M\rightarrow \infty}\sup_{n\geq M}\set{|\mtime_{\trace_1}[n] - \mtime_{\trace_3}[n]|}\\
= & \dist_{\limmaxdiff}(\trace_1,\trace_3).
\end{align*}
The desired  result follows.\qedhere
\end{proof}

\medskip\noindent\textbf{Refinement Distance Induced by $\dist_{\limmaxdiff}$.}
The trace difference metric $\dist_{\limmaxdiff}$ induces the refinement distance 
$\refine_{\limmaxdiff}(A_{\myr}, A_{\mys})$.
Formally, given two timed transition systems $A_{\myr}$, $A_{\mys}$, with initial state sets
 $S_{\myr},S_{\mys}$ respectively, the 
refinement distance of $A_{\myr}$ with respect to $A_{\mys}$ induced by $\dist_{\limmaxdiff}$
is given by 
%
\[\refine_{\limmaxdiff}(A_{\myr}, A_{\mys}) =
\sup_{\trace_{q_{\myr}}\in \mu(\td(A_{\myr}))}
\inf_{\trace_{q_{\mys}} \in \mu(\trajecs(A_{\mys}))} 
\dist_{\limmaxdiff}(\trace_{q_{\myr}},\trace_{q_{\mys}})
\]
where $\trace_{q_{\myr}}$ (respectively, $\trace_{q_{\mys}}$) is a $q_{\myr}$-trace (respectively, $q_{\mys}$-trace) for some 
$q_{\myr}\in S_{\myr}$
(respectively, $q_{\mys}\in S_{\mys}$).

\begin{proposition}
\label{prop:refinelimmaxdiffmetric}
The function  $\refine_{\limmaxdiff}()$ is a directed metric on timed transition systems.
\end{proposition}
\begin{proof}
We prove $\refine_{\limmaxdiff}(A_1,A_2) + 
\refine_{\limmaxdiff}(A_2,A_3) \ \geq \ 
\refine_{\limmaxdiff}(A_1,A_3)$.\\
The interesting case is when both $\refine_{\limmaxdiff}(A_1,A_2)$ and
$\refine_{\limmaxdiff}(A_2,A_3)$ are  finite.
%
%
Let $\refine_{\limmaxdiff}(A_1,A_2) = K_{1,2}$ and let $\refine_{\limmaxdiff}(A_2,A_3) = K_{2,3}$.
Consider any $\trace_1 \in  \mu(\td(A_1)) $.
Since $K_{1,2} = \sup_{\trace_{q_1}\in \mu(\td(A_1))}\inf_{\trace_{q_2} \in \mu(\trajecs(A_2))} 
\set{\dist_{\limmaxdiff}(\trace_{q_1},\trace_{q_2})}$, we have 
that $K_{1,2} \geq  \ \inf_{\trace_{q_2}} 
\set{\dist_{\limmaxdiff}(\trace_{1},\trace_{q_2})}$.
Hence we have
that for any given $\epsilon >0$, there exists $\trace_2 \in \mu(\trajecs(A_2))$
such that $\dist_{\limmaxdiff}(\trace_{1},\trace_2)  < K_{1,2} + \epsilon$.
Now, $\trace_2$ must be time divergent (\emph{i.e.}  $\trace_2\in \mu(\td(A_2))
$, otherwise $\dist_{\limmaxdiff}(\trace_{1},\trace_2)$ is not
finite.
Using a similar argument, we have that there exists a trace $\trace_3\in \mu(\trajecs(A_3))$ such 
that $\dist_{\limmaxdiff}(\trace_{2},\trace_3)  < K_{2,3} + \epsilon$.

Since 
\[\dist_{\limmaxdiff}(\trace_{1},\trace_2) + \dist_{\limmaxdiff}(\trace_{2},\trace_3)  \geq
\dist_{\limmaxdiff}(\trace_{1},\trace_3) \]
 we have that 
\[\dist_{\limmaxdiff}(\trace_{1},\trace_3)  <  K_{1,2}  + K_{2,3} + 2\cdot\epsilon.\]
Since there exists a $ \trace_3$ such that the above inquality 
holds for any $\epsilon >0$, we have that
\[\inf_{\trace_{q_3}\in \mu(\trajecs(A_3))}\dist_{\limmaxdiff}(\trace_{1},\trace_{q_3}) \leq 
K_{1,2}  + K_{2,3} .\]
And since this inequality holds for any $\trace_1 \in  \mu(\td(A_1)) $, we have
\[ \sup_{\trace_{q_1}\in \mu(\td(A_1))}\inf_{\trace_{q_3\in \mu(\trajecs(A_3))}} \dist_{\limmaxdiff}(\trace_{q_1},\trace_{q_3}) \leq 
K_{1,2}  + K_{2,3}. \]
The desired result follows.\qedhere
\end{proof}

\medskip\noindent\textbf{Limit-Average Trace Difference Distance.}
Given two traces $\trace = \sigma_0 \stackrel{t_0}{\rightarrow}\sigma_1  \stackrel{t_1}{\rightarrow} \sigma_2 \dots$ and 
$\trace' = \sigma'_0 \stackrel{t'_0}{\rightarrow}\sigma'_1  \stackrel{t'_1}{\rightarrow} \sigma'_2 \dots$, 
the limit-average  trace difference distance $\dist_{\limavg}(\trace,\trace')$
is defined by 
\[\dist_{\limavg}(\trace,\trace')  =
 \left\{
\begin{array}{ll}
\infty & \text{if }\sigma_j\neq \sigma'_j 
  \text{ for some } j\\
\lim_{M\rightarrow \infty}\left(\sup_{n\geq M}
\left\{\frac{\sum_{i=0}^n (|\mtime_{\trace}[i] - \mtime_{\trace'}[i]|) }{n}\right\}\right)&  \text{otherwise} 
\end{array}
\right.
\]
The distance $\dist_{\limavg}(\trace,\trace') $ indicates the long-run average of the time discrepancies
between the two traces.

\begin{proposition}
\label{proposition:LimAvgMetric}
The function  $\dist_{\limavg}()$ is a metric on timed traces.
\end{proposition}
\begin{proof}
We prove $\dist_{\limavg}(\trace_1,\trace_2) + \dist_{\limavg}(\trace_2,\trace_3) \ \geq \ 
\dist_{\limavg}(\trace_1,\trace_3)$.\\
If all  the observation sequences of $\trace_1, \trace_2, \trace_3$ are not the same,
or if  $\dist_{\limavg}(\trace_1,\trace_2) $ or $\dist_{\limavg}(\trace_2,\trace_3)$ is
infinite,
then the claim is straightforward.
So consider that the  observation sequences of the three traces are the same and that
  $\dist_{\limavg}(\trace_1,\trace_2) $ and $\dist_{\limavg}(\trace_2,\trace_3)$ are both finite.
We have $\dist_{\limavg}(\trace_1,\trace_2) + \dist_{\limavg}(\trace_2,\trace_3) $
\begin{align*}
= & \lim_{M\rightarrow \infty}\left(\sup_{n\geq M}
\left\{\frac{\sum_{i=0}^n (|\mtime_{\trace_1}[i] - \mtime_{\trace_2}[i]|) }{n}\right\}\right) \ +\
 \lim_{M\rightarrow \infty}\left(\sup_{n\geq M}
\left\{\frac{\sum_{i=0}^n (|\mtime_{\trace_2}[i] - \mtime_{\trace_3}[i]|) }{n}\right\}\right)\\
= &  \lim_{M\rightarrow \infty}\left(
 \begin{array}{l}
\sup_{n\geq M}
\left\{\frac{\sum_{i=0}^n (|\mtime_{\trace_1}[i] - \mtime_{\trace_2}[i]|) }{n}\right\} \ +\ \\
\sup_{n\geq M}\left\{\frac{\sum_{i=0}^n (|\mtime_{\trace_2}[i] - \mtime_{\trace_3}[i]|) }{n}\right\}
\end{array}
\right)
 \text{ by Lemma~\ref{lemma:sumlimit}}.\\
\geq &   \lim_{M\rightarrow \infty}  \left( \sup_{n\geq M}\left\{
 \begin{array}{l} 
\frac{\sum_{i=0}^n (|\mtime_{\trace_1}[i] - \mtime_{\trace_2}[i]|) }{n}\ +\\
\frac{\sum_{i=0}^n (|\mtime_{\trace_2}[i] - \mtime_{\trace_3}[i]|) }{n}
\end{array}
\right\}
\right) \text{ by Lemma~\ref{lemma:supsum}}.\\
\geq &  \lim_{M\rightarrow \infty}\left(\sup_{n\geq M}
\left\{\frac{\sum_{i=0}^n (|\mtime_{\trace_1}[i] - \mtime_{\trace_3}[i]|) }{n}\right\}\right)\\
= & \dist_{\limavg}(\trace_1,\trace_3).
\end{align*}
The desired result follows.\qedhere
\end{proof}

\medskip\noindent\textbf{Refinement Distance Induced by $\dist_{\limavg}$.}
The trace difference metric $\dist_{\limavg}$ induces the refinement distance 
$\refine_{\limavg}(A_{\myr}, A_{\mys})$.
The formal definition is as that for $\refine_{\limmaxdiff}(A_{\myr}, A_{\mys})$, 
replacing $\dist_{\limmaxdiff}$ with $\dist_{\limavg}$.

\begin{proposition}
The function  $\refine_{\limavg}()$ is a directed metric on timed transition systems.
\end{proposition}
\begin{proof}
The proof is similar to the proof of Proposition~\ref{prop:refinelimmaxdiffmetric}.
\qedhere
\end{proof}

\smallskip\noindent\textbf{A Note on Zeno-Asymmetry in Refinement Metrics.}
There is an asymmetry  in the definitions for refinement metrics
with respect to Zenoness as only Zeno behaviors of
$A_{\myr}$ are given special
treatment.
This is because in case of Zeno behavior by the specification, our
definitions automatically give a value of $\infty$, which is the correct notion.
That is, for $\Psi \in
      \set{\dist_{\maxdiff}, \dist_{\limmaxdiff}, \dist_{\limavg}}$, we have
$\Psi(\trace_{q_{\myr}}, \trace_{q_{\mys}}) = \infty$ if
$\trace_{q_{\myr}}$ is time divergent, and $\trace_{q_{\mys}}$ is
time convergent.

\section{Timed Simulation Relations}
\label{sec:simulation}
The general trace inclusion problem for timed systems is undecidable \cite{AlurD94}; and
simulation relations allow us to restrict our attention to a computable relation.
In this section we recall the definitions of timed and untimed simulation relations.
We also present timed and untimed simulation games which give an alternative
way of defining the simulation relations.
This will motivate the game theoretic definitions of quantitative timed 
simulation functions in the sequel.

\smallskip\noindent\textbf{Timed Simulation Relations.} 
Let $A_{\myr}$  and $A_{\mys}$ be  two \tts{}s. A binary relation 
$\preceq\, \subseteq S_{\myr}\times S_{\mys}$ is a 
\emph{timed simulation} if
$s_{\myr}\preceq s_{\mys}$ implies the following conditions: 
(1)~$\mu(s_{\myr}) = \mu(s_{\mys})$; and
 (2)~If $s_{\myr}\stackrel{t}{\rightarrow} s_{\myr}'$, then there exists $s_{\mys}'$ such 
       that $s_{\mys}\stackrel{t}{\rightarrow} s_{\mys}'$, and $s_{\myr}' \preceq s_{\mys}'$.
The state $s_{\myr}$ is timed simulated by the state $s_{\mys}$ if there exists a timed simulation $\preceq$ such that $s_{\myr} \preceq s_{\mys}$.
A binary relation $\equiv$ is a {\em timed bisimulation} if it is a symmetric timed simulation.
Two states $s_{\myr}$ and $s_{\mys}$ are timed bisimilar if there exists a timed bisimulation 
$\equiv$ with $s_{\myr} \equiv s_{\mys}$.
Timed bisimulation is stronger than timed simulation which in turn is stronger than trace inclusion.
If state $s_{\myr}$ is timed simulated by state $s_{\mys}$, then every $s_{\myr}$-trace is also a $s_{\mys}$-trace.

\smallskip\noindent\textbf{Untimed Simulation Relations.} 
{\em Untimed} simulation and bisimulation relations are defined analogously to
timed simulation and bisimulation relations  by ignoring the duration of
time steps. 
Formally, a binary relation $\preceq_u\, \subseteq S_{\myr}\times S_{\mys}$ is an (untimed) simulation 
if $s_{\myr}\preceq_u s_{\mys}$ implies the following conditions: 
 (1)~$\mu(s_{\myr}) = \mu(s_{\mys})$. 
(2)~If $s_{\myr}\stackrel{t}{\rightarrow} s_{\myr}'$, then there exists $s_{\mys}'$ and $t'\in\reals^+$ such
   that $s_{\mys} \stackrel{t'}{\rightarrow} s_{\mys}'$, and $s_{\myr}'\preceq s_{\mys}'$.
%
A symmetric untimed simulation relation is called an untimed bisimulation.

Timed simulation and bisimulation require that times be matched exactly.
This is often too strict a requirement, especially since timed models are approximations of the real world.
On the other hand, 
untimed simulation and bisimulation relations ignore the times on moves altogether.
Analogous to the  notions of quantitative refinement presented in Section~\ref{sec:refine},
we will define quantitative notions of simulation  functions which lie
in between these extremes in Section~\ref{section:QuantTimeSim}.
We will define quantitative simulation functions in a game theoretic
framework.
The motivation for the game theoretic framework for simulation relations
is presented next.

\smallskip\noindent\textbf{Timed and Untimed Simulation Games.}
There exists an alternative equivalent game theoretic view of timed simulation (a similar view
exists for  untimed simulation).
Given two timed transition systems   $A_{\myr}$ and $A_{\mys}$, 
consider a two player turn-based bipartite timed transition game structure 
$\simgame_t(A_{\myr},A_{\mys})$ with state
space $\left(S_{\myr}\times S_{\mys}\times\set{1}\right)\,\cup\, \left( S_{\myr}\times S_{\mys}\times \set{2}\right)$
(the full formal definitions of game structures will be presented in Section~\ref{section:Finite}).
The states of player~2 (the antagonist) are  $S_{\myr}\times S_{\mys} \times \set{2}$ and 
the states of player-1 (the
protagonist) are 
$S_{\myr}\times S_{\mys}\times \set{1}$.
The transitions are:
\begin{compactdesc}
\item[Player-2 transitions.]
$\tuple{s_{\myr},s_{\mys},2} \stackrel{\Delta_{\myr}}{\longrightarrow}\tuple{s_{\myr}',s_{\mys},1}$ such that
$s_{\myr} \stackrel{\Delta_{\myr}}{\longrightarrow} s_{\myr}'$ is a valid
transition in $A_{\myr}$.
\item[Player-1 transitions.]
$\tuple{s_{\myr},s_{\mys},1} \stackrel{\Delta_{\mys}}{\longrightarrow}\tuple{s_{\myr},s_{\mys}',2}$ such that
$s_{\mys} \stackrel{\Delta_{\mys}}{\longrightarrow} s_{\mys}'$ is a valid
transition in $A_{\mys}$.
\end{compactdesc} 

To decide if $s_{\mys}$ time-simulates $s_{\myr}$, we play the following game.
Let $\tuple{s_{\myr},s_{\mys},2}$ be the initial state such that $\mu(s_{\myr}) = \mu(s_{\mys}) $.
Player-2 picks a transition of some duration $\Delta_{\myr}$ 
from this state and moves to some state $\tuple{s_{\myr}',s_{\mys},1}$.
From $\tuple{s_{\myr}',s_{\mys},1}$,
player~1 then picks a transition of duration $\Delta_{\mys}$ such that
$ \Delta_{\mys} = \Delta_{\myr}$
 and moves to $\tuple{s_{\myr}',s_{\mys}',2}$ such that 
$\mu(s_{\mys}') = \mu(s_{\mys}') $.
If no such transition exists, then player~1 loses.
If the game can proceed forever without player-1 losing, then player~2 loses and 
player~1 wins.
If player~1 has a winning strategy  from $\tuple{s_{\myr},s_{\mys},2}$, then $s_{\mys}$ time-simulates $s_{\myr}$.
For untimed simulation, we ignore the time durations of the moves (player~1
can pick transitions of any duration from $A_{\mys}$).
We denote the  two player turn-based bipartite \emph{untimed} transition game as
$\simgame_u(A_{\myr},A_{\mys})$.

\section{Finite-state Game Graphs}
\label{section:Finite}

We will define the values of \emph{quantitative timed simulation functions} 
in Section~\ref{section:QuantTimeSim} through game theoretic formulations of
problems for finite-state turn based game graphs.
In this section, we first present the basic background on finite-state 
game graphs, and the relevant known results; 
then introduce new game theoretic objectives (that were not studied before 
but are required for quantitative timed simulation functions) and 
present solutions for the new objectives.

\subsection{Basic Definitions and Known Results}
In this section we present definitions of finite game graphs, plays, 
strategies, objectives, notion of winning, and the decision problems. 

\smallskip\noindent{\bf Game Graphs.}
A \emph{game graph} $G=\tuple{Q, E}$ consists of a finite set $Q$ of states 
partitioned into \mbox{player-$1$} states $Q_1$ and player-2 states $Q_2$ 
(i.e., $Q=Q_1 \cup Q_2$ and $Q_1 \cap Q_2=\emptyset$), 
and a set $E \subseteq Q \times Q$ of directed edges such that for all $q \in Q$,
there exists (at least one) $q' \in Q$ such that $(q,q') \in E$.
A \emph{player-$1$ game} is a game graph where $Q_1 = Q$ and $Q_2 = \emptyset$.
The subgraph of $G$ induced by $S \subseteq Q$ is the graph 
$\tuple{S, E \cap (S \times S)}$ (which is not a game graph in 
general); the subgraph induced by $S$ is a game graph if for all 
$s \in S$ there exist $s' \in S$ such that $(s,s')\in E$.

\smallskip\noindent{\bf Plays and Strategies.}
A game on $G$ starting from a state $q_0 \in Q$ is played in rounds as follows. 
If the game is in a player-1 state, then player~$1$ chooses the successor state from the set
of outgoing edges; otherwise the game is in a player-$2$ state, and player $2$ chooses the successor 
state from the set of outgoing edges. 
The game results in a \emph{play} from~$q_0$, i.e., 
an infinite path $\rho = q_0 q_1 \dots$ such that $(q_i,q_{i+1}) \in E$ for all $i \geq 0$. 
The prefix of length $n$ of $\rho$ is denoted by $\rho(n) = q_0 \dots q_n$. 
A \emph{strategy} for player~$1$ is a function
$\straa: Q^*Q_1 \to Q$ such that $(q,\straa(\rho\cdot q)) \in E$ for all $\rho \in Q^*$ 
and $q \in Q_1$. An \emph{outcome} of $\straa$ from~$q_0$ is a play $q_0 q_1 \dots$ such that 
$\straa(q_0 \dots q_i) = q_{i+1}$ for all $i \geq 0$ such that $q_i \in Q_1$. Strategy and outcome for
player~$2$ are defined analogously.
A player-1 strategy is \emph{memoryless} if it is independent of the history and depends only 
on the current state, and hence can be described as a function $\straa: Q_1 \to Q$.
Memoryless strategies for player~2 are defined analogously.
We denote by $\Straa$ and $\Strab$ the set of strategies for player~1 and player~2, respectively.
Given a starting state $q$, a strategy $\straa$ for player~1 and a strategy 
$\strab$ for player~2, we have a unique play $q_0 q_1 q_2\ldots$, 
such that $q_0=q$ and for all $i \geq 0$ we have that
(i)~if $q_i$ is a player-1 state, then $q_{i+1}= \straa(q_0, q_1, \ldots, q_i)$; and 
(ii)~if $q_i$ is a player-2 state, then $q_{i+1}= \strab(q_0, q_1, \ldots, q_i)$.
We denote the unique play as $\rho(\straa,\strab,q)$.

\smallskip\noindent{\bf Objectives.}
In this work we  consider both qualitative and quantitative 
objectives. 
We first introduce qualitative objectives that we use in our work.
A \emph{qualitative objective} for $G$ is a set $\phi \subseteq Q^\omega$ 
of winning plays. 
For a play $\rho$, we denote by $\Inf(\rho)$ the set of states that occur infinitely often 
in $\rho$.
We consider B\"uchi objectives, and its dual coB\"uchi objectives which are defined
as follows.
A B\"uchi objective consists of a set $B$ of B\"uchi states, and requires 
that the set $B$ is visited infinitely often.
Formally, the B\"uchi objective defines the following set of winning
plays: $\Buchi(B)=\set{\rho \mid \Inf(\rho) \cap B \neq \emptyset}$.
Dually the coB\"uchi objective consists of a set $C$ of coB\"uchi states
and requires that states outside $C$ be visited only finitely often, 
and defines the set $\coBuchi(C)=\set{\rho \mid \Inf(\rho) \subseteq C}$
of winning plays.
When we will consider qualitative objectives, the objective of player~1
will be disjunction of two coB\"uchi objectives, and the objective of
player~2 will be the complement (conjunction of two B\"uchi objectives).
The qualitative objectives will be used to model Zeno runs.
We now introduce several quantitative objectives. 

\smallskip\noindent\textbf{Quantitative Objectives.}
A \emph{quantitative objective} for $G$ is a function $f: Q^\omega \to \reals$ 
that maps every play to a real-valued number (in contrast a qualitative 
objective can be interpreted as a function $\phi: Q^\omega \to \set{0,1}$ 
that maps plays to Boolean rewards, with~1 for winning plays).
Let $w:E \to \zed$ be a \emph{weight function} and let us denote by $W$ the largest weight 
(in absolute value) according to $w$.
For a finite prefix  $\rho(n) = q_0 q_1 \dots q_n$ of a play we denote by 
$\EL(w)(\rho(n)) = \sum_{i=0}^{n-1} w(q_i,q_{i+1})$ the sum of the weights
of the prefix.
The \emph{debit} level at the end of the prefix  $\rho(n)$
is defined by 
\[\DL(w)(\rho(n)) =\max(0, -\sum_{i=0}^{n-1} w(q_i,q_{i+1})).\]
Note the negative sign in the definition.
The debit level denotes the amount by which the accumulated sum of the weights
has dipped below 0 at the end of $\rho(n)$ (if the sum of the weights is positive, \emph{i.e.}
there is a credit, 
then the debit level is defined to be 0).
We will consider the following quantitative objective functions.
\begin{compactdesc}


\item[Maximum debit level.]
 For a play $\rho$, the maximum debit level is the maximal debit level that
occurs in it.
Formally, for a play $\rho$ and the weight function $w$ we have 
\[\maxDL(w)(\rho)= \sup_n \DL(w)(\rho(n)) = 
\inf \set{v_0 \mid \forall n \geq 0. v_0 + \EL(w)(\rho(n)) \geq 0}.\]


\item[Eventual maximal debit level.]
 For a play $\rho$, the eventual maximum debit level is the maximal debit level that
occurs after  some point on in the play (\emph{i.e.} it is the maximal debit level that
occurs infinitely often in  the play).
Formally, for a play $\rho$ and the weight function $w$ we have 
\begin{align*}
\LimMaxDL(w)(\rho) & = \limsup_{n \to \infty} \DL(w)(\rho(n)) \\
& = \lim_{M \to \infty} \sup_{n\geq M} \DL(w)(\rho(n))\\
& = \inf \set{v_0 \mid \exists n_0 \geq 0. \forall n \geq n_0. v_0 + \EL(w)(\rho(n)) \geq 0}.
\end{align*}

\item[Average weight.]
The mean-payoff (or limit-average weight) objective function
on a play $\rho= q_0 q_1 \dots$ is the long-run average of the weights
of the play, i.e., $\MP(w)(\rho) = \limsup_{n \to \infty} \frac{1}{n}\cdot\EL(w)(\rho(n))$.

\item[Average debit-sum.]
Along with the previous objective, we introduce a new objective function, which we
call the average debit level that assigns to every play the long-run 
average of the debit levels.
Formally, 
\[\AvDL(w)(\rho)= \limsup_{n \to \infty} 
\frac{\sum_{i=0}^n \DL(w)(\rho(n))}{n} .\]
Note that since the debit level is defined to be~0 if the accumulated sum is positive 
(\emph{i.e.} a positive credit level), a positive credit cannot cancel out a positive
debit-sum in the averaging process in $\AvDL(w)(\rho)$.
Observe that in contrast to mean-payoff objective that is the average of the weights, 
the average debit level has the flavor of the average of the partial sums of the weights.
\end{compactdesc}
In the sequel, when the weight function $w$ is clear from context we will omit it and simply write $\EL(\rho(n))$ 
and $\MP(\rho)$, and so on.
For each of the above quantitative objectives, we will consider a version of 
the quantitative objective that is a disjunction with a coB\"uchi objective.
Formally for a quantitative objective $f$ and coB\"uchi objective $\coBuchi(C)$, 
the quantitative objective that is the disjunction of the two objectives is 
defined as follows for a play $\rho$: if $\rho \in \coBuchi(C)$, then the objective function 
assigns value~$0$  to $\rho$, otherwise it assigns value $f(\rho)$\footnote{We focus on objectives involving debits, and $0$ is the best possible debit value for 
player~1.}.
We will refer to the corresponding version of the quantitative objectives
with disjunction with coB\"uchi objective as $\maxDLC$, $\LimMaxDLC$, $\MPC$, and $\AvDLC$, 
respectively (and when the weight function $w$ and the coB\"uchi set $C$ is
clear from the context we drop them for simplicity).

\smallskip\noindent{\bf Winning Strategies, Optimal Value, and Optimal Strategies.}
A player-$1$ strategy~$\straa$ is \emph{winning} in a state~$q$
(we also say that player~$1$ is winning, or that $q$ is a winning state) 
 for a qualitative objective~$\phi$ if 
$\rho \in \phi$ for all outcomes~$\rho$ of~$\straa$ from~$q$.
The optimal value for a quantitative objective is the minimal value that 
player~1 can guarantee against all strategies of player~2.
Formally, for a quantitative objective $f$ that maps plays to real-valued rewards,
the optimal value $\Opt(f)(q)$ at state $q$ is defined as 
\[
\Opt(f)(q)=\inf_{\straa \in \Straa} \sup_{\strab \in \Strab} f(\rho(\straa,\strab,q)).\]
A strategy for player~1 is optimal if it achieves the optimal value 
against all strategies of player~2, i.e., a strategy $\straa^*$ is optimal 
if we have 
$
\Opt(f)(q)=\sup_{\strab \in \Strab} f(\rho(\straa^*,\strab,q))
$.
Similarly, a player-2 strategy $\pi_2^*$ is optimal if we have
$
\Opt(f)(q)=\inf_{\pi_1\in \Pi_1} f(\rho(\pi_1,\pi_2^*,q))
$.

We now present a theorem that summarizes known results about 
B\"uchi and coB\"uchi games, maximum debit level  (also known as
minimal initial credit for energy games), 
and mean-payoff games.
The results of B\"uchi and coB\"uchi objectives follow from~\cite{EJ91},
the results for maximum debit level  games follows from 
the results on energy games of~\cite{CD10}, and the result for
mean-payoff games follows from~\cite{CHJ05,Bouyer_ATVA11} 
(also note that in~\cite{CD10,CHJ05,Bouyer_ATVA11} player~1 has a
conjunction of energy (or mean-payoff) with parity objectives (parity objectives
generalize $\coBuchi$ objectives), 
whereas in our setting player~1 has the disjunction of energy (or 
mean-payoff) with $\coBuchi$, and thus the roles of player~1 and player~2 
in this work are exchanged as compared to~\cite{CD10,CHJ05,Bouyer_ATVA11}). 

\begin{theorem}\label{thrm_basic}
The following assertions hold for finite-state game graphs.
\begin{compactenum}
\item The set of winning states in games with disjunction of two coB\"uchi 
objectives can be computed in time $O(|Q|\cdot |E|)$, and memoryless winning 
strategies exist for player~1 and winning strategies of player~2 require one-bit 
memory (from their respective winning states).
\item The optimal value for maximum  debit level functions with coB\"uchi disjunctions 
can be computed in time $O(|Q|^2\cdot |E| \cdot W)$, and 
memoryless optimal strategies exist for player~1 and optimal strategies for player~2 
require finite memory.
If the optimal value is not $+\infty$, then the optimal value is at most $|Q| \cdot |W|$.
\item The optimal value for limit-average functions with coB\"uchi disjunctions can be 
computed in time $O(|Q|^2\cdot |E|\cdot W)$, and 
memoryless optimal strategies exist for player~1 and the optimal strategies 
of player~2 may require infinite memory.\qedhere\qed
\end{compactenum}
\end{theorem}

\subsection{New Results and Algorithms -- Eventual Maximal Debit  Level Objectives}

In this section we present a solution for games with eventual maximal debit  level objectives
(\emph{i.e.} for minimal initial credit for eventual survival).
%
We first present an  example which  illustrates the difference between maximum debit  level and eventual 
maximal debit 
level objectives.
\begin{example}[Maximum debit  level  vs eventual maximal debit  level]
Consider the game graph $G_0$  in Figure~\ref{figure:example-DL-LimDL}.
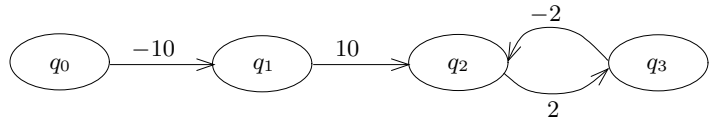
\begin{figure}[h]
    \vspace{-1em}
  \strut\centerline{\setlength{\unitlength}{0.00049869in}
\begingroup\makeatletter\ifx\SetFigFont\undefined%
\gdef\SetFigFont#1#2#3#4#5{%
  \reset@font\fontsize{#1}{#2pt}%
  \fontfamily{#3}\fontseries{#4}\fontshape{#5}%
  \selectfont}%
\fi\endgroup%
{\renewcommand{\dashlinestretch}{30}
\begin{picture}(7273,1249)(0,-10)
\put(4670,636){\ellipse{1034}{584}}
\put(525,642){\ellipse{1034}{584}}
\put(2631,623){\ellipse{1034}{584}}
\put(6748,634){\ellipse{1034}{584}}
\path(1048,613)(2128,613)
\path(1948.000,553.000)(2128.000,613.000)(1948.000,673.000)
\path(3163,613)(4153,613)
\path(3973.000,553.000)(4153.000,613.000)(3973.000,673.000)
\path(5143,523)(5145,521)(5151,516)
	(5160,508)(5173,496)(5190,482)
	(5209,466)(5230,448)(5251,431)
	(5273,415)(5293,400)(5313,387)
	(5332,375)(5351,365)(5370,356)
	(5389,348)(5408,342)(5428,335)
	(5447,331)(5467,326)(5488,322)
	(5510,319)(5533,316)(5557,313)
	(5583,311)(5608,310)(5635,309)
	(5662,309)(5688,310)(5715,311)
	(5741,313)(5767,316)(5791,319)
	(5815,323)(5838,327)(5860,332)
	(5881,337)(5901,343)(5922,350)
	(5943,358)(5963,368)(5983,379)
	(6004,391)(6026,405)(6048,421)
	(6072,439)(6096,458)(6121,478)
	(6145,499)(6168,519)(6188,536)(6223,568)
\path(6130.641,402.260)(6223.000,568.000)(6049.669,490.823)
\path(6223,658)(6221,661)(6215,666)
	(6206,676)(6192,691)(6174,709)
	(6152,731)(6129,755)(6103,780)
	(6078,805)(6053,828)(6029,850)
	(6006,869)(5984,887)(5963,903)
	(5943,917)(5923,929)(5903,940)
	(5883,949)(5863,958)(5842,966)
	(5820,973)(5798,979)(5774,985)
	(5750,990)(5725,994)(5699,997)
	(5673,999)(5646,1001)(5620,1001)
	(5593,1001)(5567,999)(5542,997)
	(5518,994)(5495,990)(5473,985)
	(5453,979)(5433,973)(5415,966)
	(5398,958)(5380,948)(5364,937)
	(5348,925)(5333,911)(5318,894)
	(5304,875)(5289,854)(5275,830)
	(5260,805)(5245,777)(5231,750)
	(5218,724)(5207,701)(5188,658)
\path(5205.868,846.893)(5188.000,658.000)(5315.631,798.394)
\put(1273,703){\makebox(0,0)[lb]{\smash{{\SetFigFont{9}{10.8}{\familydefault}{\mddefault}{\updefault}$-10$}}}}
\put(418,568){\makebox(0,0)[lb]{\smash{{\SetFigFont{9}{10.8}{\familydefault}{\mddefault}{\updefault}$q_0$}}}}
\put(2533,568){\makebox(0,0)[lb]{\smash{{\SetFigFont{9}{10.8}{\familydefault}{\mddefault}{\updefault}$q_1$}}}}
\put(4558,568){\makebox(0,0)[lb]{\smash{{\SetFigFont{9}{10.8}{\familydefault}{\mddefault}{\updefault}$q_2$}}}}
\put(3388,703){\makebox(0,0)[lb]{\smash{{\SetFigFont{9}{10.8}{\familydefault}{\mddefault}{\updefault}$10$}}}}
\put(6628,568){\makebox(0,0)[lb]{\smash{{\SetFigFont{9}{10.8}{\familydefault}{\mddefault}{\updefault}$q_3$}}}}
\put(5413,1063){\makebox(0,0)[lb]{\smash{{\SetFigFont{9}{10.8}{\familydefault}{\mddefault}{\updefault}$-2$}}}}
\put(5593,73){\makebox(0,0)[lb]{\smash{{\SetFigFont{9}{10.8}{\familydefault}{\mddefault}{\updefault}$2$}}}}
\end{picture}
}}
    \vspace{-2em}
  \caption{Game Graph $G_0$}
  \label{figure:example-DL-LimDL}
    \vspace{-0.5em}
\end{figure}
The game $G_0$ has only one play from $q_0$, namely, $q_0 \rightarrow q_1 \rightarrow \left(q_2 \rightarrow q_3\rightarrow\right)^\omega$.
It can be seen that $\Opt(\maxDL)(q_0)$ is $10$ as a debit level of $10$ is seen on the
transition from $q_0$ to $q_1$.
However, $\Opt(\LimMaxDL)(q_0)$ is only $2$, as the debit level $10$ occurs only once
in the play. 
The debit level $2$ however occurs infinitely often in the play.
\qed
\end{example}

We obtain a solution for games with eventual maximal debit  level objectives
by a reduction to  coB\"uchi games.
We start with a lemma that is required for the reduction.

\begin{lemma}\label{lemm1}
For all game graphs with a weight function $w$, the following assertions hold:
\begin{enumerate}
\item The optimal value of the eventual maximal debit  level objective
is atmost the optimal value of the maximum debit  level
objective
\emph{i.e.}, for 
all states $q$ we have
$
\Opt(\LimMaxDL)(q) \leq \Opt(\maxDL)(q)
$.
\item The optimal value   of the maximum debit level 
objective is
$+\infty$ iff  the optimal value of the eventual maximal debit  level 
is $+\infty$.
\end{enumerate}
\end{lemma}
\begin{proof}
The first item follows from definition.
The proof of the second item is as follows: 
if we have a sequence $\set{x_n}_{n \geq 0}$ of integers, then 
$\sup x_n =\infty$ iff $\lim\sup x_n=\infty$. 
Considering $\set{x_n}_{n \geq 0}$ to be the sequence $\set{\EL(\rho(n))}_{n \geq 0}$, 
we obtain the result for all plays.
Hence the result follows. 
\end{proof}

\smallskip\noindent\textbf{Reduction of $\LimMaxDL$ objective 
games to coB\"uchi Games.}

The solution for the optimal value for the $\LimMaxDL$ objective is obtained as follows.
We first compute $\Opt(\maxDL)(q)$ using algorithms of Theorem~\ref{thrm_basic}.
\begin{compactenum}
\item 
If $\Opt(\maxDL)(q)$ is infinite, then $\Opt(\LimMaxDL)(q)$ is infinite 
(by Lemma~\ref{lemm1}).
\item 
If $\Opt(\maxDL)(q)$ is finite, then by Lemma~\ref{lemm1} and by Theorem~\ref{thrm_basic}
we have $\Opt(\LimMaxDL)(q)$ 
to be  finite 
  and $\Opt(\LimMaxDL)(q) \leq |Q|\cdot W$.
\end{compactenum}
If $\Opt(\LimMaxDL)(q)$ is finite, 
the procedure
to check whether $\Opt(\LimMaxDL)(q) \leq D$ for $0 \leq D \leq |Q|\cdot W$  
is as follows: we construct a coB\"uchi 
game where we keep track of the current sum of weights; the \tcoBuchi
states are those where the tracked sum is not below $-D$;
thus ensuring that eventually the debit levels are never more than $D$.

We restrict the set, that the
tracked sums belong to, as follows.
First, we divide $Q$ into two disjoint subsets: 
$Q = Q_{\infty}  \uplus Q_{\bowtie}$ based on the value of
$\Opt(\maxDL)$. 
 States in $ Q_{\infty}$ have $\Opt(\maxDL)$ to be $+\infty$;
 states in $Q_{\bowtie}$ have $\Opt(\maxDL)$ to be finite (observe that,
by Theorem~\ref{thrm_basic}, all states in $Q_{\bowtie}$ have $\Opt(\maxDL)$
to be at most $|Q|\cdot W$).
Consider an optimal strategy $\pi_1$ for player~1 for the $\LimMaxDL$ objective
starting from a state $q$ in $Q_{\bowtie}$.
This strategy must ensure that a state in $ Q_{\infty}$ (with any tracked sum)
is never visited, as from these
states $\Opt(\maxDL)$ (and hence $\Opt(\LimMaxDL)$) is $+\infty$, thus
if $\pi_1$ were to visit  a state in $ Q_{\infty}$ (with any tracked sum)  starting from $q$, then player~2 could
 make
$\LimMaxDL$ to be  $+\infty$ from $q$ for that player-1 strategy.
Thus, we define $Q_{\infty}$ (with any tracked sum) 
to be losing sink states in the  coB\"uchi  game
 (\emph{i.e.}  these sink states
are defined to not belong to the \tcoBuchi set).

\smallskip\noindent\textbf{Bounds on Tracked Weight-Sums.}
Starting from the initial state $q$, an optimal player-1 strategy for the 
$\LimMaxDL$
objective must 
ensure that the game always stays in $Q_{\bowtie}$ states by the above argument.
Observe that if player~1
cannot avoid staying inside a negative weight-sum cycle, then all states in that cycle have $\Opt(\maxDL) = +\infty$ (and hence $\Opt(\LimMaxDL) =+\infty$).
 Thus all states in that cycle
will be outside $Q_{\bowtie}$.
Moreover, it is in the interest of player~1 to never complete a negative weight-sum
cycle. 
Thus, we have that optimal player-1 strategies  for the $\LimMaxDL$
objective
ensure that the game always stays in $Q_{\bowtie}$ states, and that a negative
weight-sum cycle is never formed.
%
Since the sum of the negative
 weights in a cycle is at most $-|Q|\cdot W$, we thus only need to keep track of weight-sums that are at least $-|Q|\cdot W$.
 If the tracked sum of the weights ever falls below $-|Q|\cdot W$,
we transition to a losing sink state in the \tcoBuchi game (\emph{i.e.} a new sink state
which is not in the \tcoBuchi set).

We now show that we only need to track weight sums that are below $|Q|\cdot W$
as follows.
Consider the states $Q_{\bowtie}$. 
From these states, starting with an initial weight sum of $0$, 
player~1 has a strategy to ensure that the sum of weights never
goes below $-|Q|\cdot W$ by definition.
This means that from  these states, starting with an initial weight sum of $|Q|\cdot W$,
player~1 has a strategy to ensure that the sum of weights never
goes below $0$.
%
%
Thus, in the game for $\LimMaxDL$,
 if the tracked sum of the weights ever exceeds $|Q|\cdot W$ at a $Q_{\bowtie}$ state,
player~1 can ensure that from that point on, non-zero debit levels will not occur.
This means that   if the tracked sum of the weights ever exceeds $|Q|\cdot W$ at 
a $Q_{\bowtie}$ state,
we can transition to a
winning sink state  in the \tcoBuchi game (\emph{i.e.} a sink state
which is defined to belong to the \tcoBuchi set).
%
%

From the above two cases
it follows that  we only need to keep track of the sum of weights that lie between 
$-|Q|\cdot W$ and $|Q|\cdot W$. 
To check whether $\Opt(\LimMaxDL)(q) \leq D$ for $0 \leq D \leq |Q|\cdot W$, we
proceed as follows:
if the sum of the weights is greater than or equal to $-D$ (and it is  a 
$Q_{\bowtie}$ state), then we call the state a coB\"uchi 
state, otherwise it is a bad state for the coB\"uchi objective.
The goal of player~1 is the coB\"uchi objective, which is equivalently
the objective to ensure that from some point on the sum of the weights
is always greater than or equal to  $-D$.
Using a binary search for $D$ for values between $0$ and $|Q|\cdot W$ 
we obtain the optimal value.
The games we construct for the binary searches 
have at most $O(|Q|^2 \cdot W)$ states and $O(|E|\cdot |Q|\cdot W)$
edges. 

For disjunction with a coB\"uchi objective, we have the same 
reduction as above, but in the end we obtain a game with disjunction 
of two coB\"uchi objectives.

\begin{example}
\label{example:EvMaxDebNew}
We illustrate the procedure of solving  eventual maximal debit level
 coB\"uchi games
with an example.
Consider the game graph $G_1$  in Figure~\ref{figure:example-AvDL-red}.
\begin{figure}[h]
  \vspace{-1em}
  \strut\centerline{\setlength{\unitlength}{0.00043745in}
\begingroup\makeatletter\ifx\SetFigFont\undefined%
\gdef\SetFigFont#1#2#3#4#5{%
  \reset@font\fontsize{#1}{#2pt}%
  \fontfamily{#3}\fontseries{#4}\fontshape{#5}%
  \selectfont}%
\fi\endgroup%
{\renewcommand{\dashlinestretch}{30}
\begin{picture}(6770,3528)(0,-10)
\put(5937,1880){\ellipse{1080}{674}}
\put(4916,480){\ellipse{1080}{674}}
\put(1032,1880){\ellipse{1080}{674}}
\put(3079,3145){\ellipse{1080}{674}}
\path(3102,2217)(4137,2217)(4137,1542)
	(3102,1542)(3102,2217)
\path(1572,1902)(3102,1902)
\path(2922.000,1842.000)(3102.000,1902.000)(2922.000,1962.000)
\path(4129,1910)(5404,1910)
\path(5224.000,1850.000)(5404.000,1910.000)(5224.000,1970.000)
\path(12,1910)(499,1910)
\path(319.000,1850.000)(499.000,1910.000)(319.000,1970.000)
\path(5667,3477)(6702,3477)(6702,2802)
	(5667,2802)(5667,3477)
\path(3642,2217)(3147,2802)
\path(3309.073,2703.347)(3147.000,2802.000)(3217.466,2625.834)
\path(6342,2802)(5892,2217)
\path(5954.190,2396.255)(5892.000,2217.000)(6049.305,2323.090)
\path(5667,3027)(5307,2757)(3957,2757)(3597,3027)
\path(3777.000,2967.000)(3597.000,3027.000)(3705.000,2871.000)
\path(6477,1850)(6657,1850)(6657,12)
	(987,12)(987,1542)
\path(1047.000,1362.000)(987.000,1542.000)(927.000,1362.000)
\path(3642,3117)(3822,3252)(5352,3252)(5667,3027)
\path(5485.654,3082.799)(5667.000,3027.000)(5555.402,3180.447)
\path(5464,1707)(5461,1706)(5455,1703)
	(5444,1699)(5428,1692)(5407,1684)
	(5382,1673)(5353,1661)(5321,1647)
	(5290,1633)(5258,1619)(5228,1605)
	(5199,1591)(5173,1578)(5148,1566)
	(5126,1553)(5106,1542)(5088,1530)
	(5071,1518)(5055,1507)(5040,1495)
	(5027,1482)(5011,1467)(4995,1450)
	(4981,1433)(4967,1415)(4954,1395)
	(4942,1375)(4930,1354)(4920,1332)
	(4910,1310)(4902,1288)(4895,1266)
	(4889,1244)(4884,1222)(4881,1201)
	(4878,1180)(4877,1160)(4877,1141)
	(4878,1122)(4880,1104)(4883,1085)
	(4887,1066)(4893,1046)(4900,1026)
	(4908,1004)(4919,980)(4930,955)
	(4944,928)(4958,900)(4972,873)
	(4986,848)(4999,825)(5022,785)
\path(4880.261,911.135)(5022.000,785.000)(4984.290,970.951)
\path(4384,470)(4381,471)(4375,474)
	(4365,480)(4349,488)(4328,499)
	(4303,513)(4275,528)(4245,545)
	(4214,562)(4183,580)(4154,598)
	(4126,615)(4101,632)(4077,648)
	(4056,663)(4037,678)(4019,693)
	(4003,708)(3989,723)(3976,739)
	(3963,755)(3950,773)(3938,792)
	(3927,812)(3916,833)(3906,855)
	(3897,878)(3888,902)(3880,927)
	(3872,952)(3866,977)(3860,1002)
	(3855,1028)(3851,1053)(3847,1077)
	(3845,1100)(3843,1123)(3841,1145)
	(3841,1166)(3840,1186)(3841,1205)
	(3842,1225)(3843,1246)(3845,1265)
	(3848,1286)(3852,1306)(3857,1328)
	(3862,1350)(3869,1374)(3876,1398)
	(3884,1424)(3892,1448)(3899,1471)
	(3906,1491)(3919,1527)
\path(3914.297,1337.322)(3919.000,1527.000)(3801.431,1378.079)
\path(5449,545)(5452,546)(5458,548)
	(5469,551)(5485,555)(5506,561)
	(5531,569)(5558,577)(5587,587)
	(5615,596)(5643,606)(5669,616)
	(5693,626)(5716,636)(5737,646)
	(5757,657)(5775,668)(5793,679)
	(5810,692)(5827,705)(5842,718)
	(5857,732)(5873,746)(5888,762)
	(5904,779)(5919,796)(5935,815)
	(5950,834)(5965,855)(5979,876)
	(5993,897)(6007,919)(6019,940)
	(6031,962)(6041,984)(6051,1005)
	(6059,1026)(6066,1047)(6073,1067)
	(6078,1087)(6081,1106)(6084,1125)
	(6086,1146)(6086,1166)(6085,1187)
	(6082,1208)(6077,1230)(6071,1253)
	(6062,1278)(6052,1304)(6041,1332)
	(6027,1362)(6013,1392)(5998,1421)
	(5983,1450)(5970,1476)(5958,1497)(5937,1535)
\path(6076.578,1406.478)(5937.000,1535.000)(5971.549,1348.435)
\put(5592,972){\makebox(0,0)[lb]{\smash{{\SetFigFont{8}{9.6}{\familydefault}{\mddefault}{\updefault}$-1$}}}}
\put(5825,1789){\makebox(0,0)[lb]{\smash{{\SetFigFont{8}{9.6}{\familydefault}{\mddefault}{\updefault}$l_2$}}}}
\put(3516,1804){\makebox(0,0)[lb]{\smash{{\SetFigFont{8}{9.6}{\familydefault}{\mddefault}{\updefault}$l_1$}}}}
\put(4849,402){\makebox(0,0)[lb]{\smash{{\SetFigFont{8}{9.6}{\familydefault}{\mddefault}{\updefault}$l_5$}}}}
\put(927,1804){\makebox(0,0)[lb]{\smash{{\SetFigFont{8}{9.6}{\familydefault}{\mddefault}{\updefault}$l_0$}}}}
\put(2097,2067){\makebox(0,0)[lb]{\smash{{\SetFigFont{8}{9.6}{\familydefault}{\mddefault}{\updefault}$-1$}}}}
\put(4670,2022){\makebox(0,0)[lb]{\smash{{\SetFigFont{8}{9.6}{\familydefault}{\mddefault}{\updefault}$2$}}}}
\put(6755,762){\makebox(0,0)[lb]{\smash{{\SetFigFont{8}{9.6}{\familydefault}{\mddefault}{\updefault}$-1$}}}}
\put(4992,1235){\makebox(0,0)[lb]{\smash{{\SetFigFont{8}{9.6}{\familydefault}{\mddefault}{\updefault}$-1$}}}}
\put(3995,890){\makebox(0,0)[lb]{\smash{{\SetFigFont{8}{9.6}{\familydefault}{\mddefault}{\updefault}$-2$}}}}
\put(2922,3072){\makebox(0,0)[lb]{\smash{{\SetFigFont{8}{9.6}{\familydefault}{\mddefault}{\updefault}$l_3$}}}}
\put(6027,3072){\makebox(0,0)[lb]{\smash{{\SetFigFont{8}{9.6}{\familydefault}{\mddefault}{\updefault}$l_4$}}}}
\put(3507,2487){\makebox(0,0)[lb]{\smash{{\SetFigFont{8}{9.6}{\familydefault}{\mddefault}{\updefault}$0$}}}}
\put(5667,2487){\makebox(0,0)[lb]{\smash{{\SetFigFont{8}{9.6}{\familydefault}{\mddefault}{\updefault}$1$}}}}
\put(4452,2487){\makebox(0,0)[lb]{\smash{{\SetFigFont{8}{9.6}{\familydefault}{\mddefault}{\updefault}$-1$}}}}
\put(4317,3342){\makebox(0,0)[lb]{\smash{{\SetFigFont{8}{9.6}{\familydefault}{\mddefault}{\updefault}$2$}}}}
\end{picture}
}}
  \caption{Game Graph $G_1$}
  \label{figure:example-AvDL-red}
\end{figure}
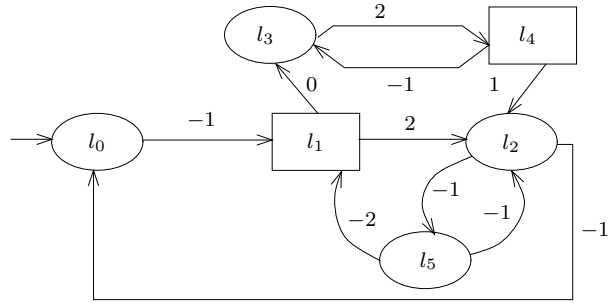
The game graph does not have any \tcoBuchi states.
The oval states are player-1 states, and the boxed states are player-2 states.
The initial state is $l_0$.
We first recall that negative weights are bad for player~1, and positive weights good;
and the dual for player~2.
Thus, the $l_3, l_4$ cycle is bad for player~2;  and the $l_2, l_5$ and $l_2, l_5, l_1$ cycles
are bad for player~1.
Observe that $\Opt(\maxDLC)(l_0) < +\infty$.
This is because  player~1 has a strategy to ensure that the debit level remains bounded.
This strategy always goes from $l_2$ to $l_0$ (its choice at $l_5$ is irrelevant).
It can be checked that with this strategy, the maximum debit level observed is $1$
(\emph{i.e}, the minimum sum of weights encountered along runs is $-1$). 
Thus, we must have $\Opt(\LimMaxDL) (l_0)  \leq  |Q| \cdot W = 6\cdot 2$ 
as argued in the discussion preceding the example.

Fix an integer $D$ with $0 \leq D \leq  |Q| \cdot W$. 
To check whether $\Opt(\LimMaxDL) (l_0)  \leq D $, we proceed as follows.
We first identify the states in $Q_{\infty}$, \emph{i.e.}, states which
 have $\Opt(\maxDL)$ to be $+\infty$. 
It can be checked that   $Q_{\infty}$ is the empty set.
Thus,  $Q_{\bowtie} = \set{l_0,l_1,l_2,l_3,l_4,l_5}$, the set of all states of $G_1$.
As explained previously, consider an enlarged game structure $G_1^{\track}$ which
keeps track of the sum of weights.
The structure of $G_1^{\track}$ is like that of $G_1$ (the weights on the edges remain
the same), in addition 
 each location of $G_1$ has
access to a finite counter which can count from  $-|Q| \cdot W$ to  $|Q| \cdot W$,
in this case from $-6\cdot 2$ to $6\cdot 2$.
This counter keeps track of the accumulative original weight sums encountered.
The starting value of the counter is $0$ at $l_0$.
After the first transition to $l_1$, its value is $-1$, and so on.
The game structure $G_1^{\track}$  also has two new sink non-\tcoBuchi  locations, 
$l_{b^1}$ and $l_{b^2}$.
At a location $l \notin \set{l_{b^1}, l_{b^2}}$,  if the value of the counter is
$c$, and the chosen edge has weight $\Delta$ then
(1)~if  $c+\Delta < -|Q| \cdot W$,
we transition to $l_{b^1}$;
(2)~if  $c+\Delta > |Q| \cdot W$,
we transition to $l_{b^2}$;
(3)~if  $|c+\Delta| \leq |Q| \cdot W$, we transition to the same location as governed by
 the corresponding
transition of  $G_1$
(and update the counter value).
Note that in the game $G_1^{\track}$, states are locations together with the
counter values (except for the states $l_{b^1}$ and $l_{b^2}$ which do not have
counter values).
 
In this modified game, we make  every state in $\set{l_0,l_1,l_2,l_3,l_4,l_5}$ with a 
counter value at least $-D$ a \tcoBuchi state. 
All other states are non-\tcoBuchi states.
Now we consider a  \tcoBuchi game played on the modified game, and see if player~1
wins from $l_0$ (starting with a counter value of $0$) for the \tcoBuchi objective.
Player~1 has a strategy to eventually keep the debit level at most $D$ iff
it wins in the modified  \tcoBuchi game.

By doing a binary search on $D$ for $0 \leq D \leq  |Q| \cdot W$, we obtain the 
optimal value of player~1 for the eventual maximal debit level
 coB\"uchi objective.
\qed
\end{example}

\begin{theorem}
The optimal player-1 strategy, and the 
optimal value $\Opt(\LimMaxDLC)(q)$ for  the eventual maximal debit  level objective 
with coB\"uchi disjunction
can be computed in time $O(|Q|^3 \cdot |E|\cdot W^2 \cdot \log (|Q|\cdot W))$.
\qedhere\qed
\end{theorem}

\subsection{New Results and Algorithms -- Average  Debit  Level Objectives}
In this section we present a solution for games with  average debit  level objectives.
We start with an example that illustrates average debit  level objectives.
\begin{example}
Consider the game graph $G_2$  in Figure~\ref{figure:example-AvDL}.
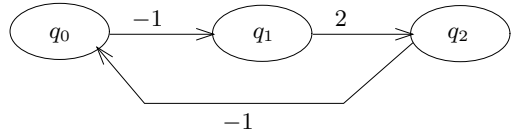
\begin{figure}[h]
    \vspace{-1em}
  \strut\centerline{\setlength{\unitlength}{0.00049869in}
\begingroup\makeatletter\ifx\SetFigFont\undefined%
\gdef\SetFigFont#1#2#3#4#5{%
  \reset@font\fontsize{#1}{#2pt}%
  \fontfamily{#3}\fontseries{#4}\fontshape{#5}%
  \selectfont}%
\fi\endgroup%
{\renewcommand{\dashlinestretch}{30}
\begin{picture}(5217,1406)(0,-10)
\put(525,1092){\ellipse{1034}{584}}
\put(2631,1073){\ellipse{1034}{584}}
\put(4692,1069){\ellipse{1034}{584}}
\path(1048,1063)(2128,1063)
\path(1948.000,1003.000)(2128.000,1063.000)(1948.000,1123.000)
\path(3163,1063)(4198,1063)
\path(4018.000,1003.000)(4198.000,1063.000)(4018.000,1123.000)
\path(4198,973)(3478,343)(1408,343)(913,928)
\path(1075.073,829.347)(913.000,928.000)(983.466,751.834)
\put(418,1018){\makebox(0,0)[lb]{\smash{{\SetFigFont{9}{10.8}{\familydefault}{\mddefault}{\updefault}$q_0$}}}}
\put(2533,1018){\makebox(0,0)[lb]{\smash{{\SetFigFont{9}{10.8}{\familydefault}{\mddefault}{\updefault}$q_1$}}}}
\put(1273,1153){\makebox(0,0)[lb]{\smash{{\SetFigFont{9}{10.8}{\familydefault}{\mddefault}{\updefault}$-1$}}}}
\put(3388,1153){\makebox(0,0)[lb]{\smash{{\SetFigFont{9}{10.8}{\familydefault}{\mddefault}{\updefault}$2$}}}}
\put(4558,1018){\makebox(0,0)[lb]{\smash{{\SetFigFont{9}{10.8}{\familydefault}{\mddefault}{\updefault}$q_2$}}}}
\put(2218,73){\makebox(0,0)[lb]{\smash{{\SetFigFont{9}{10.8}{\familydefault}{\mddefault}{\updefault}$-1$}}}}
\end{picture}
}}
    \vspace{-2.5em}
  \caption{Game Graph $G_2$}
  \label{figure:example-AvDL}
    \vspace{-0.5em}
\end{figure}
$G_2$ has only one play from $q_0$, namely, $\left(q_0 \rightarrow q_1 \rightarrow q_2 \rightarrow\right)^\omega$ 
(and similarly only one play from any state).
For this play we compute the debit  and credit  levels: let $\tuple{q, d, c}$ denote the state
$q$, and $d,c$ the debit and credit  levels at that point in the play (note that only either the debit, 
or credit level can be non-zero, by definition).
The play together with debit and credit  levels is:
$\tuple{q_0, 0, 0} \rightarrow\left( 
\tuple{q_1, 1, 0} \rightarrow \tuple{q_2, 0, 1} \rightarrow \tuple{q_0, 0, 0} \rightarrow \right)^{\omega}$.
Thus the average debit  level $\AvDL(w)(q_0)
=1/3$.
Now consider the only play from $q_2$.
The play annotated with debit and credit   levels is:
$\tuple{q_2, 0, 0} \rightarrow\left( 
\tuple{q_0, 1, 0} \rightarrow \tuple{q_1, 2, 0} \rightarrow \tuple{q_2, 0, 0} \rightarrow \right)^{\omega}$.
Note that credit levels never rise above 0 in this play.
The average debit  level $\AvDL(w)(q_2) $ for this play  is $1$.
Thus, where we ``enter'' in a cycle affects the average debit  level value.
\qed
\end{example}

The next lemma is a technical lemma on integer sequences.
\begin{lemma}\label{lemma:IncrementDecrement}
Let  $x_0, x_1, \dots $ be a sequence of integers.
The following  assertions hold.
\begin{enumerate}
\item If $x_i$ is positive for every $i$, and there
exist $i_0 \geq 0$ and $N > 0$ such that
for all $i\geq i_0$, there exists $1\leq m_i \leq N$ such that 
$x_{i+m_i} > x_i$.
Then, $
\lim_{M \rightarrow \infty} \left( \sup_{k>M}
\left\{\frac{\sum_{i=0}^{k-1} x_i}{k}\right\}\right)\, = \infty$.
\item 
Suppose (i)~there exists $W< \infty $ such that for all $i\geq 0$, 
we have $|x_{i+1}-x_i| \leq W$; and
(ii)~there exist $i_0 \geq 0$ and $N > 0$ such that
for all $i\geq i_0$, there exists $1\leq m_i \leq N$ such that  
$x_{i+m_i} < x_i$.
Then,  there exists $M\geq 0$ such that $x_i \leq 0$ for all $i \geq M$.

\end{enumerate}

\end{lemma}
\begin{proof}
We present both items of the proof.
  \begin{compactenum}
  \item  
    Consider $\sum_{i=i_0}^{i_0+\alpha\cdot N +j} x_i$ for $\alpha \geq 0 $ and
   $0\leq j < N$.
   Consider the set 
   \[X_{\alpha}=\set{x_j \mid i_o +\alpha\cdot N \leq j < i_o +\alpha\cdot (N+1)}\]
    It follows  by induction that for every $\alpha \geq 0$, we have:
   (i)~there exists 
   $x_i \in X_{\alpha}$, such that
   $x_i \geq \alpha$ (informally, the claims hold because there is an increment of at least one,
   starting from $x_{i_0}$, in every $N$ steps);
   and hence,
   (ii)~$ \sum_{i=i_0}^{i_0+\alpha\cdot N +j} x_i\,    \geq 0+1+\dots + \alpha$ (since we
   can pick $x_i \in X_{\alpha}$ such that  $x_i \geq \alpha$).
   Thus, 
 \[\frac{\sum_{i=i_0}^{i_0+\alpha\cdot N +j} x_i}{i_0+\alpha\cdot N +j}\,  \geq \,
 \frac{\alpha\cdot(\alpha+1)}{2\cdot (i_0+\alpha\cdot N +j)}\,
 \geq \frac{\alpha\cdot(\alpha+1)}{2\cdot(i_0 + (\alpha+1)\cdot N)}
\]
   for every $\alpha \geq 0 $ and
   $0\leq j < N$.
   Thus, 
   \[
   \frac{\sum_{i=i_0}^{i_0+\alpha\cdot N +j} x_i}{i_0+\alpha\cdot N +j}\,  \geq \, 
   \frac{\alpha}{2\cdot (\frac{i_0}{\alpha+1} + N)}
\]
   
   Therefore, for every $\alpha \geq 0$, we have  
   \[\sup_{k>(i_0+\alpha\cdot N)}
   \left\{\frac{\sum_{i=0}^{k-1} x_i}{k}\right\}\, \geq\,
   \frac{\sum_{i=0}^{i_0+\alpha\cdot N}  x_i}{i_0+\alpha\cdot N}\, \geq\, 
  \frac{\alpha}{2\cdot (\frac{i_0}{\alpha+1} + N)}
\]
Letting $\alpha \rightarrow \infty$, we have the desired result.

\item
  It follows from induction  that for every $\alpha \geq 0$, there
  exists $x_{i_{\alpha}} \in \set{x_j \mid i_o +\alpha\cdot N \leq j < i_o +\alpha\cdot (N+1)}$, such that
  $x_{i_{\alpha}}+\alpha \leq x_{i_0}$ (that is, $x_{i_{\alpha}}$ is at least $\alpha$ less than
  $ x_{i_0}$).
  Informally, the claims hold because there is a decrement of at least one,
  starting from $x_{i_0}$, in every $N$ steps.
  Consider any $\alpha > 1+N \cdot W+x_{i_0}$.
  Consider the set 
  \[X_{\alpha}=\set{x_j \mid i_o +\alpha\cdot N \leq j < i_o +\alpha\cdot (N+1)}\]
  Since, $|x_{i+1}-x_i| \leq W$ for $i$ in the given sequence, 
  for any $x,x' \in X_{\alpha}$, we must have $|x-x'| \leq N\cdot W$.
  Also, there exists $x_{\alpha} \in X_{\alpha}$ such that $x_{i_{\alpha}}+\alpha \leq x_{i_0}$.
  Thus, for all $x\in  X_{\alpha}$, we have
  \[ x +\alpha\,  \leq\,  x_{i_0} + N\cdot W.\]
  Since  $\alpha > 1+N \cdot W+x_{i_0}$, we have,
  \[ x + 1+N \cdot W+x_{i_0}\,  \leq\,  x_{i_0} + N\cdot W.\]
  Rearranging, we get $x\leq -1$.
  Thus, for all $i > (2+N \cdot W+x_{i_0})\cdot N$, we have $x_i \leq -1$.
  \qedhere
  \end{compactenum}
\end{proof}

\begin{corollary}\label{corollary:inftyzero}
Consider a  play $\rho = q_0 q_1 \dots$ of a finite-state game graph $G$.
The following assertions hold.
\begin{enumerate}
\item Suppose
 there
exist $i_0 \geq 0$ and $N > 0$ such that
for all $i\geq i_0$, there exists $1\leq m_i \leq N$ such that 
$ \EL(\rho(i))   >   \EL(\rho(i+m_i))$.
Then, $\AvDL(\rho)=\infty$.

\item 
Suppose
 there
exist $i_0 \geq 0$ and $N > 0$ such that
for all $i\geq i_0$, there exists $1\leq m_i \leq N$ such that 
$ \EL(\rho(i))  < \EL(\rho(i+m_i))$.
Then, $\AvDL(\rho)=0$.
\end{enumerate}

\end{corollary}
\begin{proof}
For the first assertion, it can be shown that  there
exists $i_0' \geq 0$ and $N > 0$ such that
for all $i\geq i_0'$, there exists $1\leq m_i \leq N$ such that 
$    \DL(\rho(i+m_i)) >  \DL(\rho(i)) $.
The proof of the first assertion follows from the first part of 
Lemma~\ref{lemma:IncrementDecrement}, and by the
definition of $\DL(\rho(n))$.

The proof of the second assertion follows from the second part of 
Lemma~\ref{lemma:IncrementDecrement}, and noting that
if $-\EL(\rho(n)) < 0$ then $\DL(\rho(n))=0$.
\end{proof}

\smallskip\noindent\textbf{Mean-Payoff Supremal Games for Solving
Games with Average Debit  Level Objectives.}
We define a dual objective and game to $\MPC$ in which player~1
is trying to \emph{maximize} the value; it will be used in the
solution for average debit  level objectives.
Let the quantitative objective function 
$\maxMPC$ on a play $\rho$ be defined as 
\[
\maxMPC(\rho)=
\begin{cases}
+\infty \text{ (the best player-1 payoff)}
 & 
\text{if } \rho \text{ satisfies the } 
\text{\tcoBuchi objective;}
\\
 \liminf_{n \to \infty} \frac{1}{n}\cdot\EL(w)(\rho(n)) & \text{otherwise.}
\end{cases}
\]
\noindent Let $\maxOpt( \maxMPC) $ (the game value when player~1 is 
maximizing the value of $\maxMPC$) be defined as:
\[
\maxOpt( \maxMPC) = \sup_{\pi_1 \in \Pi_1} \inf_{\pi_2 \in \Pi_2}\maxMPC (\rho(\pi_1,\pi_2,q)).
\]
We call these games  mean-payoff \emph{supremal} games\footnote{
To avoid confusion, and as a memory aid, we use the term ``supremal''  exclusively in games where player~1 is the maximizer.  The absence of ``supremal'' 
denotes  games where player~1 is the minimizer.}.
The algorithm for the solution of $\maxOpt( \maxMPC) $ is similar to the algorithm for
the solution of $\Opt(\MPC)$; and results for mean-payoff games in  Theorem~\ref{thrm_basic} 
apply also to mean-payoff supremal games.

\begin{lemma}\label{lemm:mean-payoff-cobuechi-red}
The following assertions hold: consider a weight function $w$, 
and coB\"uchi objective $\coBuchi(C)$, and then we have
\begin{enumerate}
\item 
If  $\Opt(\maxDLC)(q)=+\infty$,   then 
$\Opt(\AvDLC)(q)=+\infty.$
\item 
If  $\maxOpt(\maxMPC)(q)>0$,    then 
  $\Opt(\AvDLC)(q)=0$.
\end{enumerate}
\end{lemma}
\begin{proof}
We present proof of both the items.
\begin{compactenum}
\item If $\Opt(\maxDLC)(q)=+\infty$, then consider a finite-memory 
optimal strategy $\strab^*$ for player~2 (such a strategy 
exists by Theorem~\ref{thrm_basic}). 
Once the strategy $\strab^*$ is fixed we obtain a graph 
where only player~1 makes choices. 
Since  $\Opt(\maxDLC)(q)=+\infty$, it follows that for every 
cycle $U$ in the graph the sum of the weights in $U$ 
is negative, and there is at least one state in $U$ that is
not a coB\"uchi state (i.e., $U \cap (Q\setminus C)\neq \emptyset$). 
Since all cycles are negative the first condition of Corollary~\ref{corollary:inftyzero}
is satisfied for all paths with $N$ as the size of the graph.
Moreover the coB\"uchi objective is also falsified.
This concludes the proof of the first item.

\item Suppose  $\maxOpt(\maxMPC)(q)>0$.
%
Consider a memoryless optimal strategy $\pi_1$ for player~1 for the 
mean-payoff supremal 
objective with coB\"uchi disjunction (such a strategy exists by 
Theorem~\ref{thrm_basic}).
Since $\maxOpt(\maxMPC)(q)>0$, it follows that, in the graph obtained by 
fixing the strategy $\pi_1$, for every cycle  $U$, either the sum of the weights is
positive or $U \subseteq C$.
Consider a play $\rho$ which is an outcome of $\pi_1$, \emph{i.e.}, let
$\rho= \rho(q, \pi_1, \pi_2)$ for some player-2 strategy $\pi_2$.
For the play, either 
(i)~the \tcoBuchi objective is satisfied; or 
(ii)~$\maxMPC(\rho) > 0$, and then the second condition of Corollary~\ref{corollary:inftyzero}
is satisfied.
In either case the desired result of the second item follows.\qedhere
\end{compactenum}
\end{proof}

\smallskip\noindent\textbf{Reduction of Average Debit coB\"uchi Games
to Mean-Payoff coB\"uchi Games.}
We now use Lemma~\ref{lemm:mean-payoff-cobuechi-red} to solve the average debit 
problem.
We have the following cases.

\begin{compactenum}
\item  If $\Opt(\maxDLC)$ is $+\infty$, then $\Opt(\AvDLC)(q)=+\infty$ (by Lemma~\ref{lemm:mean-payoff-cobuechi-red}).

\item If $\maxOpt(\maxMPC)>0$, then $\Opt(\AvDLC)(q)=0$ (by Lemma~\ref{lemm:mean-payoff-cobuechi-red}).

\item The above two cases do not apply, \emph{i.e.}, $\Opt(\maxDLC)(q)$ is not $+\infty$, and 
$\maxOpt(\maxMPC)(q)\leq 0$.
\end{compactenum}
For  the last case above,
we reduce the average debit level  with \tcoBuchi disjunction problem to solving a larger mean-payoff with  \tcoBuchi disjunction
supremal game as follows: the new weights in the larger 
mean-payoff supremal game correspond to
tracked negative \emph{sums} of original weights in the  average debit level  game (\emph{i.e.} we track debit levels).
For the  mean-payoff  supremal game, we construct a new weight function according 
to the current sum of original weights: if the current sum of original weights is 
$\ell$, then the new weight function assigns value 
$\min(\ell, 0)$ 
to \emph{states} (note that games with
weights on
states can be easily transformed to games with weights on edges  by
assigning all outgoing edges from a state $q$ the state-weight of $q$).
Observe that if the  optimal value of this constructed supremal game with the new weight function with 
 mean-payoff with   \tcoBuchi disjunction supremal objective is $\alpha$; then
the optimal value for the average debit level  with  \tcoBuchi disjunction 
objective in the original
game for the original weight function is equal to  $-\alpha$.

\smallskip\noindent\textbf{Bounds on Tracked Weight-Sums.}
Consider the case when for a state $q$, we have
(i)~$\Opt(\maxDLC)(q)$ to be not $+\infty$, and 
(ii)~$\maxOpt(\maxMPC)(q)\leq 0$ (with the original weights).
\begin{compactenum}
\item 
Since $\Opt(\maxDLC)(q)$ is not $+\infty$, 
the following fact follows from Theorem~\ref{thrm_basic}.\\
\textsf{\underline{Fact 1}:} There exists a memoryless player-1 strategy $\pi_1$ such that in all plays that 
arise from $\pi_1$,
all cycles $U$ formed during the play either 
(a)~have a weight-sum that is  not  negative 
or,
(b)~ consist of only coB\"uchi states.
\item 
Since $\maxOpt(\maxMPC)(q)\leq 0$,  the following fact  follows  from Theorem~\ref{thrm_basic}.\\
\textsf{\underline{Fact 2}:} Player~2 has a strategy $\pi_2$ such that
in all plays that arise from $\pi_2$,
all cycles $U$ formed during the play 
(a)~have a weight-sum that is not  strictly positive; \emph{and}
(b)~contain at least one non-\tcoBuchi state.
\end{compactenum}
%
Now consider  average debit level  games.
\begin{compactenum}
\item 
Due to  \textsf{Fact 1}, player~1 can ensure that 
all cycles $U$ formed during the play have a weight-sum of at least $0$,
or consist of only coB\"uchi states.
Observe that  such strategies are good for player~1,
since
a negative cycle is only favorable for player~2, i.e., if a negative
cycle is executed then the average debit level increases.
\item 
Due to \textsf{Fact 2},  player~2 can ensure that no matter
the strategy of player~1, all the cycles formed will
have a weight-sum of  at most $0$, and will contain at least one non-\tcoBuchi state.
Observe that such strategies are good for player~2, since
a positive cycle is only favorable for player~1 for 
the average debit level  objective, i.e., if a positive cycle is executed
the  average debit level  decreases.
\end{compactenum}
Combining the above two statements, we have that
 there exist optimal plays (plays when both players are playing optimally),
for the average debit level  objective, where for all cycles formed along the 
play, the weight-sum   will exactly be~0, and there will be at least one 
non-\tcoBuchi state. 
\emph{Thus, if both players are playing optimally in the average debit
\tcoBuchi game, then the tracked original weight-sums  will stay in
the range $-|Q| \cdot W$ to $|Q| \cdot W$, and the \tcoBuchi objective will not be
satisfied.}
%
%
We use this fact to restrict the set that these tracked original weight-sums belong to in
the constructed mean-payoff supremal game.
We assign a new weight function:
if the current sum of original weights is 
$\ell$, then the new weight function assigns value 
$\min(\ell, 0)$ 
to \emph{states} (recall that in this supremal game, player~1 is the maximizer).
In addition, to ensure that both players play such that
the resulting plays have tracked original weight-sums in the
 range  $-|Q| \cdot W$ to $|Q| \cdot W$,
we add two new non-\tcoBuchi sink states in the constructed  mean-payoff supremal game to be used as follows.

First, if the tracked original weight-sum ever goes  below $-|Q| \vdot W$,
 we transition to a special sink non-\tcoBuchi state $q_-$ with the new \emph{state}
 weight 
$-(|Q| \vdot W +1)$.
This  state weight assignment ensures that player~1 plays such that the tracked original-weight sums
never go below $-|Q| \vdot W$,  since all states with tracked original-weight sums 
greater than or equal to  $-|Q| \vdot W$ have new state weights that are strictly
greater than $-(|Q| \vdot W +1)$,
and thus since player~1 is the maximizer  in this new game, it is in the interest 
of   player~1 
to not fall into the sink non-\tcoBuchi state $q_-$  with the lower state weight.
This first modification gives us a mean-payoff supremal game with
tracked original weight-sums in range  $-(|Q| \vdot W +1)$ to $+\infty$.

We apply a second modification:
if the tracked original weight-sum ever goes above $|Q| \cdot W$, we transition to a special sink non-\tcoBuchi
state $q_+$ with new state weight $|Q| \cdot W+1$.
%
%
This modification
 ensures that player~2 plays such that
the tracked original weight-sums never go above $|Q| \cdot W$, since all states with tracked original weight-sums 
at most $|Q| \cdot W$ 
 have new state weights at most $0$.
Since player~2 is the minimizer in this new game, it is in the interest  of player~2
 to not fall
into the sink non-\tcoBuchi state with the higher weight $|Q| \cdot W+1$.
Thus, these two modifications ensure that
we only need to keep track of original-weight sums  from 
$-|Q| \cdot W$ to $|Q| \cdot W $ in the  larger mean-payoff supremal
game.
%

\begin{example}
We illustrate the mechanisms involved in solving  average debit coB\"uchi games
with an example.
Consider the game graph $G_2$  in Figure~\ref{figure:example-AvDL-red}
from Example~\ref{example:EvMaxDebNew}.
The game graph does not have any \tcoBuchi states.
The oval states are player-1 states, and the boxed states are player-2 states.
The initial state is $l_0$.
We first recall that negative weights are bad for player~1, and positive weights are good;
and the dual for player~2.
Thus, the $l_3, l_4$ cycle is bad for player~2;  and the $l_2, l_5$ and $l_2, l_5, l_1$ cycles
are bad for player~1.
We have the following two facts for the game graph  $G_2$:
\begin{enumerate}
 \item ~$\Opt(\maxDLC)(l_0) < +\infty$, as explained in 
 Example~\ref{example:EvMaxDebNew}.
\item  $\maxOpt(\maxMPC)(l_0)\leq 0$.
This is because in the mean-payoff supremal game on $G_2$, $0$ is the maximal
value of  $ \liminf_{n \to \infty} \frac{1}{n}\cdot\EL(w)(\rho(n))$ that player~1 can ensure.
\end{enumerate}
As explained in Example~\ref{example:EvMaxDebNew}, 
consider an enlarged game structure $G_2^{\track}$ which
keeps track of the sum of weights,
together with two new sink non-\tcoBuchi  locations, 
$l_{b^1}$ and $l_{b^2}$.
The new \emph{state} weight assignment on $G_2^{\track}$ is as follows.
The new state weight of $l_{b^1}$ is $-(1+|Q| \cdot W)$; of
 $l_{b^2}$ is $1+|Q| \cdot W$; and for other locations, if the counter value is
$c$, then the state weight is $\min(c,0)$.

In the game  $G_2^{\track}$, a sample run is $l_0, l_1, (l_3, l_4)^{12}$; the value
of the counter after this run is $12$. However, $l_0, l_1, (l_3, l_4)^{13}$ is \emph{not}
a valid run as the counter value would be $13$.
Instead, the following run is valid: $l_0, l_1, (l_3, l_4)^{12}, l_3, l_{b^2},  l_{b^2},   l_{b^2}$.
Similarly, the run $l_0, l_1, l_2, (l_5, l_2)^{6}$ is valid (counter value $-11$), 
but  $l_0, l_1, l_2, (l_5, l_2)^{7}$ is not. 

Let us consider the objective $\maxMPC$ on the game graph $G_2^{\track}$
with the new state weights; where
the objective of  player~1 is to maximize the average of the new state weights
(which depend on the counter values)
seen in plays (as a simplification, there are no \tcoBuchi  states in  $G_2^{\track}$).
As proved in the discussion of the reduction of
average debit coB\"uchi games
to supremal mean-payoff coB\"uchi games, if the optimal value
(when both players are playing optimally)  of the $\maxMPC$ objective
 is $\alpha$ for the game graph $G_2^{\track}$, then 
the optimal value for the average debit level  with  \tcoBuchi disjunction 
objective in the original game  $G_2$ is $-\alpha$.
It can be seen that the optimal value of  $\maxMPC$  is $-1/3$, corresponding to plays
$(l_0,l_1, l_2)^{\omega}$; the corresponding counter values are $(0,-1, 1)^{\omega}$, with the
corresponding state weights being  $(0,-1, 0)^{\omega}$; the average of the state weights is thus
$-1/3$.
It can also be seen that the optimal value of the average debit objective is $1/3$ in
the original game  $G_2$; this value is achieved in the play $(l_0,l_1, l_2)^{\omega}$
where the debit value sequence is  $(0,1, 0)^{\omega}$.
\qed
\end{example}

The constructed game has $O(|Q|^2 \cdot W)$ states, $O(|E|\cdot |Q|\cdot W)$ edges,
and the supremal absolute value of the new weight function  is $O(|Q| \cdot W)$.
Thus our reduction and Theorem~\ref{thrm_basic} yield the following result for 
average debit level  objectives.

\begin{theorem}
The optimal player-1 strategy, and the 
optimal value $\Opt(\AvDL)(q)$ for average debit level  objective 
with coB\"uchi disjunction
can be computed in time $O(|Q|^6\cdot |E| \cdot W^4)$.
\qedhere\qed
\end{theorem}

\subsection{From Debit Level  to Difference Level  Objectives}

An easy extension of the debit level  objectives  are \emph{difference level} objectives --- 
instead of the debit levels that
arise in plays,  we consider the absolute values of the sum of the weights
(\emph{i.e} we consider $|\EL(\rho(n))|$ values).
Debit level objectives only consider sums of weights when the sums are negative;
difference level objectives consider the absolute values of sums of weights, with
the sums being both positive and negative.
We call the corresponding versions as $\DiffL$ instead of $\DL$.
These games  can be solved using  two weight functions (the original weight
function and its negation), and then applying  results for two-dimensional
energy and mean-payoff games with disjunction with coB\"uchi objectives.
Applying our techniques to solve eventual maximal debit level, and average debit 
level, along with the results of~\cite{Chaloupka,VR11,CDHR10} we obtain 
Theorem~\ref{theorem:FiniteFinal}.

\begin{theorem}
\label{theorem:FiniteFinal}
The optimal player-1 strategy, and the 
optimal value for difference-sum function with coB\"uchi disjunction,
$\Opt(\maxDiffLC)(q)$, the optimal value $\Opt(\EvMaxDiffLC)(q)$ for  the 
eventual maximal difference  level objective with coB\"uchi disjunction, and
the optimal value $\Opt(\AvDiffLC)(q)$ for average difference level 
objective with coB\"uchi disjunction, can all be computed in 
$O(\mathit{poly}(Q,E,W))$ time, where $\mathit{poly}$ is a polynomial
function.\qedhere\qed
\end{theorem}

\section
{Quantitative Timed Simulation Functions}
\label{section:QuantTimeSim}

In this section,
define quantitative timed simulation functions (\qtsf{s})
for timed transition systems in a game theoretic framework.



\smallskip\noindent\textbf{Timed Transition Game Structures.}
A \emph{timed transition game structure} is a tuple 
$\simgame_t = \tuple{S ,  \rightarrow }$ where
\begin{compactitem}
\item $S$ is the set of states, consisting of  \mbox{player-$1$} states $S_1$ and player-2 states $S_2$ 
(\emph{i.e.}, $S=S_1 \cup S_2$ and $S_1 \cap S_2=\emptyset$), 
\item $\rightarrow\, \subseteq\, S\times \reals^+ \times S$ is the transition relation such that
for all $s\in S$ there exists at least one $s'\in S$ such that for some $\Delta$, we have
$s\stackrel{\Delta}{\longrightarrow}s'$.
\end{compactitem}
Plays, objectives, strategies, outcomes \emph{etc.} are as in finite games (Section~\ref{section:Finite}).

\smallskip\noindent\textbf{Quantitative Timed Simulation Functions (\qtsf{s}).}
Analogous to the game theoretic presentation of timed simulation games,
we now present a game theoretic definition of \qtsf{s}.
Recall the two player turn-based bipartite timed transition game structure
$\simgame_t(A_{\myr},A_{\mys})$ defined in Section~\ref{sec:simulation}.
Consider a play $\rho$ in $\simgame_t(A_{\myr},A_{\mys})$:
$\tuple{s_{\myr}^0,s_{\mys}^0,2} \stackrel{\Delta_{\myr}^0}{\longrightarrow}
\tuple{s_{\myr}^1,s^0_{\mys},1}  \stackrel{\Delta_{\mys}^0}{\longrightarrow}
\tuple{s_{\myr}^1,s_{\mys}^1,2} \stackrel{\Delta_{\myr}^1}{\longrightarrow} \cdots$.
Let $\rho(\myr)$ be the projection on $A_{\myr}$, thus
$\rho(\myr) $ is the $A_{\myr}$ 
trajectory $ s_{\myr}^0 \stackrel{\Delta_{\myr}^0}{\longrightarrow} s_{\myr}^1
 \stackrel{\Delta_{\myr}^1}{\longrightarrow} \cdots
$.
Note that $\rho(\myr)$ is a valid trajectory in $A_{\myr}$.
We define  $\rho(\mys)$ similarly. 

\begin{definition}[Quantitative Timed  Simulation Objectives  Over  Game Plays]
\label{defintion:MetricPlays}
Recall the $\dist_{\maxdiff}, \dist_{\limmaxdiff}, \dist_{\limavg}$ trajectory 
difference metrics defined in Section~\ref{sec:refine}.
For $\Psi\in \set{\dist_{\maxdiff}, \dist_{\limmaxdiff}, \dist_{\limavg}}$, 
we define $\Psi^{\td}()$ as follows for  a play $\rho$
in $\simgame_t(A_{\myr},A_{\mys})$:
\[
\Psi^{\td} (\rho) =
 \begin{cases}
0  & \text{if } \rho(\myr) \notin \td(A_{\myr})\\
\Psi(\rho(\myr), \rho(\mys)) & \text{otherwise}
\end{cases}
\qedhere\qed
\]
 \end{definition}



\begin{definition}[{\qtsf{s}}]
\label{def:qsim}
Let $A_{\myr},A_{\mys}$ be timed transition systems, and let
$\simgame_t(A_{\myr},A_{\mys})$ be the two player turn-based bipartite timed transition game structure
defined in Section~\ref{sec:simulation}.
The value of the \qtsf,
denoted 
$\simfunc_{\Psi^{\td}}(\tuple{s_{\myr}, s_{\mys}})$, for
$s_{\myr}$ and $s_{\mys}$ states of $A_{\myr}$ and $A_{\mys}$ respectively, and
for $\Psi^{\td} \in 
\set{\dist_{\maxdiff}^{\td}, \dist_{\limmaxdiff}^{\td}, \dist_{\limavg}^{\td}}$,
 is defined as follows:
\[
\simfunc_{\Psi^{\td}}(\tuple{s_{\myr}, s_{\mys}})=  
\inf_{\pi_{\mys} \in \Pi_{\mys}} \sup_{\pi_{\myr}\in \Pi_{\myr}} 
 \Psi^{\td}\left(\rho\left(\pi_{\myr},\pi_{\mys}, \tuple{s_{\myr}, s_{\mys},2} \right)\right)
\]
where
$\rho\left(\pi_{\myr},\pi_{\mys}, \tuple{s_{\myr}, s_{\mys},2} \right)$ is the trajectory which results
given  the player-1 strategy $\pi_{\mys} \in \Pi_{\mys}$ and the player-2 strategy $\pi_{\myr}\in \Pi_{\myr}$.
Equivalently,
$
\simfunc_{\Psi^{\td}}(\tuple{s_{\myr}, s_{\mys}})=  \Opt(\Psi^{\td})(\tuple{s_{\myr}, s_{\mys},2})$
\qed
\end{definition}

The next proposition states that the refinement distance between two systems 
$A_{\myr}, A_{\mys}$ with the initial states $q_{\myr}, q_{\mys}$ respectively 
is at most 
the value of the corresponding simulation function.
\begin{proposition}[\qtsf{s} Over-Approximate Refinement Distances]
\label{proposition:SimulationOverapproximatesRefinement}
Let $A_{\myr}$ and $ A_{\mys}$ be two \tts{}s with the initial states $q_{\myr}$ and 
$ q_{\mys}$
 respectively.
For $(\Lambda, \Psi^{\td}) \in \set{(\maxdiff, \dist_{\maxdiff}^{\td}), 
(\limmaxdiff, \dist_{\limmaxdiff}^{\td}),
(\limavg, \dist_{\limavg}^{\td})
}$, we have 
\[\refine_{\Lambda}(A_{\myr}, A_{\mys}) \leq 
\simfunc_{\Psi^{\td}}(\tuple{s_{\myr}, s_{\mys}}).\]
\end{proposition}
\begin{proof}
Consider any $\vartheta > \simfunc_{\Psi^{\td}}(\tuple{s_{\myr}, s_{\mys}})$,
and any time-divergent $s_{\myr}$-trajectory $\traj_{\myr}$  of $A_{\myr}$.
The trajectory $\traj_{\myr}$ corresponds naturally to a player-2 strategy
$\pi_{\myr}$ in the game $\simgame_t(A_{\myr},A_{\mys})$, where player 2
picks transitions which lead to the $A_{\myr} $-projected trajectory  $\traj_{\myr}$
no matter what player 1 does.
By the definition of the simulation function, there must exist a player-1 strategy
such that the $A_{\mys} $-projected trajectory is at most $\vartheta$ away from
 $\traj_{\myr}$ according to the metric $\dist_{\Lambda}$.
The desired result follows.
\qedhere
\end{proof}

We next show that the simulation functions are actually directed metrics.
\begin{proposition}[\qtsf{s} are Directed Metrics]
For $\Psi^{\td} \in 
\set{\dist_{\maxdiff}^{\td}, \dist_{\limmaxdiff}^{\td}, \dist_{\limavg}^{\td}}$,
the function $\simfunc_{\Psi^{\td}}()$ is a directed metric over states of
\tts{}s.
\end{proposition}
\begin{proof}
It is clear that for any \tts and for any state $s$, we have $\simfunc_{\Psi^{\td}}(s,s)=0$.
We now show the triangle inequality.
Let $A_a,A_b, A_c$ be timed transition systems, with initial states $s_a, s_b, s_c$ respectively.
We show $\simfunc_{\Psi^{\td}}(s_a,s_b) + \simfunc_{\Psi^{\td}}(s_b,s_c) \geq
 \simfunc_{\Psi^{\td}}(s_a,s_c) $.

Given two states  $s_{\myr}\in A_{\myr}$ and $s_{\mys}\in  A_{\mys}$ for any two
systems $A_{\myr}$ and $A_{\mys}$, consider the game 
$\simgame_t(A_{\myr},A_{\mys})$.
We recall that player 1 is the player which is trying to simulate the other player.
We say a player-1 strategy $\pi_{\mys}$ is $\varepsilon$-\emph{optimal} for the
objective  $\Psi^{\td}$  from the state $ \tuple{s_{\myr}, s_{\mys},2} $ if
$\simfunc_{\Psi^{\td}}(\tuple{s_{\myr}, s_{\mys}}) + \varepsilon \geq 
\sup_{\pi_{\myr}\in \Pi_{\myr}} 
 \Psi^{\td}\left(\rho\left(\pi_{\mys},\pi_{\myr}, \tuple{s_{\myr}, s_{\mys},2} \right)\right)$.
 Let $K_{a,b}, K_{b,c}, K_{a,c}$ denote  $\simfunc_{\Psi^{\td}}(s_a,s_b) ,
 \simfunc_{\Psi^{\td}}(s_b,s_c) ,
 \simfunc_{\Psi^{\td}}(s_a,s_c) $ respectively.
It suffices to show that for every $\varepsilon >0$ there exists a
 player-1 strategy $\pi_{a,c,1}^*$ such that 
$K_{a,b}+ K_{b,c} + \varepsilon \geq 
\sup_{\pi_{a,c,2}\in \Pi_{a,c,2}} 
 \Psi^{\td}\left(\rho\left(\pi_{a,c,1}^*,\pi_{a,c,2}, \tuple{s_{a}, s_{c},2} \right)\right)$
for the game
$\simgame_t(A_a,A_c)$ from the state $ \tuple{s_a, s_c,2} $ for the objective
$\Psi^{\td}$ .
We construct such a player-1 strategy as follows.

Consider the case when $K_{a,b}< \infty$ and $ K_{b,c} < \infty$ (otherwise the claim is trivially proved). 
Consider an $
\varepsilon/2$-optimal player-1 strategy $\pi_{a,b,1}^{\varepsilon/2}$ for the game
$\simgame_t(A_a,A_b)$ from the state $ \tuple{s_a, s_b,2} $ for the objective
$\Psi^{\td}$.
The player-1 strategy $\pi_{a,b,1}^{\varepsilon/2}$  can be used to  map any finite or infinite 
$s_a$-trajectory $\traj_a$  of  $A_a$ to a
unique $s_b$-trajectory $\traj_b$ of  $A_b$:
consider a player-2 strategy
which ``looks'' only at the $A_a$ component and blindly generates $\traj_a$;
the trajectory $\traj_b$ is the corresponding $A_b$ trajectory as a result
of $\pi_{a,b,1}^{\varepsilon/2}$ playing against this player-2 strategy.
Let this map from trajectories of $A_a$  to those of $A_b$ be denoted as
$\mymap_{a,b, \varepsilon/2}$.
Observe  that for an infinite time-divergent trajectory $\traj_a$, we have
$\Psi(\traj_a, \mymap_{a,b, \varepsilon/2}(\traj_a)) < K_{a,b}+\varepsilon/2$ by the definition
of the  player-1 strategy $\pi_{a, b,1}^{\varepsilon/2}$  being $
\varepsilon/2$ optimal.
Similarly, there exists  an $
\varepsilon/2$-optimal player-1 strategy $\pi_{b,c,1}^{\varepsilon/2}$ for the game
$\simgame_t(A_b,A_c)$ from the state $ \tuple{s_b, s_c,2} $ for the objective
$\Psi^{\td}$.


We now define a player-1 strategy $\pi_{a,c,1}^*$ for the game
$\simgame_t(A_a,A_c)$ from the state $ \tuple{s_a, s_c,2} $ as follows (it
essentially picks the same moves as $\pi_{b,c,1}^{\varepsilon/2}$).
Consider a $\simgame_t(A_a,A_c)$  play 
$\rho_{a,c}^f = \tuple{s_{a}^0,s_{c}^0,2} \stackrel{\Delta_{a}^0}{\longrightarrow}
\tuple{s_{a}^1,s^0_{c},1}  \stackrel{\Delta_{c}^0}{\longrightarrow}
\tuple{s_{a}^1,s_{c}^1,2} \stackrel{\Delta_{a}^1}{\longrightarrow} \cdots
\tuple{s_{a}^n,s^{n-1}_{c},1}
$.
Let $\mymap_{a,b,\varepsilon/2}\left(s_{a}^0  \stackrel{\Delta_{a}^0}{\longrightarrow}
s_{a}^1 \cdots
 \stackrel{\Delta_{a}^{n-1}}{\longrightarrow}
s_{a}^n\right) = 
s_{b}^0  \stackrel{\Delta_{b}^0}{\longrightarrow}
s_{b}^1 \cdots
 \stackrel{\Delta_{b}^{n-1}}{\longrightarrow}
s_{b}^n
$.
Using $\mymap_{a,b,\varepsilon/2}()$,
we map a play $\rho_{a,c}^f $ of $\simgame_t(A_a,A_c)$
to a play $\myhap_{b,c, \varepsilon/2}(\rho_{a,c}^f) $  of $\simgame_t(A_b,A_c)$ using 
$\mymap_{a,b, \varepsilon/2}()$ on the $A_a$ components, leaving the $A_c$ components unchanged, as
follows.
The play $\myhap_{b,c, \varepsilon/2}(\rho_{a,c}^f) $ of $\simgame_t(A_b,A_c)$ 
is  defined to be 
$ \tuple{s_{b}^0,s_{c}^0,2} \stackrel{\Delta_{b}^0}{\longrightarrow}
\tuple{s_{b}^1,s^0_{c},1}  \stackrel{\Delta_{c}^0}{\longrightarrow}
\tuple{s_{b}^1,s_{c}^1,2} \stackrel{\Delta_{b}^1}{\longrightarrow} \cdots
\tuple{s_{b}^n,s^{n-1}_{c},1}
$ (note that the components with the subscript $c$ remain unchanged).
We note that $\myhap_{b,c, \varepsilon/2}(\rho_{a,c}^f) $ is a valid play of  $\simgame_t(A_b,A_c)$.
Finally,  $\pi_{a,c,1}^*(\rho_{a,c}^f)$ is defined to be 
$\pi_{b, c,1}^{\varepsilon/2} \left( \myhap_{b,c, \varepsilon/2}\left(\rho_{a,c}^f\right) \right)$;
that is,  $\pi_{a,c,1}^*(\rho_{a,c}^f)$ is defined the the $A_c$-state $s_c^*$ where
 $s_c^*$ is the state prescribed by the player-1 strategy $\pi_{b, c,1}^{\varepsilon/2}$
on the $\simgame_t(A_b,A_c)$-run obtained from $\rho_{a,c}^f$ by changing all the
$A_a$ components to those in 
$ \mymap_{a,b, \varepsilon/2}(\rho_{a,c}^f(a))$.
Intuitively, given a finite play $\rho_{a,c}^f$ of $\simgame_t(A_a,A_c)$, player~1
(i)~
first obtains a finite trajectory  $\xi_b$ in $A_b$ 
by
mapping $\rho_{a,c}^f(a)$ (the  $\rho_{a,c}^f$ trajectory projected onto $A_a$) to $\xi_b$ 
using the game  $\simgame_t(A_a,A_b)$
and the player-1 strategy $\pi_{a, b,1}^{\varepsilon/2}$; 
(ii)~utilizes the fact that 
the finite trajectories  $\rho_{a,c}^f(c)$ and $\xi_b$ correspond to a play
$\rho_{b,c}^f$ of $\simgame_t(A_b,A_c)$, and
(iii)~uses  $\pi_{b,c,1}^{\varepsilon/2}(\rho_{b,c}^f)$ to prescribe the next $A_c$-state
 in the game
$\simgame_t(A_a,A_c)$ from $\rho_{a,c}^f$.

We claim the player-1 strategy  $\pi_{a,c,1}^*$ is   such that 
$K_{a,b}+ K_{b,c} + \varepsilon \geq 
\sup_{\pi_{a,c,2}\in \Pi_{a,c,2}} 
 \Psi^{\td}\left(\rho\left(\pi_{a,c,1}^*,\pi_{a,c,2}, \tuple{s_{a}, s_{c},2} \right)\right)$.
Consider any player-2 strategy $\pi_{a,c,2}$ in the game $\simgame_t(A_a,A_c)$, and
the resultant play $\rho_{a,c}^*$ when   player~1 plays with the strategy $\pi_{a,c,1}^*$.
If $\rho_{a,c}^*(a)$ is time-convergent, we are done.
Assume $\rho_{a,c}^*(a)$ is time-divergent.
The function $\myhap_{b,c, \varepsilon/2}()$ can be seen as also mapping \emph{infinite} plays 
$\rho_{a,c}$ of $\simgame_t(A_a,A_c)$ to  plays  $\rho_{b,c}$ 
of $\simgame_t(A_b,A_c)$ (the infinite mapping being
the limit of the finite mappings, and noting that our games have infinite plays).
We have 
the following facts (denoting $\myhap_{b,c, \varepsilon/2}(\rho_{a,c}^*) = \rho_{b,c}^*$):
\begin{compactenum}
\item
 $\Psi\left(\rho_{a,c}^*(a), \rho_{b,c}^*(b)\right) < K_{a,b}+\varepsilon/2$
(due to the properties of $\mymap_{a,b, \varepsilon/2}()$).
\item Recall that 
$\rho_{b,c}^*$ is \emph{defined} to the play obtained from $\rho_{a,c}^*$
by  substituting the  $A_a$ projected run with 
$\mymap_{a,b, \varepsilon/2}\left(\rho_{a,c}^*\left(a\right)\right)$.
We claim the following fact:
$\rho_{b,c}^*$ \emph{is also equal to the  play which arises in the game $\simgame_t(A_b,A_c)$ when
player 1 plays with the strategy $\pi_{b,c,1}^{\varepsilon/2}$ against a player-2 strategy
which simply picks $A_b$-transitions leading to the projected  $A_b$ trajectory 
$\mymap_{a,b, \varepsilon/2}\left(\rho_{a,c}^*\left(a\right)\right)$}.
That is, the $A_c$ run which arises due to the player-1 strategy $\pi_{a,c,1}^*$  in the play
$\rho_{a,c}^*$ is the same
as the  $A_c$ run which arises due to the player-1 strategy $\pi_{b,c,1}^{\varepsilon/2}$ in the
play $\rho_{b,c}^*$.
Thus $\rho_{b,c}^*(c) = \rho_{a,c}^*(c)$.

The reason for this fact is that $\rho_{a,c}^*$
is the play  when   player~1 plays with the strategy $\pi_{a,c,1}^*$.
The strategy $\pi_{a,c,1}^*$ is such that,  \emph{by definition},  the $A_c$-state proposed  by
$\pi_{b,c,1}^{\varepsilon/2}$ on $\rho_{b,c}^*[0..n]$, and the $A_c$-state proposed by
$\pi_{a,c,1}^*$ on $\rho_{a,c}^*[0..n]$ are exactly the same.

\end{compactenum}
From the second item above, we have that 
$\Psi\left(\rho_{b,c}^*(b), \rho_{b,c}^*(c)\right) < K_{b,c}+\varepsilon/2$ as
$\pi_{b,c,1}^{\varepsilon/2}$  is an $\varepsilon/2$ optimal player-1
strategy in the game  $\simgame_t(A_b,A_c)$.
Combining with the first item in the enumeration above,
we get by Propositions~\ref{proposition:DistMaxDiffMetric}, 
\ref{proposition:LimMaxDiffMetric} and \ref{proposition:LimAvgMetric} that 
$\Psi\left(\rho_{a,c}^*(a),  \rho_{a,c}^*(c)\right) 
<K_{a,b}+ K_{b,c}+\varepsilon$.
Noting that $\rho_{a,c}^*(a)$ is time-divergent, we get that
$\Psi^{\td}(\rho_{a,c}^*) <K_{a,b}+ K_{b,c}+\varepsilon$.
The triangle inequalities for the simulation functions follow.
\qedhere
\end{proof}




\section{Computation of Quantitative Simulation Functions on Timed Automata}
\label{section:ComputationTimedAutomata}

In this section we obtain algorithms for computing  \qtsf{s} on 
\emph{timed automata} ~\cite{AlurD94} (which  suggest a finite syntax for specifying
infinite-state timed  structures), by reducing the problem to games on finite-state graphs.
The solution involves the following steps.
We first enlarge the  \tts corresponding to the given timed
automaton
$\A$, in Subsection~\ref{subsection:Enlarge}, in order
to measure elapsed time, and to measure the integer ``time-ticks'' 
(or integer time boundaries) 
crossed during executions of  $\A$ 
(if the current real-valued time is $\Delta$, then $\floor{\Delta}$ integer time-ticks
have elapsed).
In 
Subsection~\ref{subsection:IntegerSim},
we define \emph{integer-time} \qtsf{s}
which depend only on the elapsed integer time-ticks  and show that
these integer-time simulation functions are close to the original
(real-valued) simulation functions.
Next, we show that these integer-time \qtsf{s}
can be computed by a reduction to  finite state game graphs in
Subsection~\ref{subsection:ReductionFinite}.
Finally, we present the algorithm which ties all the steps together, and
show that we can compute the quantitative simulation functions
to within any desired degree of accuracy on timed automata.

\subsection{Timed Automata}
\label{subsection:TimedAutomata}


In this section, we briefly recall the model (for a detailed treatment, see~\cite{AlurD94}).
A \emph{timed automaton} $\A$ is a tuple $\tuple{L, \Sigma, C, \mu, \rightarrow, \gamma, S_0}, $ where
\begin{compactitem}
\item $L$ is the set of locations; and  $\Sigma$ is the set of atomic propositions.
\item $C$ is a finite set of clocks.
A \emph{clock valuation} $v : C\mapsto \reals^+$ for a set of clocks 
$C$ assigns a real value to each clock in $C$.
\item $ \mu : L \mapsto 2^{\Sigma}$ is the observation map (it does not depend on clock values).
\item $\rightarrow\, \subseteq\, L\times L\times 2^C\times\Phi(C)$ gives the set of transitions, where $\Phi(C) $ is the set of clock constraints generated by $\psi \ :=\ x\leq d\mid d\leq x \mid \neg\psi\mid \psi_1\wedge \psi_2$.
\item 
$\inv: L\mapsto\clkcond(C)$ is a function that assigns to 
every location an invariant on clock valuations.
All clocks increase uniformly at the same rate.
When at location~$l$, a valid execution
 must  move out of $l$ 
before the invariant $\inv(l)$ expires.
Thus, the timed automaton  can stay at a location only as long as the invariant is 
satisfied by the clock values.

\item $S_0 \subseteq L\times{\reals^+}^{|C|}$ is the set of initial states.

\end{compactitem}
Each clock increases at rate $1$ inside a location.
A \emph{clock valuation} is a function  $\kappa : C\mapsto\reals_{\geq 0}$
that maps every clock to a nonnegative real. 
The set of all clock valuations for $C$ is denoted by $K(C)$.
Given a clock valuation $\kappa\in K(C)$ and a time delay 
$\Delta\in\reals_{\geq 0}$, we write 
$\kappa +\Delta$ for the clock valuation in $K(C)$ defined by 
$(\kappa +\Delta)(x) =\kappa(x) +\Delta$ for all clocks $x\in C$.
For a subset $\lambda\subseteq C$ of the clocks, we write 
$\kappa[\lambda:=0]$ for the clock valuation in $K(C)$ defined by 
$(\kappa[\lambda:=0])(x) = 0$ if $x\in\lambda$, 
and $(\kappa[\lambda:=0])(x)=\kappa(x)$ if $x\not\in\lambda$.
A clock valuation $\kappa\in K(C)$ \emph{satisfies} the clock constraint 
$\theta$, written $\kappa\models \theta$, if the condition 
$\theta$ holds when all clocks in $C$ take on the values specified 
by~$\kappa$.
A \emph{state} $s=\tuple{l,\kappa}$ of the timed automaton  $\A$ is a 
location $l\in L$ together with a clock valuation $\kappa\in K(C)$ such 
that the invariant at the location is  satisfied, that is,
$\kappa\models\inv(l)$.
We let $S$ be the set of all states of~$\A$.
 An edge $\tuple{l,l',\lambda, g}$ represents a transition from location $l$ 
to location $l'$ when the clock values at $l$ satisfy the constraint $g$.
 The set $\lambda\subseteq C$ gives the clocks to be reset to $0$ with this transition.
The semantics of timed automata are given as timed transition systems.
 This is standard (see \emph{e.g.} \cite{AlurD94}), and omitted here.


\smallskip\noindent{\bf Clock Region Equivalence.}
Clock region equivalence, denoted as  $\cong$ is an equivalence relation  on 
states of timed automata.
The equivalence classes of the relation are called \emph{regions}, and induce
a time abstract bisimulation on the corresponding timed transition system.
There are finitely many clock regions;
more precisely, the number of clock regions is bounded by 
$|L|\cdot\prod_{x\in C}(c_x+1)\cdot |C|!\cdot 4^{|C|}$.
For a real $t\ge 0$, let  $\fractional(t)=t-\floor{t}$ denote the 
fractional part of~$t$.
Given a timed automaton game $\A$, for each clock $x\in C$, let $c_x$
denote the largest integer constant that appears in any clock
constraint involving $x$ in~$\A$ (let $c_x=1$ if there is no clock
constraint involving~$x$).
Two states $\tuple{l_1,\kappa_1}$ and $\tuple{l_1,\kappa_1}$ are said to be
\emph{region equivalent} if all the following conditions are satisfied:
(a)~ $l_1 = l_2$,
(b)~ for all clocks $x$, we have $\kappa_1(x) \leq c_x $ iff $\kappa_2(x) \leq c_x $,
(c)~for all  clocks $x$ with $\kappa_1(x) \leq c_x $, we have 
$\floor{\kappa_1(x)}=\floor{\kappa_2(x)}$,
(d)~for all  clocks $x,y$ with $\kappa_1(x) \leq c_x $ and $\kappa_1(y) \leq c_y $,
$\fractional(\kappa_1(x)) \leq \fractional(\kappa_1(y))$ iff
$\fractional(\kappa_2(x)) \leq \fractional(\kappa_2(y))$, and
(e)~for all  clocks $x$ with $\kappa_1(x) \leq c_x $, we have
$\fractional(\kappa_1(x))=0$ iff $\fractional(\kappa_2(x))=0$.
Given a state $\tuple{l,\kappa}$ of $\A$, we denote the region containing $\tuple{l,\kappa}$ 
as $\reg(\tuple{l,\kappa})$.

\smallskip\noindent\textbf{Region Graph.}
The region graph $\reg(\A)$ corresponding to $\A$ is the time-abstract
bisimulation quotient graph induced by the region equivalence relation (see  \cite{AlurD94} for
details).
The states of $\reg(\A)$ are the regions of $\A$.
There is a transition $R\rightarrow R'$ in the region graph iff there exists $s\in R$ and
$s'\in R'$ such that there is a transition from $s$ to $s'$ according to the
semantics of $\A$.

\subsection{Enlarging the Timed Automaton \tts}
\label{subsection:Enlarge}
For ease of presentation  we assume that 
 all clocks are bounded, \emph{i.e.}, that
the invariants of each location can be conjuncted with the clause
$\wedge_{x\in C}\left( x\leq c_{\max}\right)$ for some constant $c_{\max}$.
The  general case where clocks may be
unbounded can be solved using similar algorithms, with some additional bookkeeping.

Given a timed automata $\A$ where all the clocks are bounded
by $c_{\max}$, let $\ensys$ denote the timed transition system
obtained by adding to $\A$ an extra clock $z$, which cycles between~$0$ and~$1$, for
measuring elapsed time, and an integer valued variable $\ticks$ which takes on
values in $\nat_{\leq c_{\max} }$, where $ \nat_{\leq c_{\max}}$ denotes
the set $\set{0,1,\dots, c_{\max}}$.
Formally, the set of states of $\ensys$ is 
$S^{\ensys}= S\times \reals_{[0,1)}\times \nat_{\leq c_{\max}}$, where 
$S$ is the set of states
of $\A$.
The state $\tuple{s,\z,\myd}$ of $\ensys$ has the following
components:
\begin{compactenum}[$-$]
\item $s$ is the state
of the original timed automaton $ \A$;
\item $\z$ is the value of the added clock $z$ which gets reset to 0
every time it crosses 1 (i.e., if $\kappa'$ is the clock valuation resulting
from letting time $\Delta$ elapse from an initial clock valuation $\kappa$, 
then, $\z=\kappa'(z) = (\kappa(z)+\Delta)\mod 1$); and
\item $\myd$ denotes the value of the integer variable $\ticks$, and is equal to
the number of integer boundaries crossed by the added clock $z$
since the last transition: if the clock valuation in the previous state was $\kappa$, and the
transition time duration is $\Delta$, then $\myd=\floor{\kappa(z)+\Delta}$ in the
current state, where
$\floor{}$ denotes the integer floor function.
Note that since all the clocks in $\A$ are bounded by $c_{\max}$, we have
$\myd \leq c_{\max}$, as the maximum duration of a transition  is $c_{\max}$, and
$\kappa(z) < 1$ in the previous state.
Note that $\myd$ can have a value of $1$ as the result of a transition of duration $\Delta< 1$,
\emph{e.g.}, if the clock $z$ had a value of $0.9$ in the previous state, and $\Delta =0.2$,
then $\myd=1$.
\end{compactenum}
The region equivalence relation can be expanded to $\ensys$ states.
Two states $\tuple{\tuple{l_1,\kappa_1},\z_1,\myd_1}$ and 
$\tuple{\tuple{l_2,\kappa_2}\z_2,\myd_2}$ of   $\ensys$ are defined to be
region equivalent if $\tuple{l_1, \myd_1} = \tuple{l_2, \myd_2}$, and
$\kappa_1^{z=\z_1} \cong \kappa_2^{z=\z_2} $, where
$\kappa_i^{z=\z_i}$ denotes the clock valuation $\kappa_i$ on $C$ expanded
to a clock valuation to $C\cup\set{z}$ by mapping $z $ to $\z_i$ (we denote the
enlarged clock valuation be denoted as $\w{\kappa}$).
Similar to the region graph $\reg(\A)$, we define an untimed 
 finite state bisimulation quotient
graph $\reg(\en{\A})$ for $\en{\A}$.

Given a state $s$ of $\A$, we denote by $\en{s}$ the state $\tuple{s,0,0}$ of
$\ensys$.
For a state trajectory 
$\traj = s_0 \stackrel{t_0}{\rightarrow} s_1 \stackrel{t_1}{\rightarrow} \dots$,
we let $\traj[i]$ denote the state $s_i$.
Given a state trajectory $\traj$ of the timed automaton $\A$, we denote by
$\en{\traj}$ the $\ensys$ trajectory
$\en{\traj[0]} \stackrel{t_0}{\rightarrow} \w{s}_1 \stackrel{t_1}{\rightarrow} \w{s}_2\dots$, where $\w{s}_i= \tuple{s_i, \z_i, \myd_i}$, and $ \z_i, \myd_i$ values are
according to the times of the transitions (letting $ \en{\traj[0]}= \w{s}_0$).
That is, $\en{\traj}$ denotes the trajectory obtained by adding the clock $z$, and the 
integer variable $\ticks$, where the values for both the new variables are
set to $0$ in the starting state $\en{\traj[0]}$.
The new variables just observe the time, and the integer boundaries crossed for each
transition according to the semantics for $\ensys$ described previously.
The  first component of $\en{\traj[i]}$ is the same as the state $\traj[i]$  for all $i$.

The next lemma shows that a trajectory is time-divergent iff it satisfies a B\"uchi
constraint.
\begin{lemma}
\label{lemma:TimeDivEncode}
Let $\traj$ be a  trajectory  of a timed automaton $\A$ in which
all clocks are bounded by $c_{\max}$.
The trajectory $\traj$ is time-divergent iff
$\en{\traj}$ satisfies the  B\"uchi condition $\Buchi\left( \bigvee_{i=1}^{c_{\max}}\ticks=i\right)$
\end{lemma}
\begin{proof}
The proof follows from the fact that  trajectory $\traj$ is \emph{not} time-divergent
iff global time does not progress beyond some integer $U$.
This happens iff time crosses only finitely many integer boundaries.
Now, global time crosses an integer boundary at step $n$ iff 
$ \left(\bigvee_{i=1}^{c_{\max}}\ticks=i\right)$ is true at step $n$.
Thus  trajectory $\traj$ is \emph{not} time-divergent iff
$\left(\bigvee_{i=1}^{c_{\max}}\ticks=i\right)$ is true only finitely often.
Equivalently,   trajectory $\traj$ is  time-divergent iff
$\left(\bigvee_{i=1}^{c_{\max}}\ticks=i\right)$ is true infinitely often.
\qedhere
\end{proof}

\subsection{Integer-Time  Quantitative Timed Simulation Functions}
\label{subsection:IntegerSim}

In this subsection we define quantitative simulation functions where only the
integer ``time-ticks'' encountered are of relevance (as opposed to the exact
real-valued times for the original \qtsf{s}).
The utility of these integer-time simulation functions is that they can be computed
over timed automata by reductions to finite state games.
These simulation functions are also close in value 
 to the real-valued  \qtsf{s}.
First, we define a notion of \emph{integer time} which measures the number of integer time-ticks crossed upto the current time point.


\begin{definition}[Integer Time]
For the trajectory $\en{\traj}$, let $\inttime_{\en{\traj}}[i]$ denote the number
of integer boundaries crossed upto the $i$-th transition:
$\inttime_{\en{\traj}}[i] = \floor{\mtime_{\en{\traj}}[i] }$.
\end{definition}

We have the following lemma which expresses $\inttime_{\en{\traj}}[i]$ in terms
of the of the values of the $\ticks$ variable in trajectories.
Note that the value of the $\ticks$ variable is zero in the first state
of a valid trajectory $\en{\traj}$ of $\en{\A}$.
\begin{lemma}
\label{lemma:IntTimeViaTicks}
Let $\traj$  be a trajectory  of a timed automaton $\A$ in which
all clocks are bounded. We have: 
$\inttime_{\en{\traj}}[i] =
\sum_{j=0}^i \myd_j$,  
where $\myd_i$ is the value of the $\ticks$ variable at the state  $\en{\traj}[i]$. 
\end{lemma}
\begin{proof}
The proof follows from the definition of the $\ticks$ variable updates: the updates
count the integer boundaries crossed by the clock  $z$ which measures elapsed time
\qedhere
\end{proof}
%

Using the notion of integer-time, we next define \emph{integer-time 
trajectory  difference metrics.}

\begin{definition}
[The Integer-Time Trajectory  Difference Metrics
$ \intdist_{\maxdiff}$, $\intdist_{\limmaxdiff}$,  and $\intdist_{\limavg}$]
Corresponding to the trajectory difference metric $\dist_{\varphi}()$, 
for $\varphi=\maxdiff, \limmaxdiff, \limavg$,
we define
the integer-time trace difference metric  $\intdist_{\varphi}()$, 
by substituting $\inttime\mspace{15mu}()$ for $\mtime()$ in the definition of
$\dist_{\varphi}()$. 
\emph{E.g.}, letting $\en{\traj}[n]=\tuple{\tuple{l_n,\kappa_n},\z_n,\myd_n}$ and
$\en{\traj'}[n]=\tuple{\tuple{l_n', \kappa_n'},\z_n',\myd_n'}$,  we have:
\[ \intdist_{\maxdiff}(\en{\traj},\en{\traj'})  =
 \left\{
\begin{array}{ll}
\infty & \text{if } \mu(l_n)\neq \mu(l_n') 
   \text{ for some } n\\
 \sup_n\set{|\inttime_{\en{\traj}}(n) - \inttime_{\en{\traj'}}(n)|} & \text{otherwise} 
\end{array}
\right.
\qedhere\qed
\]
%
\end{definition}

\begin{proposition}
The functions 
$ \intdist_{\maxdiff}$, $\intdist_{\limmaxdiff}$,  and $\intdist_{\limavg}$ are metrics over
timed trajectories.
\end{proposition}
\begin{proof}
The  proofs are along similar lines to the corresponding claims for
$ \dist_{\maxdiff}$, $\dist_{\limmaxdiff}$,  and $\dist_{\limavg}$. \qedhere
\end{proof}

The following lemma shows that $\intdist_{\varphi}()$ closely approximates
$\dist_{\varphi}()$.

\begin{lemma}
\label{lemma:IntTraj}
Let $\traj_1$ and $\traj_2$ be two trajectories  of a timed automaton $\A$.
The following assertions are   true for $\varphi\in\set{\maxdiff, \limmaxdiff, \limavg}$.
\begin{compactenum}
\item $\dist_{\varphi}(\en{\traj_1},\en{\traj_2})  = \infty$
iff $\intdist_{\varphi}(\en{\traj_1},\en{\traj_2})  = \infty$.
\item If both $\dist_{\varphi}(\en{\traj_1},\en{\traj_2})$ and
$\intdist_{\varphi}(\en{\traj_1},\en{\traj_2}) $ are less than $\infty$,
then
 \[\mspace{-60mu}
\dist_{\varphi}(\en{\traj_1},\en{\traj_2}) +1 \geq   
\intdist_{\varphi}(\en{\traj_1},\en{\traj_2}) \geq \dist_{\varphi}(\en{\traj_1},\en{\traj_2}) -1.\]
\end{compactenum}
\end{lemma}
\begin{proof}
The first claim is obvious.
We prove the second claim.
Let us denote the sequence $\mtime_{\en{\traj}}(n)$ as $x(n)$, the sequence
 $\mtime_{\en{\traj'}}(n)$ as $x'(n)$, the sequence $\inttime_{\en{\traj}}(n)$ as
$y(n)$ and the sequence $\inttime_{\en{\traj'}}(n)$ as
$y'(n)$.
We have for all $n\geq 0$,
\begin{align*}
x(n) -1 & \, <\,  y(n) \, \leq\,  x(n)\\ 
x'(n) -1 & \,<\,   y'(n) \, \leq\,  x'(n).
\end{align*}
Thus, we have
\[x(n) -x'(n) -1 \, < \, y(n)-y'(n) \, < \, x(n) -x'(n) +1\]
Hence
\[|x(n) -x'(n)| -1 \, < \, |y(n)-y'(n)| \, < \, |x(n) -x'(n)| +1\]
It follows that
\[
\sup_n|x(n) -x'(n)| -1 \, < \, \sup_n|y(n)-y'(n)| \, < \, \sup_n|x(n) -x'(n)| +1
\]
Thus, we have the results for $\varphi=\maxdiff$.

We also have the following two relationships:
\begin{align*}
\lim_{U\rightarrow \infty}\sup_{n> U}|x(n) -x'(n)| -1&  \, < \, 
\lim_{U\rightarrow \infty}\sup_{n> U}|y(n)-y'(n)|\\ 
\lim_{U\rightarrow \infty}\sup_{n> U}|y(n)-y'(n)|&  \, < \, 
\lim_{U\rightarrow \infty}\sup_{n>U} |x(n) -x'(n)| +1.
\end{align*}
This gives the results for $\varphi=\limmaxdiff$.

Next, we note that for every $n$, the following two relationships hold.
\begin{align*}
  \frac{\sum_{j=0}^nx(j) -n}{n+1} & \, <\,  \frac{\sum_{j=0}^ny(j)}{n+1} \, \leq\,    \frac{\sum_{j=0}^nx'(j)}{n+1}\\ 
 \frac{\sum_{j=0}^nx'(j) -n}{n+1} & \, <\,  \frac{\sum_{j=0}^ny'(j)}{n+1} \, \leq\,    \frac{\sum_{j=0}^nx'(j)}{n+1}
 \end{align*}
And thus,

\begin{align*}
  \frac{\sum_{j=0}^nx(j) }{n+1}-1 & \, <\,  \frac{\sum_{j=0}^ny(j)}{n+1} \, \leq\,    \frac{\sum_{j=0}^nx'(j)}{n+1}\\ 
 \frac{\sum_{j=0}^nx'(j) }{n+1}-1 & \, <\,  \frac{\sum_{j=0}^ny'(j)}{n+1} \, \leq\,    \frac{\sum_{j=0}^nx'(j)}{n+1}
\end{align*}

Then, applying similar reasoning as in $\varphi=\limmaxdiff$,
we get the results for $\varphi=\limavg$.
\qedhere
\end{proof}

Using $ \intdist_{\maxdiff}$, $\intdist_{\limmaxdiff}$,  and $\intdist_{\limavg}$,
we can define integer-time  quantitative simulation functions
 $\simfunc_{\intPsi^{\td}}$ 
 which
approximate $\simfunc_{\Psi^{\td}}$ for
$\Psi^{\td} \in 
\set{\dist_{\maxdiff}^{\td}, \dist_{\limmaxdiff}^{\td}, \dist_{\limavg}^{\td}}$.
The definitions follow along similar lines to the definitions for $\simfunc_{\Psi^{\td}}$.
We present them formally below.
First, we  present integer-time quantitative objectives which
map simulation game plays to integer valued numbers.

\begin{definition}[Integer-Time Quantitative Objectives  for  Timed Simulation Games]
For the trajectory difference metrics 
$\intPsi \in \set{ \intdist_{\maxdiff}, \intdist_{\limmaxdiff},  \intdist_{\limavg}}$,
we define the integer valued quantitative objective $\intPsi^{\td}()$
as follows for  a play $\rho$
in the timed simulation game $\simgame_t(\en{\A_{\myr}},\en{\A_{\mys}})$:
\[
\intPsi^{\td} (\rho) =
 \begin{cases}
0  & \text{if } \rho(\myr) \notin \td(\en{\A_{\myr}})\\
\intPsi(\rho(\myr), \rho(\mys)) & \text{otherwise}
\end{cases}
\qed
\]
\end{definition}
The integer-time  quantitative simulation functions $\simfunc_{\intPsi^{\td}}(\tuple{s_{\myr}, s_{\mys}})$,
can now  be defined exactly as in Definition~\ref{def:qsim}, using
$\intPsi^{\td}$ instead of $\Psi^{\td}$. 
The formal definition is given below in Definition~\ref{def:intqsim}.

\begin{definition}[Integer-Time  \qtsf{s}]
\label{def:intqsim}
Let $\A_{\myr},\A_{\mys}$ be timed automata, with the corresponding enlarged
timed transition systems $\en{\A_{\myr}}, \en{\A_{\mys}}$ respectively, and let
$\simgame_t(\en{A_{\myr}},\en{A_{\mys}})$ be the two player turn-based bipartite timed 
simulation  game structure.
The value of the integer-time 
\qtsf
$\simfunc_{\intPsi^{\td}}(\tuple{\en{s_{\myr}}, \en{s_{\mys}}})$, for
$\en{s_{\myr}}$ and $\en{s_{\mys}}$ states of $\en{\A_{\myr}}$ and $\en{\A_{\mys}}$ respectively, and
for $\intPsi \in 
\set{\intdist_{\maxdiff}, \intdist_{\limmaxdiff}, \intdist_{\limavg}}$,
 is defined as follows.
\[
\simfunc_{\intPsi^{\td}}(\tuple{\en{s_{\myr}}, \en{s_{\mys}}})=  
\inf_{\pi_{\mys} \in \Pi_{\mys}} \sup_{\pi_{\myr}\in \Pi_{\myr}} 
 \intPsi^{\td}\left(\rho\left(\pi_{\myr},\pi_{\mys}, \tuple{\en{s_{\myr}}, \en{s_{\mys}},2} \right)\right)
\]
where
$\rho\left(\pi_{\myr},\pi_{\mys}, \tuple{\en{s_{\myr}}, \en{s_{\mys}},2} \right)$ is the trajectory which results
given  the player-1 strategy $\pi_{\mys} \in \Pi_{\mys}$ and the player-2 strategy $\pi_{\myr}\in \Pi_{\myr}$.
\qed
\end{definition}

Let $\rho$ be a play of the simulation game $\simgame_t(\en{\A_{\myr}},\en{\A_{\mys}})$.
The next lemma states that closeness of the trajectories
$\rho(\myr), \rho(\mys)$ according to integer trajectory distances is approximately
the same as according to the normal (real-valued) trajectory distances,
\begin{lemma}
\label{lemma:IntTrajTimeDiv}
Let $\A_{\myr},\A_{\mys}$ be timed automata, with the corresponding enlarged
timed transition systems $\en{\A_{\myr}}, \en{\A_{\mys}}$ respectively, and let
$\simgame_t(\en{\A_{\myr}},\en{\A_{\mys}})$ be the two player turn-based bipartite timed 
simulation  game structure.
The following assertions are   true for
$\tuple{\intPsi, \Psi} \in 
\set{\tuple{\intdist_{\maxdiff},\dist_{\maxdiff}}, 
\tuple{\intdist_{\limmaxdiff}, \dist_{\limmaxdiff}}, 
\tuple{   \intdist_{\limavg},  \dist_{\limavg}}}$, for any
play
$\rho $ of $\simgame_t(\en{A_{\myr}},\en{A_{\mys}})$.
\begin{compactenum}
\item $\intPsi^{\td}(\rho) = \infty $
iff
$\Psi^{\td}(\rho) = \infty $.
\item If both $\intPsi^{\td}(\rho) $ and
$\Psi^{\td}(\rho) $ are less than $\infty$,  then, 
$
\left \arrowvert \intPsi^{\td}(\rho) - \Psi^{\td}(\rho) \right\arrowvert
\, \leq \, 1
$
\end{compactenum}
\end{lemma}
\begin{proof}
The results follow from Lemma~\ref{lemma:IntTraj}
and by the definitions of $\intPsi^{\td}(\rho) $ and
$\Psi^{\td}(\rho) $.
\qedhere
\end{proof}

We next present a result concerning number sets. 
This result will be used to show that the integer simulation functions are
close in value to the real-valued simulation functions.

\begin{lemma}
\label{lemma:InfSupEqualNum}
Let $\set{\tuple{x_{r, s}, y_{r,s}} \mid r\in R,\,  s\in S 
}$
be a set of tuples of numbers for some give sets $R,S$
-such that 
$x_{r, s}\in \reals^+_{\infty}$ and 
$ y_{r, s}\in \reals^+_{\infty}
$
where
$  \reals^+_{\infty}= \reals^+\cup \set{\infty}$.
Let both the following conditions hold:
\begin{compactenum}
\item For all $r,s$ we have $x_{r, s}= \infty$ iff
$y_{r, s}= \infty$.
\item There exists some $\alpha \in \reals^+$ such that
for all $r, s$, if   $x_{r, s}$ and $y_{r, s}$ are both finite, then
$|x_{r, s}- y_{r, s} | \leq \alpha$.
\end{compactenum}
Then, the following assertion are true.
\begin{compactenum}
\item 
$\inf_{s\in S}\sup_{r\in R} x_{r, s} \, = \, \infty$ 
iff
$\inf_{s\in S}\sup_{r\in R} y_{r, s} \, = \, \infty$.
\item 
If $\inf_{s\in S}\sup_{r\in R} x_{r, s} \, < \, \infty$ and 
$\inf_{s\in S}\sup_{r\in R} y_{r, s} \, < \, \infty$  then
\[
\left\arrowvert
\inf_{s\in S}\sup_{r\in R}  x_{r, s}\ -\ 
 \inf_{s\in S}\sup_{r\in R}  y_{r, s}
\right\arrowvert
\leq \alpha
\]
\end{compactenum}
\end{lemma}
\begin{proof}
  We prove both the assertions.
  \begin{compactenum}
  \item Suppose $\inf_{s\in S}\sup_{r\in R} x_{r, s} =  \infty$ (the other
    direction is symmetric).
    We must have that for every $s\in S$ the
    entity $\sup_{r\in R} x_{r, s}  = \infty$ (otherwise the $\inf$ would have been smaller).
    We show that:\\
    \underline{\textsf{Fact-1:}} For every $s\in S$,
    if $\sup_{r\in R} x_{r, s}  = \infty$.
    the
    entity $\sup_{r\in R} y_{r, s}  = \infty$.

    Fix some $s\in S$.
    \begin{compactitem}
    \item 
      If there exists some $r\in R$ such that $x_{r, s}  = \infty$, then
      by the conditions of the lemma, $y_{r, s}  = \infty$.
      Thus  $\sup_{r\in R} y_{r, s}  = \infty$.
    \item Suppose for all $r\in R$ we have  $x_{r, s}  < \infty$.
      By the conditions of the lemma,  for all $r\in R$ 
      we have $|x_{r, s}- y_{r, s} | \leq \alpha$.
      Thus, if  $\sup_{r\in R} x_{r, s}  = \infty$, then $\sup_{r\in R} y_{r, s}  = \infty$.
    \end{compactitem}
    Thus \underline{\textsf{Fact-1}} is true.
    Hence, $\sup_{r\in R} y_{r, s}  = \infty$ for every $s\in S$.
    Thus, $\inf_{s\in S}\sup_{r\in R} y_{r, s} =  \infty$.
    
  \item Suppose we have both $\inf_{s\in S}\sup_{r\in R} x_{r, s}  < \infty$ and 
    $\inf_{s\in S}\sup_{r\in R} y_{r, s}  < \infty$.
        
    Fix some $s\in S$.
    \begin{compactitem}
    \item 
      Suppose $\sup_{r\in R}  x_{r, s} = \infty$.
      We have that $\sup_{r\in R}  y_{r, s} = \infty$
       by  \underline{\textsf{Fact-1}} above. 
    \item 
    Suppose $\sup_{r\in R}  x_{r, s} < \infty$
    (note that there must exist at least one such  $s$
    otherwise $\inf_{s\in S}\sup_{r\in R} x_{r, s}  = \infty$).
    Thus, for this $s$, we have that for all $r\in R$,
    the quantity $x_{r, s} < \infty$.
    By the conditions of the lemma, we have that 
    for this $s$, for all $r\in R$,
    the quantity $y_{r, s} < \infty$, and that
    $|x_{r, s}- y_{r, s} | \leq \alpha$.
    This implies that $\sup_{r\in R}  y_{r, s} < \infty$, and that
    \[|\sup_{r\in R}  x_{r, s}\ -\ 
    \sup_{r\in R}  y_{r, s}|
    \leq \alpha\]

      \end{compactitem}
   Let  $p_{\mys}= \sup_{r\in R}  x_{r, s}$, and $q_{\mys} = \sup_{r\in R}  y_{r, s}$.
   From above, we have that for all $s$, it holds that either
   \begin{compactitem}
   \item $p_{\mys}=q_{\mys} = \infty$, or
   \item 
   $|p_{\mys}-q_{\mys}| \leq \alpha$.
 \end{compactitem}
   Also, it holds that for at least one $s$, we have $p_{\mys} < \infty$.
   Thus, can throw away the $p_{\mys}$ numbers such that
   $p_{\mys} =  \infty$ in the computation of $\inf_{\mys} p_{\mys}$.
   For the rest, since $|p_{\mys}-q_{\mys}| \leq \alpha$, and since
  $p_{\mys} \geq 0 $ and $q_{\mys}\geq 0$,
   we have that $|\inf_{\mys}p_{\mys}-\inf_{\mys}q_{\mys}| \leq \alpha$.
   Thus, the second part of the assertion is true.
\qedhere
\end{compactenum}

\end{proof}

The following proposition states that the integer simulation functions closely
approximate the original \qtsf{s}.
\begin{proposition}[Integer-Time \qtsf{s} Approximate
Exact \qtsf{s}]
\label{proposition:IntSimToSim}
Let $\A_{\myr},\A_{\mys}$ be timed automata,  with the corresponding enlarged
timed transition systems $\en{\A_{\myr}}, \en{\A_{\mys}}$ respectively, and let
$\simgame_t(\en{A_{\myr}},\en{A_{\mys}})$ be the two player turn-based bipartite timed simulation game structure.
For $\tuple{\intPsi, \Psi} $ in
$\set{\tuple{\intdist_{\maxdiff},\dist_{\maxdiff}}, 
\tuple{\intdist_{\limmaxdiff}, \dist_{\limmaxdiff}}, 
\tuple{   \intdist_{\limavg},  \dist_{\limavg}}}$,
we have the following assertions to be true.
\begin{compactenum}
\item $\simfunc_{\intPsi^{\td}}(\tuple{\en{s_{\myr}}, \en{s_{\mys}}}) = \infty$
iff 
$\simfunc_{\Psi^{\td}}(\tuple{\en{s_{\myr}}, \en{s_{\mys}}})= \infty$.
\item If $\simfunc_{\intPsi^{\td}}(\tuple{\en{s_{\myr}}, \en{s_{\mys}}}) < \infty$ and
 $\simfunc_{\Psi^{\td}}(\tuple{\en{s_{\myr}}, \en{s_{\mys}}})< \infty$, then
 \[
\left\arrowvert
\simfunc_{\intPsi^{\td}}(\tuple{\en{s_{\myr}}, \en{s_{\mys}}})-
\simfunc_{\Psi^{\td}}(\tuple{\en{s_{\myr}}, \en{s_{\mys}}})
\right\arrowvert \leq 1
\]
\end{compactenum}
\end{proposition}
\begin{proof}
The proof follows from Lemma~\ref{lemma:InfSupEqualNum} and
Lemma~\ref{lemma:IntTrajTimeDiv}.
\qedhere
\end{proof}

\subsection{Reduction to Games on Finite Weighted Game Graphs}
\label{subsection:ReductionFinite}
In this section we show how to compute the values of the integer-time
\qtsf{s}
by reducing the problem  to finite state games.
First, we show that the values of the integer-time \qtsf{s}
are exactly the same on \emph{discrete time} region graphs as on timed automata.

\smallskip\noindent\textbf{The Integer Trace Difference Metrics
and Simulation Functions  on  Region
Graphs.}
We first lift the integer trace difference metrics $\intPsi^{\td}$
for
$\intPsi \in \set{ \intdist_{\maxdiff}, \intdist_{\limmaxdiff},  \intdist_{\limavg}}$
to region graphs.
Let $\reg(\en{\A})$ be the region graph corresponding to
the enlarged timed automaton  structure $\en{\A}$ as defined in
Subsection~\ref{subsection:Enlarge}.
We note that the $\ticks$ variable in $\en{\A}$  counts the elapsed integer time
boundaries crossed by the global clock $z$ since the last transition in $\A$..
Thus $\reg(\en{\A})$ can be viewed as a \emph{discrete time} transition
structure.
For a trajectory $\en{\traj}$ of $\reg(\en{\A})$ ,  we use
 Lemma~ \ref{lemma:IntTimeViaTicks} as defining
$\inttime_{\reg(\en{\traj})}[i]$ in terms of the $\ticks$ variable.
For \tts corresponding to  region graph  $\reg(\en{\A})$ we define
$\td(\reg(\en{\A})$ as the set of runs satisfying
the B\"uchi condition $\Buchi\left( \bigvee_{i=1}^{c_{\max}}\ticks=i\right)$.
By Lemma~\ref{lemma:TimeDivEncode}, this has the intended
meaning of encoding time divergence.
Let the 
Consider the (discrete) timed simulation game
$\simgame_t\left(\reg(\en{\A_{\myr}}), \reg(\en{\A_{\mys}})\right)$,
and let us use the observation function $\mu_{\simgame}$  defined as
$\mu_{\simgame}\left(\tuple{l,\kappa,\z, \myd} \right) = \mu(l)$.
 We define $\intPsi^{\td}$ for
$\intPsi \in \set{ \intdist_{\maxdiff}, \intdist_{\limmaxdiff},  \intdist_{\limavg}}$
on plays of $\simgame_t\left(\reg(\en{\A_{\myr}}), \reg(\en{\A_{\mys}})\right)$ as usual
using $\inttime$.
The next lemma states that the values of the integer simulation functions
on the region graphs are the same as that on timed automata.

%

\begin{lemma}
Let $\A_{\myr}, \A_{\mys}$ be timed automata, and let $\reg(\en{\A_{\myr}}), \reg(\en{\A_{\mys}})$
be region graphs of the corresponding enlarged timed game
structures $\en{\A_{\myr}}, \en{\A_{\mys}}$ respectively.
For any states $\en{s_{\myr}}$ of $\en{\A_{\myr}}$ and
$\en{s_{\mys}}$ of $\en{\A_{\mys}}$, we have
\[
\simfunc_{\intPsi^{\td}}^{\simgame_t(\en{\A_{\myr}}, \en{\A_{\mys}})}
\Bigl(\tuple{\en{s_{\myr}}, \en{s_{\mys}}}\Bigr)  
\ =\ 
 \simfunc_{\intPsi^{\td}}^{\simgame_t(\reg(\en{\A_{\myr}}), \reg(\en{\A_{\mys}}))}
\Bigl(\tuple{\reg(\en{s_{\myr}}), \reg(\en{s_{\mys}})}\Bigr) 
\]
where $\intPsi
\in \set{ \intdist_{\maxdiff}, \intdist_{\limmaxdiff},  \intdist_{\limavg}}$.
\end{lemma}
\begin{proof}
For any timed automata $\A$, we have that
$\reg(\en{\A})$ is a bisimulation quotient of $\en{\A}$
for the enlarged region equivalence relation (as defined in
Subsection~\ref{subsection:Enlarge}).
Thus, given any play $\rho$ of $\simgame_t(\en{\A_{\myr}}, \en{\A_{\mys}})$,
there exists a play $\rho_{\reg}$ of 
 $\simgame_t(\reg(\en{\A_{\myr}}), \reg(\en{\A_{\mys}}))$ such that
$\rho_{\reg}(\myr)$ and  $\rho_{\reg}(\mys)$ have the same
 integer time observation trace sequences  as $\rho(\myr)$ and  $\rho(\mys)$
(note that the enlarged region equivalence relation ensures that the values
of the $\ticks$ variables match at each step).
The dual fact for any play $\rho_{\reg}$ of 
 $\simgame_t(\reg(\en{\A_{\myr}}), \reg(\en{\A_{\mys}}))$  also holds due to the 
bisimulation.
Since $\intPsi$ depends only on the integer time plays of the game
structures, we have the desired result.
\qedhere
\end{proof}

\smallskip\noindent\textbf{The weighted finite untimed game graph 
$\simreggame\bigl(\reg(\ensysr),\reg(\ensyss)\bigr)$.}
Now we construct a finite weighted game graph $\simreggame(\reg(\ensysr),(\ensyss))$, on which
we can use the algorithms of Section~\ref{section:Finite}, to  compute the values of the
integer-time \qtsf{s}
for the game
$\simgame_t(\reg(\en{\A_{\myr}}), \reg(\en{\A_{\mys}}))$.
The game structure $\simreggame$ is essentially the simulation game 
 $\simgame_t$ over the region graphs,  where weights 
are assigned to transitions based on the $\tick$ values of the region states.
Formally,
$\simreggame\bigl(\reg(\ensysr),\reg(\ensyss)\bigr)$ 
(denoted $ \simreggame$ in short)
is the 
tuple $\tuple{S^{\simreggame},    \rightarrow ^{\simreggame}, 
w^{\simreggame}}$, where
\begin{compactitem}
\item $S^{\simreggame} = S^{\simreggame}_1 \cup S^{\simreggame}_2$, and
  \begin{compactitem}[$\star$]
  \item 
    The set of player-$2$ states   is  
    $S^{\simreggame}_2 = S^{\reg(\ensysr)} \times S^{\reg(\ensyss)} \times\set{2}$, 
    where
    $ S^{\reg(\ensysr)}$ is the set of states of $\reg(\ensysr)$, and $ S^{\reg(\ensyss)}$ 
    is the set of states of $\reg(\ensyss)$.
  \item The set of player-1 states is 
    $S^{\simreggame}_1 =S^{\reg(\ensysr)} \times S^{\reg(\ensyss)} \times
    \set{1}$.
  \end{compactitem}

\item   $\rightarrow^{\simreggame}$ is the set of edges where
  \begin{compactitem}[$\star$]
  \item The  player-2 transitions  are:\\
    $\tuple{\reg\left(\tuple{l_{\myr}, \widehat{\kappa}_{\myr}, \myd_{\myr}}\right), 
\reg\left( \en{s_{\mys}}\right), 2} 
    \longrightarrow  $
    $ \tuple{\reg\left(\tuple{l_{\myr}', \widehat{\kappa}'_{\myr}, \myd_{\myr}'}  \right), \reg\left( \en{s_{\mys}}\right), 1}$,
    such that
      $\reg\left(\tuple{l_{\myr}, \widehat{\kappa}_{\myr}, \myd_{\myr}}\right)
      \longrightarrow  
      \reg\left(\tuple{l_{\myr}', \widehat{\kappa}'_{\myr}, \myd_{\myr}'}  \right)$
      is a valid transition in
      $\reg(\ensysr)$.
  \item 
   The player-1 transitions are:\\
    $\tuple{\reg\left(\en{s_{\myr}}\right), 
\reg\left(\tuple{l_{\mys}, \widehat{\kappa}_{\mys}, \myd_{\mys}}\right), 
1} 
    \longrightarrow  $
    $ \tuple{\reg\left(\en{s_{\myr}}\right),
\reg\left(\tuple{l_{\mys}', \widehat{\kappa}'_{\mys}, \myd_{\mys}'}  \right), 
 2}$,
    such that
   \begin{compactenum}
    \item 
$\reg\left(\tuple{l_{\mys}, \widehat{\kappa}_{\mys}, \myd_{\mys}}\right)
      \longrightarrow  
      \reg\left(\tuple{l_{\mys}', \widehat{\kappa}'_{\mys}, \myd_{\mys}'}  \right)$
      is a valid transition in
      $\reg(\ensyss)$; and
    \item $\mu\left(\reg\left(\en{s_{\myr}}\right) \right) =
      \mu\left(\reg\left(\en{s_{\mys}'}\right) \right) $, that is, the
      observation on the (timed automaton) 
      location of $\reg(s_{\mys}')$ is the same as the observation
      on the location of  $\reg(s_{\myr})$.
    \end{compactenum}
    If there is no outgoing transition from a player-1 state according to the 
    above rules, we add a dummy transition to a sink state $s_{\sink}$ which
    we define to be such that the $\Opt$ value for player~1 is
    $\infty$ for all objectives from $s_{\sink}$.

  \end{compactitem}

\item 
The weight function $w^{\simreggame}$ is given as follows.
  \begin{compactitem}[$\star$]
  \item
    $    w^{\simreggame}(e_2) = 0$ for any edge $e_2$ originating from
    a player-2 state.
    


  \item 
    $
    w^{\simreggame}\left(
      \begin{array}{l}
        \tuple{\reg\left(\tuple{l_{\myr}, \widehat{\kappa}_{\myr}, \myd_{\myr}} \right), 
          \reg\left(\tuple{l_{\mys}, \widehat{\kappa}_{\mys}, \myd_{\mys}}
          \right), 1} \longrightarrow\\
        \qquad\qquad 
        \tuple{\reg\left(\tuple{l_{\myr}, \widehat{\kappa}_{\myr}, \myd_{\myr}} \right), 
          \reg\left(\tuple{l_{\mys}', \widehat{\kappa}'_{\mys},  \myd_{\mys}'}\right), 2}
      \end{array}
    \right)$ 
    is the value
$  \myd_{\myr}-  \myd_{\mys}'$.
  \end{compactitem}
  We note that     $\myd_{\mys}'$. is the number
  of  integer boundaries crossed by the clock $z$ in a  transition to go
  from any state in $\reg\left(\tuple{l_{\mys}, \widehat{\kappa}_{\mys}, \myd_{\mys}} \right)$
  to any state in $\reg\left(\tuple{l_{\mys}', \widehat{\kappa}'_{\mys}, \myd_{\mys}'} \right)$, and
  similarly for  $\myd_{\myr}$
  Thus, the quantity $  \myd_{\myr}-  \myd_{\mys}'$ encodes the difference of the 
integer boundaries crossed by the clock $z$ in  the region graphs $\reg(\ensysr)$
and $\reg(\ensyss)$ during the last step in the simulation game.

\smallskip\noindent\emph{Intuitive explanation of $\simreggame$:}
The simulation game $\simgame_t$ can be viewed as proceeding in a sequence of
rounds -- in each round first player~2 picks a transition in $\A_{\myr}$, and then 
player~1 picks a transition in $\A_{\mys}$, trying to the match the move of player~2.
The weighted game $\simreggame$ can similarly be viewed as proceeding in a sequence
of rounds.
First player~2 takes a transition from
a state $\tuple{\reg\left(\tuple{l_{\myr}, \widehat{\kappa}_{\myr}, \myd_{\myr}}\right), 
\reg(\en{s_{\mys}}), 2} $ to
$\tuple{\reg\left(\tuple{l_{\myr}', \widehat{\kappa}'_{\myr}, \myd_{\myr}'}\right),
 \reg(s_{\mys}), 1}$ corresponding to
the transition in the timed automation $\A_{\myr}$. 
The integer boundaries crossed by the global clock are recorded in the
update $\myd_{\myr}'$ (but the weight of the transition is taken as $0$).
Denoting  $\tuple{l_{\myr}', \widehat{\kappa}'_{\myr}, \myd_{\myr}'}$ as $\en{s'_{\myr}}$, and letting
$\en{s_{\mys}} = \tuple{l_{\mys}, \widehat{\kappa}_{\mys}, \myd_{\mys}}$, the next transition 
is from 
$\tuple{\reg\left(\en{s'_{\myr}}\right), 
\reg\left(\tuple{l_{\mys}, \widehat{\kappa}_{\mys}, \myd_{\mys}}\right), 1}$ to
 $ \tuple{\reg\left(\en{s'_{\myr}}\right),
\reg\left(\tuple{l_{\mys}', \widehat{\kappa}'_{\mys}, \myd_{\mys}'}  \right), 
 2}$,
corresponding to a player-1 transition  in the timed automation $\A_{\mys}$ in the
simulation game $\simgame_t$.
The duration of the player-1 transition in $\A_{\mys}$ corresponds to $ \myd_{\mys}'$
integer boundaries being crossed by the clock $z$ of $\A_{\mys}$ 
Thus, the difference in the integer boundaries crossed in the trajectories of
$\A_{\myr}$ and $\A_{\mys}$ for  this round is $ \myd_{\mys}' -  \myd_{\myr}'$,
and this is the weight of the second transition of $\simreggame$.


\end{compactitem}



The next lemma states that to compute the values of the integer-time
\qtsf{s}
on the region graphs, we can  use the 
objectives $\maxDiffLC, \EvMaxDiffLC, \AvDiffLC$ on the weighted
finite game $\simreggame\bigl(\reg(\ensysr),\reg(\ensyss)\bigr)$.

\begin{lemma}
\label{lemma:SimToGame}
Let $\A_{\myr}$ and $\A_{\mys}$ be well-formed timed automata such that
all clocks are bounded by $c_{\max}$, and let  
$\simreggame\bigl(\reg(\ensysr),\reg(\ensyss)\bigr)$ be the 
weighted game structure 
corresponding to  $\simgame_t(\reg(\en{\A_{\myr}}), \reg(\en{\A_{\mys}}))$,
 as described above.
Fix the coB\"uchi  objective $ \coBuchi(\ticks_{\myr}=0)$ in the following.
For
$\tuple{\intPsi, \Xi} $ equal to
$\tuple{\intdist_{\maxdiff},  \maxDiffLC}$, or
$\tuple{\intdist_{\limmaxdiff}, \EvMaxDiffLC}$, or
$\tuple{\intdist_{\limavg}, \AvDiffLC}$, we have
\[ 
\simfunc_{\intPsi^{\td}}^{\simgame_t(\reg(\en{\A_{\myr}}), \reg(\en{\A_{\mys}}))}
\Bigl(\tuple{\reg(\en{s_{\myr}}), \reg(\en{s_{\mys}})}\Bigr) 
=
\Bigl(\Opt^{\simreggame\bigl(\reg(\ensysr),\reg(\ensyss)\bigr)}
\left(\Xi\right)\Bigr)
\Bigl(\tuple{\reg(\en{s_{\myr}}), \reg(\en{s_{\mys}}),2}\Bigr) 
\]
\end{lemma}
\begin{proof}
Note that every finite play $\rho^{\simgame_t}$ 
of $\simgame_t\left(\reg(\en{\A_{\myr}}), \reg(\en{\A_{\mys}})\right)$  in which player~1
has not lost (in the simulation game) 
corresponds to a finite play $\rho^{\simreggame}$ 
in $\simreggame$ in which the sink location $s_{\sink}$
has not been visited, and similarly for the other direction 
(for starting states $\tuple{\reg(\en{s_{\myr}}), \reg(\en{s_{\mys}}),2}$).
The move choices for both players are the same, apart from $s_{\sink}$
transitions.

Observe that  any two states $\reg(\en{s_{\myr}}), \reg(\en{s_{\mys}}$ are not untimed
similar in $\simgame_t\left(\reg(\en{\A_{\myr}}), \reg(\en{\A_{\mys}})\right)$ iff
in the game $\simreggame$,
for every player-1 strategy, player~2 has a strategy which forces the play
into the sink location and thus leads to an $\infty$ value for all the quantitative 
objectives.
Thus, consider the case where 
$\reg(\en{s_{\myr}}), \reg(\en{s_{\mys}}$ are untimed
similar.
Now, $\td$ has been shown to be equivalent to 
$\Buchi\left( \bigvee_{i=1}^{c_{\max}}\ticks=i\right)$ 
 earlier, on the region graphs.
This B\"uchi condition is equivalent to
$\neg \coBuchi(\ticks=0)$.
Thus,
the condition $\rho^{\simgame_t}(\myr)\notin \td$  holds
iff    $ \rho^{\simreggame} \in \coBuchi(\ticks_{\myr}=0)$ holds.
Finally, we note that for any  play
$\rho^{\simgame_t}
\bigl(\pi_{\myr},\pi_{\mys}, \tuple{\reg(\en{s_{\myr}}), \reg(\en{s_{\mys}}),2}\bigr)$,
the corresponding play
$\rho^{\simreggame}
\bigl(\pi_{\myr},\pi_{\mys}, \tuple{\reg(\en{s_{\myr}}), \reg(\en{s_{\mys}}),2}\bigr)$ is such that
\begin{compactenum}
\item 
Forall  $i>0$,
we have
$\inttime_{\left(\rho^{\simgame_t}\right)(\myr)}[i] - 
\inttime_{\left(\rho^{\simgame_t}\right)(\mys)}[i] \ =
\ 
\sum_{j=1}^{i}
w^{\simreggame}
\Bigl(\rho^{\simreggame}[2j-1] \longrightarrow \rho^{\simreggame}[2j] \Bigr)$.
\item For every $i\geq 0$, we have
$w^{\simreggame}
\Bigl(\rho^{\simreggame}[2i] \longrightarrow \rho^{\simreggame}[2i+1] \Bigr)
= 0$
\end{compactenum}
The desired results follow.
\qedhere
\end{proof}


\subsection{Integer-Time  Simulation Functions Approximate
Real-Valued Simulation Functions}

\smallskip\noindent\textbf{Precision of the Integer-Time \qtsf{s}.}
Given a positive integer $\alpha \geq 1$, and a timed automaton $\A$, let $\alpha\cdot \A$
denote the timed automaton obtained from $\A$ by multiplying every constant by $\alpha$.
Note that if clocks are bounded by $c_{\max}$ in $\A$, then clocks are bounded by
$\alpha\cdot c_{\max}$ in $\alpha\cdot \A$.
The automaton $\alpha\cdot \A$ is just $\A$ with a blown up timescale.
One time unit in $\A$ corresponds to $\alpha$ time units in $\alpha\cdot \A$.
We let $\alpha\cdot \symb{\A} = \symb{\alpha\cdot \A}$,
and 
$\alpha\cdot \tuple{l, \kappa, \z, \myd} =  
\tuple{l,  \alpha\cdot \kappa, \fractional(\alpha\cdot \z), 
\floor{\alpha\cdot \z}+ \alpha\cdot\myd}$
where $\fractional(\beta)$ denotes the fractional part of $\beta$, \emph{i.e.}
$\beta-\floor{\beta}$ for $\beta\geq 0$.
Note that in $\alpha\cdot \symb{\A}$, the clock $z$ \emph{still cycles from 0 to 1}.
Thus, we first blow up the time scale of $\A$ to obtain $\alpha\cdot \A$, 
and \emph{then} take the expanded
game structure $\symb{\alpha\cdot \A}$.

\begin{lemma}
\label{lemma:Zoom}
Let $\A_{\myr},\A_{\mys}$ be timed automata, with the corresponding enlarged
timed transition systems $\en{\A_{\myr}}, \en{\A_{\mys}}$ respectively, and let
$\simgame_t(\en{A_{\myr}},\en{A_{\mys}})$ be the two player turn-based bipartite timed 
simulation  game structure.
For $\Psi \in
 \set{\dist_{\maxdiff}, \dist_{\limmaxdiff}, \dist_{\limavg}}$, for any $\alpha$ a
positive integer, and for
any states $\en{s_{\myr}}$ and $\en{s_{\mys}}$  of $\ensysr$ and $\ensyss$ respectively,
we have
\[
\alpha\cdot \simfunc_{\Psi^{\td}}^{\ensysr, \ensyss}(\tuple{\en{s_{\myr}}, \en{s_{\mys}}})=  
\simfunc_{\Psi^{\td}}^{\alpha\cdot \ensysr, \alpha\cdot \ensyss}
(\tuple{\alpha\cdot \en{s_{\myr}}, \alpha\cdot \en{s_{\mys}}})
\]

\end{lemma}
\begin{proof}
The proof follows from observing that the times in $\alpha\cdot \A$ are just
the times in $\A$ multiplied by $\alpha$.
\qedhere
\end{proof}

The following lemma states that integer-time \qtsf{s} 
can
approximate the exact
\qtsf{s}
to within any desired
degree of accuracy.

\begin{proposition}[Integer-Time \qtsf{s}
Approximate
Exact
\qtsf{s}
to Any Desired Degree]
\label{proposition:IntSimToSimDegree}
Let $\A_{\myr},\A_{\mys}$ be timed automata,  with the corresponding enlarged
timed transition systems $\en{\A_{\myr}}, \en{\A_{\mys}}$ respectively, and let
$\simgame_t(\en{A_{\myr}},\en{A_{\mys}})$ be the two player turn-based bipartite timed simulation game structure.
For $\tuple{\intPsi, \Psi} $ in
$\set{\tuple{\intdist_{\maxdiff},\dist_{\maxdiff}}, 
\tuple{\intdist_{\limmaxdiff}, \dist_{\limmaxdiff}}, 
\tuple{   \intdist_{\limavg},  \dist_{\limavg}}}$,
and for any positive integer $\alpha > 0$,
we have the following assertions to be true.
\begin{compactenum}
\item $\simfunc_{\intPsi^{\td}}(\tuple{\alpha\cdot \en{s_{\myr}}, \alpha\cdot 
\en{s_{\mys}}}) = \infty$
iff  
$\simfunc_{\Psi^{\td}}(\tuple{\en{s_{\myr}}, \en{s_{\mys}}})= \infty$.
\item If $\simfunc_{\intPsi^{\td}}(\tuple{\alpha\cdot\en{s_{\myr}}, \alpha\cdot\en{s_{\mys}}}) < \infty$ and
 $\simfunc_{\Psi^{\td}}(\tuple{\en{s_{\myr}}, \en{s_{\mys}}})< \infty$, then
 \[\mspace{-30mu}
\left\arrowvert
\alpha^{-1}\cdot \simfunc_{\intPsi^{\td}}(\tuple{\alpha\cdot\en{s_{\myr}}, \alpha\cdot\en{s_{\mys}}})-
\simfunc_{\Psi^{\td}}(\tuple{\en{s_{\myr}}, \en{s_{\mys}}})
\right\arrowvert \leq \frac{1}{\alpha}
\]
\end{compactenum}
\end{proposition}
\begin{proof}
The proof follows from 
Lemma~\ref{lemma:Zoom} and Proposition~\ref{proposition:IntSimToSim} applied
to $\alpha\cdot \A_{\myr}$ and $\alpha\cdot \A_{\mys}$.
\qedhere
\end{proof}

\subsection{Final Algorithms and Results}
We now present the final algorithm for computing
the values for the \qtsf{s}
$\simfunc_{\Psi^{\td}}$ for 
$\Psi \in
\set{\dist_{\maxdiff}, \dist_{\limmaxdiff}, \dist_{\limavg}}$,
 to within
any desired degree of accuracy.
The algorithm is listed in  the function $h_{\Psi, \alpha}(s_{\myr}, s_{\mys})$.
The proof of the correctness of the algorithm uses
 Proposition~\ref{proposition:IntSimToSimDegree}, and
Lemma~\ref{lemma:SimToGame},
and the results of the previous section on games on finite state game graphs.

\begin{theorem}
Let $\A_{\myr}$ and $\A_{\mys}$ be well-formed timed automata such that
all clocks are bounded by $c_{\max}$, and let  $\alpha \geq 1$ be a positive integer.
Let
$\simfunc_{\Psi^{\td}}$ denote the \qtsf
for $\Psi \in
\set{\dist_{\maxdiff}, \dist_{\limmaxdiff}, \dist_{\limavg}}$.
The function $h_{\Psi, \alpha}()$ is such that
for
 any states $s_{\myr}$ of $\A_{\myr}$ and $s_{\mys}$ of $\A_{\mys}$,  either
\begin{compactenum}
\item $\simfunc_{\Psi^{\td}}(\tuple{s_{\myr}, s_{\mys}}) = $
$h_{\Psi, \alpha}(s_{\myr}, s_{\mys}) = \infty$; or
\item Both values are finite and
$
\left\arrowvert\simfunc_{\Psi^{\td}}(\tuple{s_{\myr}, s_{\mys}}) \, -\, h_{\Psi, \alpha}(s_{\myr}, s_{\mys})
\right\arrowvert \leq \frac{1}{\alpha}
$
\end{compactenum}
\end{theorem}
\IncMargin{1em}
\begin{function}
     \SetKwInOut{Input}{Input}
    \SetKwInOut{Output}{Output}
    \Input{States $s_{\myr},s_{\mys}$ from $\A_{\myr},\A_{\mys}$ respectively;
    $\alpha$ a positive integer;\\
      $\Psi \in
      \set{\dist_{\maxdiff}, \dist_{\limmaxdiff}, \dist_{\limavg}}$}
    \Output{A number 
    approximating $\simfunc_{\Psi^{\td}}(\tuple{s_{\myr}, s_{\mys}}) $ (maximum error difference: 
  1/$\alpha$)}
    $\reg(\en{\alpha\cdot \A_{\myr}}), \reg(\en{\alpha\cdot \A_{\mys}}) :=$ 
    Region graphs of the expanded timed game structures\\
    \hspace*{55mm} $\en{\alpha\cdot \A_{\myr}}$ and $\en{\alpha\cdot \A_{\mys}}$\;
    $\simreggame: = \simreggame\bigl(\reg(\en{\alpha\cdot\A_{\myr}}),
\reg(\en{\alpha\cdot \A_{\mys}})\bigr)$ 
       \tcc*{Finite weighted game constructed from the region graphs}
    \Switch{$\Psi$}
    {
      \uCase{$\dist_{\maxdiff}$}
      {
        $\Xi :=  \maxDiffLC_{\coBuchi(\ticks_{\myr}=0)}$\;
      }
      \uCase{$ \dist_{\limmaxdiff}$}
      {
        $ \Xi := \EvMaxDiffLC_{\coBuchi(\ticks_{\myr}=0)}$\;
      }
      \uCase{$\dist_{\limavg}$}
      {
        $\Xi := \AvDiffLC_{\coBuchi(\ticks_{\myr}=0)}$\;
        }
      }
 \Return{      $\alpha^{-1}\cdot\Opt^{\simreggame}
   (\Xi)\Big(\left\langle
         \reg(\en{\alpha\cdot s_{\myr}}),\ 
         \reg(\en{\alpha\cdot s_{\mys}}),\ 
         2
          \right\rangle\Big)$
      }\;
      
%
 \caption{$h_{\Psi, \alpha}$($s_{\myr}, s_{\mys}$)}
  \label{function:hfun}
\end{function}
\DecMargin{1em}
\begin{proof}
The proof follows from
  Proposition~\ref{proposition:IntSimToSimDegree}
and Lemma~\ref{lemma:SimToGame}.
Since 
    $\simreggame\bigl(\reg(\en{\alpha\cdot\A_{\myr}}),
\reg(\en{\alpha\cdot \A_{\mys}})\bigr)$
  is a finite weighted game graph, the value of
  $h_{\Psi ,\alpha}(s_{\myr}, s_{\mys})$ can be computed using the algorithms of 
  Section~\ref{section:Finite}.
\qedhere
\end{proof}

\section{Concluding Remarks}
We have defined three ways of quantifying timing mismatches, and
have presented algorithms for computing the values of three kinds  of 
quantitative timed simulation functions
   which quantify corresponding timing mismatches between two timed automata to within 
any desired degree of accuracy.
We note that the optimal player-1 strategies in the weighted game
 $\simreggame$ used in Function $h_{\Psi, \alpha}()$ 
are also computable, and are witnesses to the
quantitative simulation function values (similar to
simulation relations witnessing the simulation decision problem).
We expect that the algorithms presented in this paper will contribute to the further
development of approximation 
theories for continuous, switched and hybrid dynamical systems for the automatic synthesis of more powerful controllers.

\renewcommand*{\bibfont}{\raggedright}
\printbibliography

\end{document}